\documentclass[aps,reprint,twocolumn,superscriptaddress]{revtex4-2}

\usepackage{comment}
\usepackage{tikz}
\usepackage{amsmath}
\usepackage{amssymb}  
\usepackage{amsthm}
\usepackage{amsfonts}
\usepackage{graphicx}
\usepackage{bbm}
\usepackage{url}
\usepackage{braket}
\usepackage{ upgreek }
\usepackage{dsfont}
\usepackage{makecell}
\usepackage[plain]{fancyref} %
\usepackage{algorithm}
\usepackage{algpseudocode}

\usepackage{float}
\newfloat{algorithm}{t}{lop}

\usepackage{hyperref}

\theoremstyle{plain}   
\newtheorem{theorem}{Theorem}[section]
\newtheorem{thm}{Theorem}[section]

\newtheorem{lem}{Lemma}[section]
\newtheorem{cor}{Corollary}[section]
\newtheorem{prop}{Proposition}[section]

\newtheorem{defi}{Definition}[section]

\theoremstyle{remark}

\newcommand*{\fancyrefthmlabelprefix}{thm}
\newcommand*{\fancyreflemlabelprefix}{lem}
\newcommand*{\fancyrefcorlabelprefix}{cor}
\newcommand*{\fancyrefdefilabelprefix}{defi}
\frefformat{plain}{\fancyreflemlabelprefix}{lemma\fancyrefdefaultspacing#1}
\Frefformat{plain}{\fancyreflemlabelprefix}{Lemma\fancyrefdefaultspacing#1}
\frefformat{plain}{\fancyrefthmlabelprefix}{theorem\fancyrefdefaultspacing#1}
\Frefformat{plain}{\fancyrefthmlabelprefix}{Theorem\fancyrefdefaultspacing#1}
\frefformat{plain}{\fancyrefcorlabelprefix}{corollary\fancyrefdefaultspacing#1}
\Frefformat{plain}{\fancyrefcorlabelprefix}{Corollary\fancyrefdefaultspacing#1}
\frefformat{plain}{\fancyrefdefilabelprefix}{definition\fancyrefdefaultspacing#1}
\Frefformat{plain}{\fancyrefdefilabelprefix}{Definition\fancyrefdefaultspacing#1}
\newcommand*{\fancyrefalglabelprefix}{alg}
\newcommand*{\frefalgname}{algorithm}
\newcommand*{\Frefalgname}{Algorithm}
\frefformat{plain}{\fancyrefalglabelprefix}{%
  \frefalgname\fancyrefdefaultspacing#1%
}%
\Frefformat{plain}{\fancyrefalglabelprefix}{%
  \Frefalgname\fancyrefdefaultspacing#1%
}%

\newcommand*{\fancyrefapplabelprefix}{app}
\newcommand*{\Frefappname}{Appendix}
\frefformat{plain}{\fancyrefapplabelprefix}{%
  \frefappname\fancyrefdefaultspacing#1%
}%
\Frefformat{plain}{\fancyrefapplabelprefix}{%
  \Frefappname\fancyrefdefaultspacing#1%
}%

\definecolor{Green}{HTML}{00AD69}  %

\def\beq{\begin{equation}}
\def\eeq{\end{equation}}
\def\bq{\begin{quote}}
\def\eq{\end{quote}}
\def\ben{\begin{enumerate}}
\def\een{\end{enumerate}}
\def\bit{\begin{itemize}}
\def\eit{\end{itemize}}

\def\l|{\left|}
\def\r|{\right|}

\newcommand\R{\mathbbm{R}}
\newcommand\N{\mathbbm{N}}
\newcommand\M{\mathcal{M}}
\newcommand\D{\mathcal{D}}

\newcommand{\cL}{\mathcal{L}}

\newcommand{\avg}[1]{\langle #1 \rangle} %

\newcommand{\linGen}{\cG}

\newcommand{\tr}[1]{\operatorname{tr}\left[#1\right]}

\newcommand{\id}{\text{id}}

\definecolor{mulberry}{rgb}{0.77, 0.29, 0.55}

\newcommand{\cG}{\mathcal{G}}

\newcommand{\norm}[1]{\left\|#1\right\|}

\newcommand{\cO}{\mathcal{O}}
\newcommand{\tcO}{\tilde{\mathcal{O}}}

\newcommand{\Tr}{\operatorname{tr}}

\newcommand{\polylog}{\textrm{polylog}}

\begin{document}
\title{Efficient and robust estimation of many-qubit Hamiltonians}
\author{Daniel Stilck Fran\c{c}a}
 \email{daniel.stilck\_franca@ens-lyon.fr}
 \affiliation{QMATH, Department of Mathematical Sciences, University of Copenhagen, Universitetsparken 5, 2100 Copenhagen, Denmark}
 \affiliation{Univ Lyon, ENS Lyon, UCBL, CNRS, Inria, LIP, F-69342, Lyon Cedex 07, France}

 \author{Liubov A. Markovich}
\affiliation{QuTech and Kavli Institute of Nanoscience, Delft University of Technology, 2628 CJ, Delft, The Netherlands}

\author{V. V. Dobrovitski}
\affiliation{QuTech and Kavli Institute of Nanoscience, Delft University of Technology, 2628 CJ, Delft, The Netherlands}

\author{Albert H. Werner}
 \affiliation{QMATH, Department of Mathematical Sciences, University of Copenhagen, Universitetsparken 5, 2100 Copenhagen, Denmark}
\affiliation{NBIA, Niels Bohr Institute, University of Copenhagen, Blegdamsvej 17, 2100 Copenhagen, Denmark}

\author{Johannes Borregaard}
 \affiliation{QuTech and Kavli Institute of Nanoscience, Delft University of Technology, 2628 CJ, Delft, The Netherlands}%

\begin{abstract}
Characterizing the interactions and dynamics of quantum mechanical systems is an essential task in the development of quantum technologies. We propose an efficient protocol based on the estimation of the time-derivatives of few qubit observables using polynomial interpolation for characterizing the underlying Hamiltonian dynamics and Markovian noise of a multi-qubit device. For finite range dynamics, our protocol exponentially relaxes the necessary time-resolution of the measurements and quadratically reduces the overall sample complexity compared to previous approaches. Furthermore, we show that our protocol can characterize the  dynamics of systems with algebraically decaying interactions. The implementation of the protocol requires only the preparation of product states and single-qubit measurements. Furthermore, we develop a shadow tomography method for quantum channels that is of independent interest. This protocol can be used to parallelize  to learn the Hamiltonian, rendering it applicable for the characterization of both current and future quantum devices.
\end{abstract}

\maketitle

\section{Introduction}
Large quantum devices consisting of tens to hundreds of qubits have been realized across various hardware architectures~\cite{Arute2019,Zhong2020,Scholl2021,Ebadi2021} representing a significant step towards the realization of quantum computers and simulators with the potential to solve outstanding problems intractable for classical computers~\cite{Cirac2012,MichaelA.Nielsen2010}. 
However, continued progress towards this goal requires careful characterization of the underlying Hamiltonians and dissipative dynamics of the hardware to mitigate errors and engineer the desired dynamics. The exponential growth of the dimension of the state space of a quantum device with the number of qubits renders this an outstanding challenge broadly referred to as the Hamiltonian learning problem~\cite{Eisert2020,Shulman2014,Zhang2014,Zhang_2015,PhysRevLett.102.187203,Sheldon_2016,Sone_2017,lindblad_tomo,Wang_2015,Wang2017,Valenti2019,Valenti2021,Granade2012,PhysRevLett.112.190501,PhysRevA.89.042314,PhysRevA.84.012107,Qi2019,Chertkov_2018,PhysRevLett.122.020504,Anshu2021,Evans2019,Li_2020,2108.04842,PhysRevLett.107.210404,Bairey2020,Zubida2021,Gu2022,Yu2022,Rattacaso2022}.   

To tackle this challenge, previous approaches make strong assumptions such as the existence of a trusted quantum simulator capable of simulating the unknown Hamiltonian~\cite{PhysRevLett.112.190501,PhysRevA.89.042314} or the capability of preparing particular states of the Hamiltonian such as steady states and Gibbs states~\cite{PhysRevA.92.052322,PhysRevLett.122.020504,Anshu2021,2108.04842,Qi2019,2107.03333}, which may be difficult for realistic devices subject to various decoherence mechanisms. 

Alternatively, several works~\cite{PhysRevLett.107.210404,Bairey2020,Zubida2021} are built on the observation that a Master equation describes the evolution of any system governed by Markovian dynamics. 
Through this, one obtains a simple linear relation between time derivatives of expectation values and the parameters of the Hamiltonian, jump operators and decay rates (jointly referred to as the parameters of the Lindbladian $\cL$) governing the system. Furthermore, for finite range interactions, these approaches can estimate the parameters of the Lindbladian to a given precision from a number of samples that is independent of the system's size~\cite{PhysRevLett.107.210404,Bairey2020,Zubida2021}. %
\begin{center}
    \begin{figure}[t]
    \includegraphics[width=1.\linewidth]{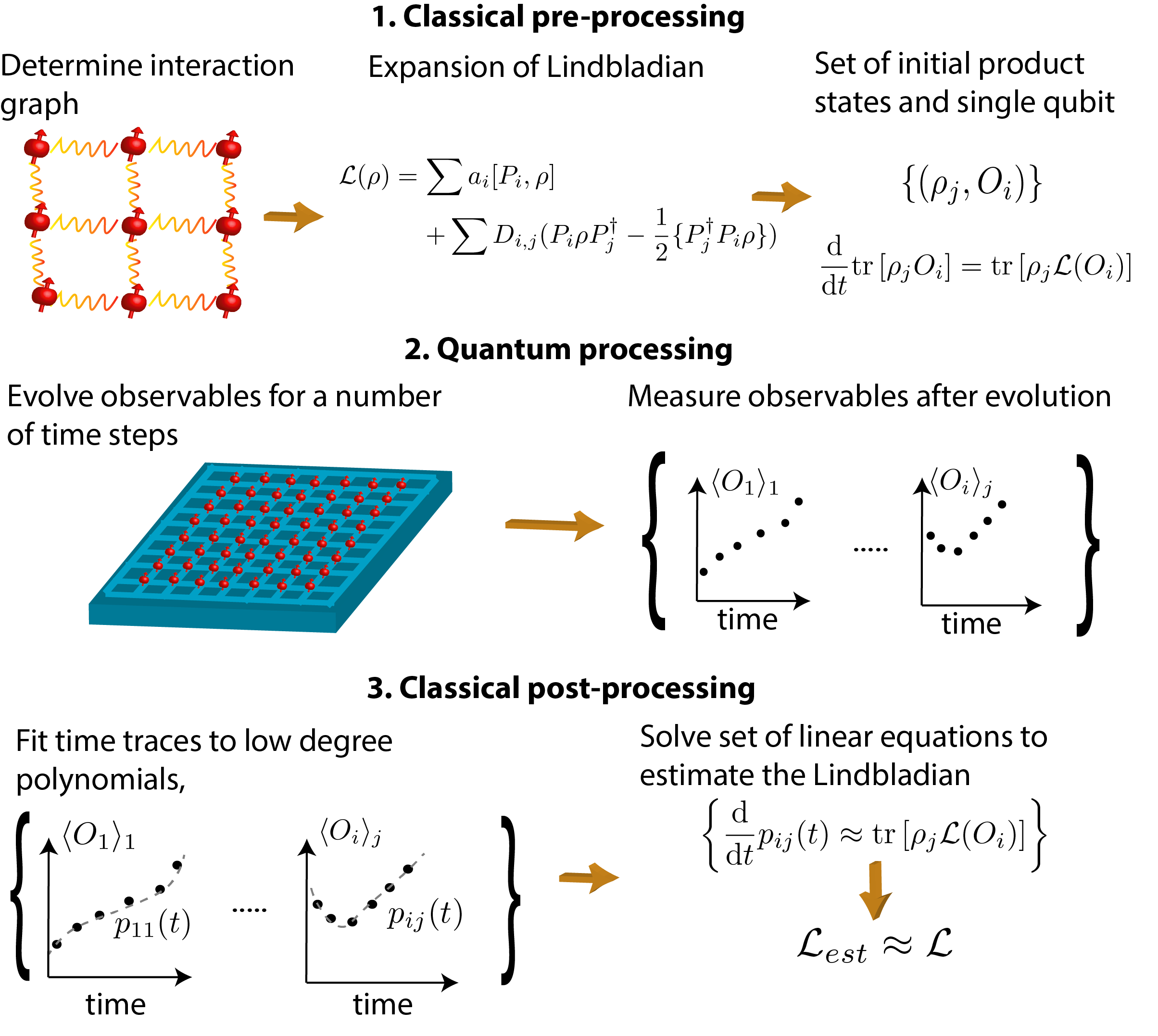}
    \label{fig:precision_0}
    \caption{Sketch of the proposed protocol to estimate an unknown Lindbladian, $\cL$, of a multi-qubit device. In the first step of classical pre-processing, the interaction graph between qubits is identified from the physical connectivity of the device. Then the unknown Lindbladian is written in a general form using an operator basis of Pauli strings, $\{P_i\}$ (see main text) and a suitable set of initial states and observables, $\{(\rho_j,O_i)\}$ is chosen. In the second step of quantum processing, a time trace of each element of the set is obtained from evolution and measurement on the quantum device. In the last step of classical post-processing, each time trace is fitted to a low-degree polynomial to estimate the derivative of the observable. From these, an estimate of the Lindbladian,$\cL_{est}$, is obtained from the Master equation.}
    \end{figure}
\end{center}

A significant drawback of these approaches is that the time derivatives are estimated using finite difference methods. Obtaining a good precision thus requires high time resolution, which is experimentally challenging given the finite operation time of gates and measurements. To estimate a Lindbladian parameter up to an additive error $\epsilon$, the system has to be probed at times $\cO(\epsilon)$ apart and expectation values of observables have to be estimated up to a precision of $\cO(\epsilon^2)$, which translates to an overall $\cO(\epsilon^{-4})$ sample complexity to estimate each parameter.

In this article, we propose a novel protocol that alleviates these daunting experimental requirements. Our protocol requires only a time resolution of $\cO(\polylog(\epsilon^{-1}))$ representing an exponential improvement compared to previous protocols and gives an overall sample complexity to recover \emph{all} parameters of a local $n$ qubit Lindbladian up to precision $\epsilon$ of $\cO(\epsilon^{-2}\textrm{polylog}(n,\epsilon^{-1}))$. We obtain this by estimating time derivatives using multiple temporal sampling points and robust polynomial interpolation \cite{kane_robust_2017}. Furthermore, we show how to use shadow process tomography methods to estimate multiple parameters in parallel. In particular, we address shortcomings of previous results~\cite{processtomo,processtomo2} in extending the framework of classical shadows to processes, a result that is of independent interest.
We also extend our analysis to long-range (algebraically decaying) interactions in the systems, obtaining the first results for such systems to the best of our knowledge. The necessary operations for our protocol are measurements in Pauli basis on time-evolved product states consisting of Pauli eigenstates. These minimal requirements make our protocol feasible for characterization of both current and future quantum devices.  

\section{Results}
In order to use our protocol for an efficient characterization of a quantum device, two assumptions should be fullfilled: 
\begin{enumerate}
    \item The quantum device implements an (unknown) Markovian quantum evolution on $n$ qubits described by a time-independent Lindbladian, $\cL$. 
    \item We assume knowledge of the general structure of the interaction graph of the device i.e. which qubits are coupled to each other. Importantly, no assumptions are made regarding the exact strength or even form of the couplings.  
\end{enumerate}

The first assumption ensures that the evolution of a general observable, $\avg{O}$ is described by the Master equation, i.e. $\frac{d}{dt}\avg{O}=\avg{\cL(O)}.$ We note that the Lindbladian captures both the Hamiltonian evolution and the dissipative dynamics of the device. Furthermore, the assumption of time-independence applies to the run-time of the experimental characterization.  

The second assumption bounds the size of the estimation task. If the interaction graph was completely unknown, our protocol could, in principle, be applied but would require the estimation of an exponentially growing number of general multi-qubit coupling terms as the number of qubits increases. However, having prior knowledge that, e.g. nearest neighbor couplings dominate in the device, makes the estimation task tractable.   

Using the knowledge of the interaction graph, one can expand the Lindbladian in an operator basis, $\{\sigma_i\}$ constructed from tensor products of single-qubit Pauli matrices and the identity. Such an expansion is always possible since this basis amounts to a Hilbert-Schmidt orthogonal set of traceless Hermitian operators spanning the entire vector space. Estimating the set of expansion coefficients $\{\alpha_i\}$ gives an estimation of $\cL$ and thus a full characterization of the system. 

It is well known that the Master equation for the time derivative of the expectation value of a local observable $O$ at time $t=0$ for a given initial state $\rho$ of the system gives us a linear equation for the expansion coefficients~\cite{PhysRevLett.107.210404,Bairey2020,Zubida2021}. We use this to estimate the expansion coefficients going through three stages of \emph{classical pre-processing}, \emph{quantum processing}, and \emph{classical post-processing}.  \newline \newline
\textbf{Classical pre-processing} 

After expanding $\cL$ in an operator basis, the following steps are performed. 
\begin{enumerate}
    \item Find a suitable complete set, $\{(\rho_j,O_i)\}$ of multi-qubit product states $(\rho_j)$ and observables $(O_i)$ for which the Master equation involves only a few selected expansion parameters of the Lindbladian for each element of the set. The set is complete in the sense that all expansion coefficients can be found by solving the Master equations for all elements in the set. As we show below, such a set can readily be found by considering initial states where only a few qubits are initialized as different eigenstates of the Pauli matrices while the remaining qubits are prepared in the maximally mixed state $I/2$.  
    \item Calculate the expectation values appearing on the right hand side of the Master equations $\frac{\text{d}}{\text{d}t}\!\tr{\rho_j O_i}=\tr{\rho_j \cL(O_i)}$ for all elements in the set $\{(\rho_j,O_i)\}$. Since both the initial states and the observables are products, this can be done efficiently. 
    \end{enumerate}  
\textbf{Quantum processing}

In order to solve for the expansion coefficients $\{\alpha_i\}$, we also need the values of the time-derivatives appearing on the left hand side of the Master equations, i.e. $\frac{\text{d}}{\text{d}t}\!\tr{\rho_j O_i}$. These are estimated using the quantum device in the following way.   
    \begin{enumerate}
    \item The quantum device is prepared in initial state $\rho_j$ and evolved for a time $t_k\in\{t_0,t_1,\ldots, T\}$ after which the observable $O_i$ is measured.  
    \item The above procedure is repeated for each element in the set $\{(\rho_j,O_i)\}$ for all evolution times $t_k$ to obtain estimates of $\avg{O_i(t_k)}_j=\tr{\rho_j(t_k) O_i}$ where $\rho_j(t_k)$ is the state of the system having evolved for time $t_k$ from the initial state $\rho_j$. We note that the single qubit mixed states can be simulated by sampling eigenstates of the Pauli matrices at random. 
\end{enumerate}
\textbf{Classical post-processing}

The final part of the characterization involves estimating $\frac{\text{d}}{\text{d}t}\!\tr{\rho_j O_i}$ from the experimentally obtained time trace of $\avg{O_i(t_k)}_j$ and solving for the expansion coefficients $\{\alpha_i\}$. This involves 
\begin{enumerate}
    \item Fit the time trace of $\avg{O_i(t_k)}_j$ with a low-degree polynomial in the time, $p_{i,j}(t)$ and estimate $\frac{\text{d}}{\text{d}t}\!\tr{\rho_j O_i}$ as $\frac{\text{d}}{\text{d}t}p_{i,j}(t)|_{t=0}$. This is done for each element in the set $\{(\rho_j,O_i)\}$. 
    \item Solve the set of linear equations from the Master equations $\frac{\text{d}}{\text{d}t}\!\tr{\rho_j O_i}=\tr{\rho_j \cL(O_i)}$ with respect to the expansion coefficients ($\{\alpha_i\}$). This is possible since  $\frac{\text{d}}{\text{d}t}\!\tr{\rho_j O_i}$ has been estimated from the polynomial fits and all expectation values appearing in $\tr{\rho_j \cL(O_i)}$ have been calculated leaving the $\alpha_i$'s as the only unknown parameters.  
\end{enumerate}

Following the steps above, a complete characterization of the underlying Hamiltonian and dissipative dynamics of the quantum device as given by the Lindbladian is obtained. The two key steps of the protocol are the choice of the set $\{(\rho_j,O_i)\}$ and the polynomial interpolation used to obtain estimates of the time derivatives. Below, we outline the details of both steps, show how shadow tomography methods can be used to parallelize the procedure, and provide rigorous guarantees on the precision of the protocol. Importantly, we show that Lieb-Robinson bounds on the spread of correlations in the system can be used to ensure robust polynomial fitting of the time traces of expectation values allowing for an exponential relaxation of the temporal resolution compared to finite difference methods rendering the protocol feasible for near-term quantum devices. 

\subsection{Choosing the set of initial states and observables} 
The first step in the classical pre-processing is to expand $\cL$ in an operator basis constructed from tensor products of single-qubit Pauli matrices and the identity. 
The right hand side (rhs) of the Master equation $\frac{\text{d}}{\text{d}t}\!\tr{\rho_j O_i}=\tr{\rho_j \cL(O_i)}$ 
can be expanded as a sum of single Pauli matrices and their products. Our goal is to isolate the unknown expansion coefficients.
To this end, we consider an initial state of the form
\begin{eqnarray}\label{1326}
\rho_{k,l}^{(i,j)}=\frac{(I+\sigma^{(i)}_{k})}{2}\otimes \frac{(I+\sigma^{(j)}_{l})}{2}\otimes \rho_{N-2},
\end{eqnarray}
where the $i$'th and $j$'th qubit are prepared in eigenstates of the Pauli matrices while the state of the remaining $N-2$ qubits, $\rho_{N-2}$, is assumed to be the maximally mixed state. 

For a state of the form in Eq.~(\ref{1326}) the rhs of the Master equation (see above) can be simplified greatly depending on the choice of the observable $O$. This is due to the properties of the Pauli matrices namely that they have vanishing trace and that 
\begin{equation} \label{1327}
\sigma_{k} \sigma_{l}=\delta_{kl}I+i\varepsilon_{klm}\sigma_{m}, 
\end{equation}
where $\delta_{kl}$ is the Kroenecker delta function and $\varepsilon_{klm}$ is the Levi-Civita symbol. From this, it follows that if a single qubit Pauli observable ($O=\sigma^{(i)}_{m}$) is chosen, then only the single qubit terms of the rhs of the Master equation involving the $i$'th qubit will have non-vanishing trace and, using the relation in Eq.~(\ref{1327}), the different single qubit Pauli expansion coefficients (the coefficients of terms in the expansion that only involves single qubit Pauli matrices) can be isolated. 

After isolating the single qubit expansion coefficients, the expansion coefficients related to two-qubit Pauli terms ($\sigma^{(i)}_{m}\otimes\sigma^{(j)}_{n}$) can be isolated by choosing observables of the form $O=\sigma^{(i)}_{m}\otimes\sigma^{(j)}_{n}$ in a similar manner. This procedure can be iterated to isolate higher and higher order expansion coefficients by considering observables involving more and and more qubits.

In the supplemental material, we provide a detailed derivation of how all expansion parameters can be isolated for a general Hamiltonian with terms coupling from two to $k$ qubits and arbitrary single qubit dissipation terms. We note that already for two qubit dissipation terms, deriving linear combinations of initial states and expectation values that allow us to isolate different parameters is quite cumbersome and we do not do this explicitly. However, from a numerical point of view this is a trivial task. Indeed, as remarked before, each pair of Pauli strings gives us access to a linear equation for the different parameters of the evolution. 

After collecting enough equations to ensure that the linear system is invertible, the precision with which we need to estimate each expectation value to ensure a reliable estimation of the parameters is controlled by the condition number of the matrix describing the system of linear equations. As both estimating the condition number and solving the linear system can be done efficiently, we conclude that estimating dissipative terms acting on a constant number of qubits does not posses a significant challenge from a numerical perspective.

A desirable feature for a Hamiltonian learning protocol is that the state preparation and measurement steps are simple, parallelizable and even independent of the parameter being estimated. In this paper, we propose a variation of the classical shadows protocol of~\cite{Huang2020} for process tomography that achieves that and is of interest on its own. Given a quantum channel $\Phi$ acting on $N$ qubits, our protocol estimates overlaps of the form $2^{-n}\tr{P^a_{m}\Phi(P^b_{l})}$ for Pauli strings $P^a_m,P^b_l$. Here $P^a_m$ refers to the $m$'th Pauli string in the collection of Pauli strings that differ from the identity on at most $\omega_a$ sites. Note that estimates of such overlaps is all that is required as input for our Hamiltonian learning algorithm. By only requiring the preparation of random product Pauli eigenstates and measurements in a random Pauli bases, our protocol only requires $\mathcal{O}(3^{w_a+w_b}\epsilon^{-2}\log(K_1K_2))$ samples to estimate such overlaps up to error $\epsilon$ for $K_1K_2$ pairs of Pauli strings of weight at most $w_a$ and $w_b$ respectively. For local Hamiltonians on a lattice, using this protocol gives a logarithmic in system size sample complexity to determine \emph{all} parameters of the evolution.

\subsection{Robust polynomial interpolation} As described above, a key step in our learning algorithm is  to obtain information about the time-derivatives of observables at $t=0$. For this, we rely on robust polynomial interpolation. Accordingly, based on expectation values $\avg{O_i(t_k)}_j$ for a set of times $t_k$ we want to extract a polynomial $p_{i,j}(t)$  such that swe can estimate $\frac{\text{d}}{\text{d}t}\!\tr{\rho_j O_i}$ as $\frac{\text{d}}{\text{d}t}p_{i,j}(t)|_{t=0}$. For this approach to work, we have to be able to control the degree of the polynomial $p_{i,j}(t)$ in order to give an upper bound on the number of sampling points $t_k$ for which we will have to determine $\avg{O_i(t_k)}_j$ experimentally. In the following, we briefly outline how such a guarantee on the degree of $p_{i,j}(t)$ can be obtained and refer to Append.~\ref{sec:quality_approx_poly} for a detailed proof.

Our argument proceeds in two steps. Firstly, we establish that the expectation value $\avg{O(t)}$ of a local observable $O$ that evolves under a Lindbladian $\cL_B$ restricted to some sub-region $B$ up to some time $t_{max}$, can indeed be approximated up to error $\varepsilon$ by a degree-$d$ polynomial, where $d$ depends linearly on the size of $B$, $t_{max}$ and $\log(\varepsilon^{-1})$. 
Hence, for the second step of our argument, it remains to show under which circumstances, we can restrict the evolution of the Pauli-strings $P^a_{m}$ we identified in the previous step to a local generator. The 
main insight here is that for finite range (or sufficiently quickly fast decaying) interactions, the dynamics of any local observable $O$ exhibits an effective light cone quantified by a Lieb-Robinson bound (LRB) \cite{Lieb_1972,poulin_lieb-robinson_2010,bach_lieb-robinson_2014,hastings_locality_2010,kliesch2014lieb,kuwahara_strictly_2020}. The LR-bound in turn allows us to restrict the Lindbladian on the full system to a generator coupling only systems in the vicinity of the support of $O$, where the size of this shielding region only grows linearly with $t_{max}$. Hence, bringing these two arguments together, we can first employ the LR-bound to restrict the dynamics to a sub-region around the support of the Pauli-string, $P^a_{m}$, and then approximate  the corresponding evolution on that finite region up to error $\varepsilon$ by a polynomial of degree $\cO\left[\operatorname{poly}(t_{\max},\log(\epsilon^{-1}))\right]$. Now, making use of the techniques from Ref.~\cite{kane_robust_2017}, we can extract the first derivative of this polynomial from measurements at $\cO\Big[\operatorname{polylog}(\epsilon^{-1})\Big]$ different times $t_k$.

\section{Numerical examples}\label{sec:numerics}
To investigate the performance of our protocol for experimentally relevant parameters, we performed numerical simulations of a multi-qubit superconducting device. We
consider a system with tunable couplers similar to the Google Sycamore chip~\cite{Arute2019}. This design relies on a cancellation of the next-next-nearest coupling between two qubits through the direct coupling with a coupler~\cite{Yan2018,Sung2021}. 
We consider a generic system consisting of a 2D grid of qubits with exchange coupling between nearest neighbors. 
The dynamics are described through a Lindblad equation
with the effective two-qubit Hamiltonian for each neighboring qubit pair ($i,j$)~\cite{Yan2018,Sung2021} 
\begin{eqnarray}
\label{eq:hamil1}
\!\!\!H_{ij}=\sum_{k=i,j}\frac{1}{2}\tilde{\omega}_k\sigma^{(k)}_z\!\!+\!\!\left[\frac{g_ig_j}{\Delta_{ij}}+g_{ij}\right](\sigma^{(i)}_+\sigma^{(j)}_-\!+\!\sigma^{(i)}_-\sigma^{(j)}_+)\quad
\end{eqnarray} 
for $i\neq j=1,\dots, n$ and a dissipation term $\D^{(i)}$ acting on the $i$'th qubit and  having jump operators $\sigma^{(i)}_-, \sigma^{(i)}_+$ (generalised amplitude damping) and $\sigma^{(i)}_z$ (pure dephasing).
Here $\tilde{\omega}_j=\omega_j+\frac{g_j^2}{\Delta_j}$ is the Lamb-shifted qubit frequency, $g_i$ is the coupling between the $i$'th qubit and the coupler, and $g_{ij}$ is the direct two-qubit coupling. We have assumed that $\Delta_j=\omega_j-\omega_c<0$ where $\omega_c$ ($\omega_j$) is the frequency of the coupler ($j$'th qubit) and have defined  $1/\Delta_{ij}=(1/\Delta_i+1/\Delta_j)/2$. By adjusting the frequencies of the coupler and the qubits, the effective qubit-qubit interaction can be cancelled up to experimental precision. Typical qubit frequencies are around $5-6$ GHz~\cite{Arute2019}), while $\Delta\sim-1$ GHz, $g_{ij}\sim10-20$ MHz, and $g_j\sim100$ MHz~\cite{Yan2018,Sung2021}. 
In our simulation, we assume that all qubit frequencies and couplings have been characterized up to a precision of $100$ kHz using standard characterization techniques~\cite{Arute2019} and consequently, that all couplers have been tuned off with the same precision i.e. $\frac{g_ig_j}{\Delta_{ij}}+g_{ij}\sim 100$ kHz. Considering a layout of 16 qubits (see Fig.~\ref{fig:precision_2d} for the interaction graph), we randomly sample all qubit frequencies and qubit-qubit interactions according to Gaussian distributions with zero mean and standard deviation of $100$ kHz. 

In addition to the Hamiltonian evolution, we also include dissipative dynamics in our numerical simulation. We include quasi-static random frequency shifts of the qubits leading to effective dephasing with a characteristic timescale of $T_2^* \sim 150$ $\mu$s as well as pure dephasing resulting in a transverse relaxation on a timescale $T_2\sim 60$ $\mu$s representing state of the art coherence times~\cite{Arute2019,Sung2021,Sung2021}. Finally, we include longitudinal relaxation of the qubits through an amplitude damping channel on the time scale of $T_1\sim60$ $\mu$s. We refer to Appendix \ref{sec:example_Hamiltonian} for a more detailed discussion and Table~\ref{tab_2_1} for the sampled parameters of our simulation. 

\begin{center}
    \begin{figure}[t]
    \includegraphics[width=1.\linewidth]{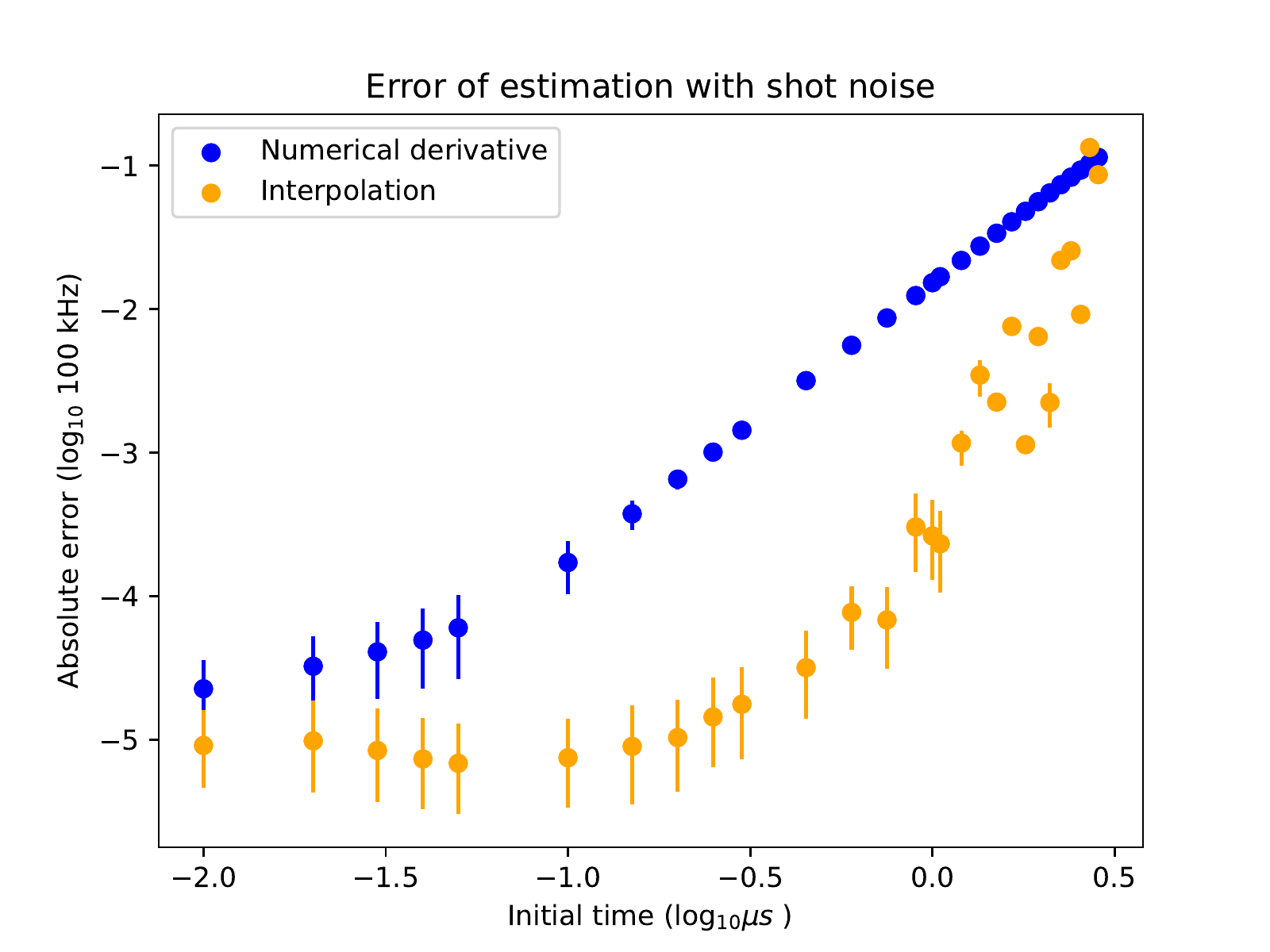}
    \caption{The median quality of recovery of of one $2$-qubit coupling using our method and those based on numerical derivatives \cite{PhysRevLett.107.210404,Bairey2020,Zubida2021} as a function of the initial time. We assumed that the total time of the experiment is fixed. That is, we let the initial time times the total number of samples for each time step to be a constant ($10^7$ for this plot). For each initial time, we simulated $1000$ instances of the recovery protocol, always adding shot noise with the same standard deviation to the data. The dots correspond to the median quality of recovery, whereas the lower and upper end correspond to the $25$ and $75$ percentile. 
    We ran the simulation on a system with $16$ qubits. }
    \label{fig:precision_1}
\end{figure}
\end{center} 

\begin{center}
    \begin{figure}[t]
        \includegraphics[width=1.\linewidth]{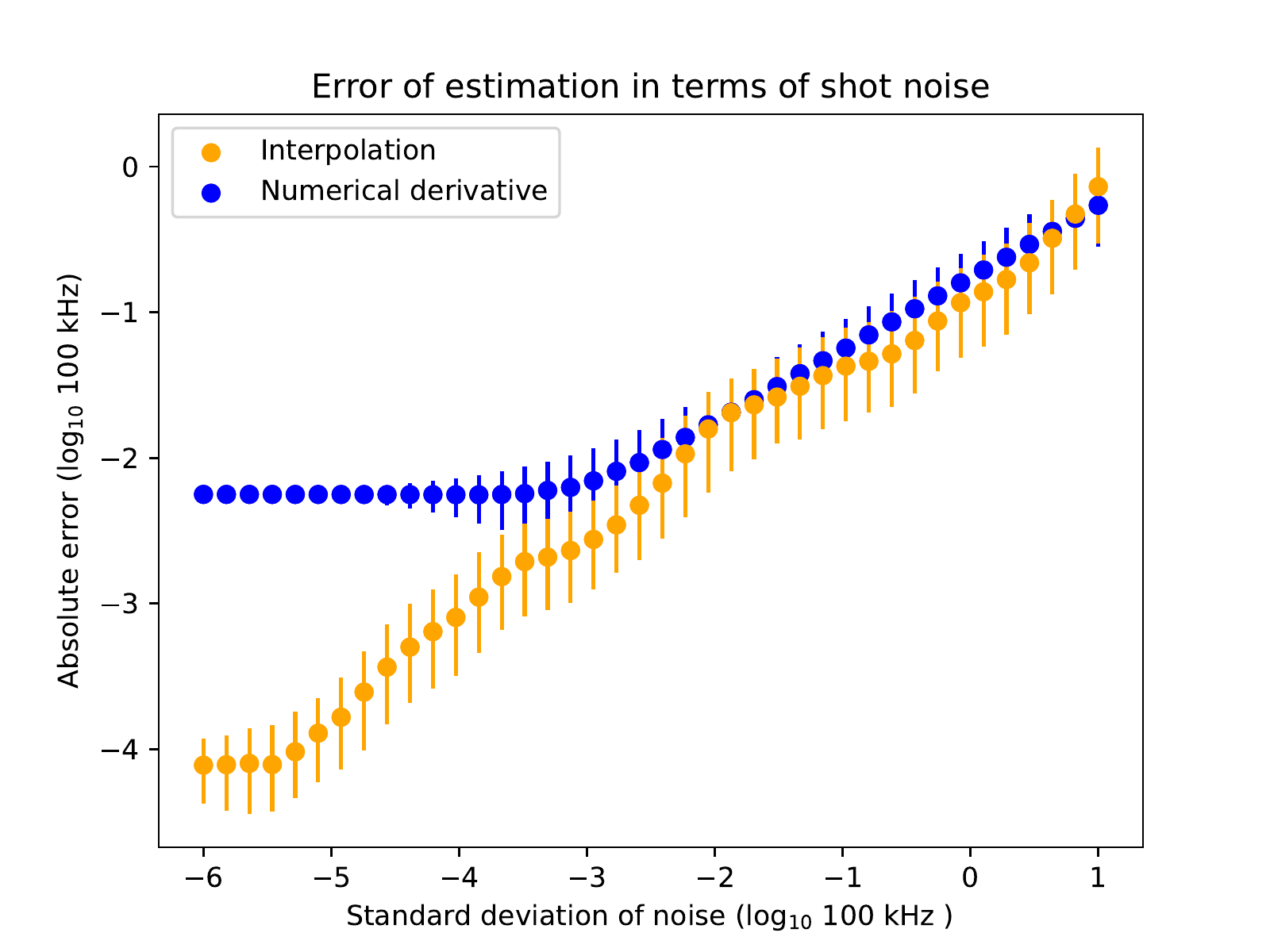}
         \caption{Median quality of recovery of one $2$-qubit coupling using our method and those based on numerical derivatives~\cite{PhysRevLett.107.210404,Bairey2020,Zubida2021} as a function of the standard deviation of the shot noise. The initial time for this estimate is $30$ ns and here we also generated $1000$ instances of the noise with a given standard deviation. The plot shows the median quality of the recovery and the $25$ and $75$ percentiles.
    We see that the quality of the recovery for the interpolation decays approximately linearly with the shot noise. For the numerical derivative, we see two regimes: first a linear decay of the error until a shot of noise of order $10^{-3}$. After that, the error plateaus and does not improve even with smaller shot noise. This is because for numerical derivative methods, at this point the dominant error source comes from the choice of initial time. Importantly, we see that interpolation consistently provides better estimates than the numerical derivatives method.
        }
        \label{fig:shot_noise} 
    \end{figure}  
    
\end{center} 
In Fig.~\ref{fig:precision_1}, we plot the average estimation error as a function of the temporal resolution set by the value of the initial time step, $t_0$. For this plot, we only included the Hamiltonian evolution in the numerical simulation together with quasi-static random frequency shifts of the qubits. This was to lower the run time of the simulation allowing us to investigate the performance for a broad range of initial times. We assumed the total run time of the experiment was fixed such that $t_0\times S$ is constant, where $S$ is the number of samples. From the figure, we clearly see the improved scaling of our protocol of the estimation error with the time-step size compared to using a finite difference  method~\cite{PhysRevLett.107.210404,Bairey2020,Zubida2021}. Besides already performing better at the time resolution for moderate values of the initial time, we see that after a threshold initial time around $10^{-0.7}$, the performance is not limited by the initial time, only the shot noise. In contrast, the finite difference method still require smaller initial times to improve on the error with the same shot noise. 

We also investigated the robustness of our method with respect to shot-noise for a fixed time resolution. For these simulations, we again only included the Hamiltonian evolution together with quasi-static random frequency shifts of the qubits to have a practical run time of the simulation. From Fig.~\ref{fig:shot_noise} we see that for a fixed time resolution of $30$ ns our protocol results in an average estimation error that improves linearly with the shot-noise down to an error below $10^{-4}$. This is in contrast to finite difference methods, where the estimation error plateaus around $10^{-3}$ since it becomes limited by the time resolution. This is a clear effect of the exponential improvement of our protocol w.r.t. the time resolution compared to finite difference methods. 

Finally, we performed a full numerical simulation including also the pure dephasing and amplitude damping noise as described above and estimated the $\sigma_X\sigma_X$ couplings between the qubits. As shown in Fig.~\ref{fig:precision_2d}, we obtain reliable estimates of all 22 couplings demonstrating how our method allows the estimation of specific terms in the Lindbladian despite the dynamics being governed by the full (dissipative) Lindbladian. For simplicity, we did not explicitly estimate the single qubit Hamiltonian parameters and the Lindbladian decay rates. 

For all estimations above, we fitted to degrees $1-7$ and picked the one with the smallest average error. 
We note that, although the robust interpolation methods of~\cite{kane_robust_2017} in principle require random times, we performed numerical experiments with deterministic times on systems with $16$ qubits.

\begin{center}
    \begin{figure}[t]
    \includegraphics[width=1.\linewidth]{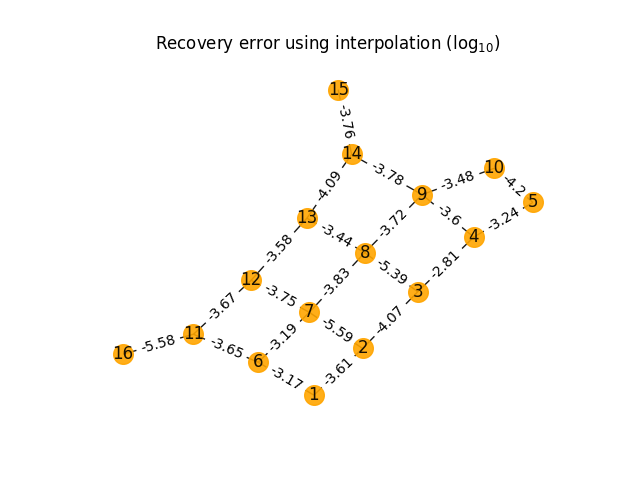}
    \includegraphics[width=1.\linewidth]{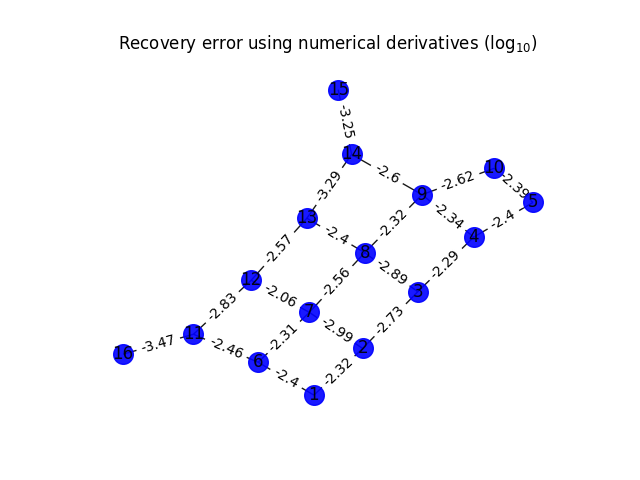}
    \caption{
        Error in recovery of $\sigma_X\sigma_X$ couplings of a quantum systems with a geometry similar to the Sycamore processor usign numerical derivatives and interpolation. Note that while we only plot the estimation of the Hamiltonian couplings, the numerical simulation included the full Lindbladian including both dephasing due to quasi-static random frequency shifts of the qubits, pure dephasing and amplitude damping noise. The initial time for each coupling was $0.1$ $\mu$s in the simulation. Note that interpolation consistently outperforms numerical derivatives, sometimes by several orders of magnitude.}
    \label{fig:precision_2d}
\end{figure}
\end{center}

\section{Conclusion and Discussion}
In conclusion, we have proposed a novel Hamiltonian learning protocol based on robust polynomial interpolation that has rigorous guarantees on the estimation error. Our protocol offers an exponential reduction in the required temporal resolution of the measurements compared to previous methods and a quadratic reduction in the overall sampling complexity for finite range interactions. Our protocol only requires the preparation of single qubit states and single qubit measurements in the Pauli bases making it suitable for characterization of both near term and future quantum devices.

Furthermore, the recovery of multiple parameters can be highly paralelized by resorting to a variation of classical shadows to quantum channels we introduce here. Besides being a protocol of independent interest, we believe our work constitutes the first application of classical shadows to process tomography.

Our method allows for the characterization of a general local Markovian evolution consisting of a unitary Hamiltonian part and a dissipative part. While we have only explicitly considered single qubit dissipation here, we believe that our protocol is also valid for general multi-qubit dissipation as outlined above but leave the explicit analysis of this to future work.  We have also analysed the performance of our protocol for algebraically decaying interactions which we believe to be the first results for Hamiltonian learning of such systems. The convergence of our method can be ensured for interactions decaying faster than the dimension of the system. We note, however, that improved bounds on the locality of such systems might improve this result in the future.

\section{Methods}
Here we detail and formalize our results regarding the estimation error guarantees of our protocol. In particular, we detail the use of Lieb-Robinson bounds on the spread of correlations in the system to bound the error. Furthermore, we outline the shadow tomography method for the parallization of the measurements.  

\subsection{Derivative estimation} 

Define $f(t)=\tr{e^{t\cL}(O)\rho}$ and $\cL_B$ to be the Lindbladian truncated to a subregion $B$ of the interaction graph. Our protocol consists of first estimating $f(t_i)$ up to an error $\cO(\epsilon)$ for random times $t_1,\ldots,t_m$. The curve of $f(t)$ is then fitted to a low-degree polynomial $p$, and $p'(0)$ is taken as an estimate for the derivative $f'(0)=\tr{\cL(O)\rho}$. Below we prove the accuracy and robustness of this method. The first step is Theorem \ref{thm:polynomialswork_main}, which establishes under what conditions $f(t)$ is indeed well-approximated by a low-degree polynomial. 

\begin{theorem}\label{thm:polynomialswork_main}
    Let $\cL$ be a local Lindbladian on a $D$-dimensional lattice. Moreover, let $t_{\max},\epsilon>0$ and $O_Y$ be a $2$-qubit observable, such that $\|O_Y\|\leq 1$, holds. 
    Then there is a polynomial $p$ of degree 
    \begin{align}\label{equ:polynomial_thm1}
        d=\cO\left[\operatorname{poly}(t_{\max},\log(\epsilon^{-1}))\right],
    \end{align}
    such that for all $0\leq t\leq t_{\max}$:
    \begin{align}\label{equ:quality_approx_1}
        \left|\tr{e^{t\cL_V}(O_Y)\rho}-p(t)\right|\leq\epsilon,
    \end{align}
    and $     p'(0)=\tr{\linGen(O_Y)\rho}$,     holds.
    \end{theorem}
The main technical tool required for the proof are Lieb-Robinson bounds (LRB) \cite{Lieb_1972,poulin_lieb-robinson_2010,bach_lieb-robinson_2014,hastings_locality_2010,kliesch2014lieb,kuwahara_strictly_2020}, which ascertain that the dynamics of local observables under a time evolution with a local Lindbladian have an effective lightcone. More precisely, we need 
\begin{align}\label{equ:LRexpo_main_text}
    &\|(e^{t\cL_B}-e^{t\cL_V})(O_Y)\|\leq \nonumber\\&c_1\operatorname{exp}(-\mu\operatorname{dist}(Y,V\backslash\{B\}))(e^{vt}-1),
\end{align}
to hold for constants $c_1,\mu$ and $v$, where $\operatorname{dist}()$ is the distance in the graph.

From the LRB we can show that the dynamics is well-approximated by a low-degree polynomial. 
We leave the details of the proof to Sec.~\ref{sec:quality_approx_poly} in the supplemental material and only discuss the main steps here. 
The general idea of going from the LRB to the low-degree polynomial is to truncate the Taylor series of the evolution under $\cL_B$ for $B$ large enough and take that as the approximating polynomial. As the derivatives of the evolution under $\cL_B$ only scale with the size of the region $B$, this allows us to show that the Taylor series converges quickly.

Now that we have concluded that the expectation value is well-approximated by a small degree polynomial, we continue to show that we can reliably extract the derivative from approximations of the expectation values for different $t$. This is formally stated in the following theorem.  

\begin{theorem}\label{thm:interpolation_robust}
   Let $\cL$ be a Lindbladian on a $D$-dimensional regular lattice.
    Suppose we can measure the expectation value of two-body Pauli observables on Pauli eigenstates in the time interval $[t_0,t_{\max}]$ under $\cL$ for $t_0$ as
    \begin{align}
        t_0^{-1}=\cO\left[\operatorname{polylog}(\epsilon^{-1})\right]
    \end{align}
    and $t_{\max}=2+t_0$.
     Then, measuring the expectation values at
    \begin{eqnarray}\label{equ:number_points}
        m=\cO\Big[\operatorname{polylog}(\epsilon^{-1})\Big]
    \end{eqnarray}
    random times up to precision $\cO(\epsilon/\operatorname{polylog}(\epsilon^{-1}))$, is sufficient to obtain an estimate $\hat{a}_{\alpha_1,\alpha_2}^{(m_1)}$ of  $a_{\alpha_1,\alpha_2}^{(m_1)}$ satisfying
    \begin{align}\label{equ:final_error_2}
        \left|\hat{a}_{\alpha_1,\alpha_2}^{m_1}-a_{\alpha_1,\alpha_2}^{m_1}\right|= \epsilon.
    \end{align}
    This yields a total sample complexity of $S=\cO\left(\epsilon^{-2}\operatorname{polylog}(\epsilon^{-1})\right)$.
    \end{theorem}

Of course, the same results also hold for the parameters $a_{\alpha_1}^{(m_1)}$ and those of $\cL^{(m_1)}$.
Importantly, Theorem~\ref{thm:interpolation_robust} bypasses both requiring small initial times and $\cO(\epsilon^{-4})$ sample complexities.

To go from Thm.~\ref{thm:polynomialswork_main} to Thm.~\ref{thm:interpolation_robust} we first need to establish that we can robustly infer an approximation of $p$ from finite measurement data subject to shot noise. Subsequently, we need to show that it will also allow us to reliably estimate $p'(0)$. Let us start with approximating $p$.
\paragraph{Robust polynomial interpolation:} 
We will resort to the robust polynomial interpolation methods of \cite{kane_robust_2017} to show Thm.~\ref{thm:interpolation_robust}. We review their methods in more detail in \cite[Sec.~\ref{sec:robust_interpolation}]{SM}. But they depart from the assumption that we get $m$ (randomly sampled) points $x_1,\ldots,x_m\in[t_0,t_{\max}]$ and $y_1,\ldots,y_m\in\R$. In our setting, the $x_i$ correspond to different times and the $y_i$ to approximations of the expectation value of the evolution at that time.
Furthermore, the $y_i$ satisfy the promise that there exists a polynomial $p$ of degree $d$ and some $\sigma>0$ such, that
\begin{align}
    y_i=p(x_i)+w_i,\quad |w_i|\leq\sigma,
\end{align}
hold, for strictly more than half of the $y_i$.  The rest might be outliers. In our setting, the magnitude of $\sigma$ corresponds to amount of shot noise present in the estimates of the expectation values.

The authors of~\cite{kane_robust_2017} then show that by sampling $m=\cO(d\log(d))$ points from the Chebyshev measure on $[t_0,t_{\max}]$, a combination of $\ell_1$ and $\ell_\infty$ regression allows us to find a polynomial $\hat{p}$ of degree $d$ that satisfies:
\begin{align}\label{equ:polynomial_close_main}
    \max\limits_{x\in[t_0,t_{m
    ax}]}|p(x)-\hat{p}(x)|=\cO(\sigma).
\end{align}

Although the details of the $\ell_1$ and $\ell_\infty$ interpolation are more involved and described in Sec.~\ref{sec:robust_interpolation} of the supplemental material, a rough simplification of the procedure is the following. First, we find a polynomial $p_1$ of degree $d$ that minimizes $\sum_i|p_1(x_i)-y_i|$. 
After finding $p_1$ we compute the polynomial $p_{\infty}$ that minimizes $\max_i|p_\infty(x_i)-(y_i-p_1(x_1))|$.
We then output $\hat{p}=p_1+p_\infty$ as our guess polynomial. 
Note that finding both $p_1$ and $p_\infty$ can be cast as linear programs and thus can be solved efficiently~\cite{StephenBoyd2004}.

By combining this result with Thm.~\ref{thm:polynomialswork_main}, we robustly extract a polynomial that approximates the curve $t\mapsto \tr{e^{t\cL}(O_Y)\rho}$ up to $\cO(\epsilon)$ for $t\in[t_0,t_{\max}]$.
Indeed, we only need to estimate the expectation value $f(t_i)$ up to $\epsilon$ for enough $t_i$ and run the polynomial interpolation.

Note that Eq.~\eqref{equ:polynomial_close_main} only allows us to conclude that $p-\hat{p}$ is small. However, we are ultimately interested in the curve's derivative at $t=0$, as the derivative contains information about the parameters of the evolution.
For arbitrary smooth functions, two functions being close on an interval does not imply that their derivatives are close as well. Fortunately, for polynomials the picture is simpler. A classical result from approximation theory, Markov brother's inequality~\cite{markoff_ber_1916}, allows us to quantify the deviation of the derivatives given a bound on the degree and a bound like Eq.~\eqref{equ:polynomial_close_main}. Putting these observations together, we arrive at Thm.~\ref{thm:interpolation_robust}.
The details of the proof are given in Sec.~\ref{sec:robust_interpolation} of the supplemental material.
\subsection{Generalizations of Thm.~\ref{thm:interpolation_robust}}
We also generalize Thm.~\ref{thm:interpolation_robust} in two directions. First, we extend the results to interactions acting on $k$ qubits instead of $2$. As long as the noise is constrained to acts on $1$ qubit and $k=\cO(1)$, this generalization is straightforward. Indeed, we only need to measure an observable that has the same support as the Pauli string and does not commute with it, as it is then always possible to find a product initial state that isolates the parameter. Generalizing to noise acting on more than one qubit makes it more difficult to isolate the parameters of the evolution as described in the main text. In that case, it then becomes necessary to solve a system of linear equations that couples different parameters. Although our method still applies, analysing this scenario would require picking the observables and initial states in a way that the system of equations is well-conditioned and we will not discuss this case in detail here.

Second, another important generalization is to go beyond short-range systems. Although we have only stated our results for short-range systems in Thm.~\ref{thm:interpolation_robust}, our techniques apply to certain long-range systems. 
As this generalization is more technical, we leave the details to Appendix~\ref{sec:robust_interpolation} and constrain ourselves to discussing how the statement of Thm.~\ref{thm:interpolation_robust} changes for more general interactions.

Only one aspect of the precious discussion changes significantly for long range interactions: how the r.h.s. of Eq.~\eqref{equ:LRexpo_main_text} generalizes. More precisely, let us assume that for some injective function $h:\R\to\R$ with $h(x)=o(1)$,  we have
\begin{align}\label{equ:LRgeneral_main_text}
    &\|(e^{t\cL_B}-e^{t\cL_\Gamma})(O_Y)\|\leq \nonumber\\&h(\operatorname{dist}(Y,V\backslash\{B\}))(e^{vt}-1).
\end{align}
For instance, for short-range or exponentially decaying interactions,  $h$ will be an exponentially decaying function.
Then we can  restate Thm.~\ref{thm:interpolation_robust} in terms of $h^{-1}$. 
As we show in Thm.~\ref{thm:choice_initial_final} in Appendix F, for a precision parameter $\epsilon>0$ and evolution on a $D$-dimensional lattice, assume that we pick the initial time as 
\begin{align*}
t_0=\cO\left[\left(h^{-1}\left(\frac{\epsilon}{2(e^{2.5v}-1)}\right)^{D}\log(\epsilon^{-1})\right)^{-2}\right].
\end{align*}
Furthermore, assume that we estimate the expectation value of local observables up to precision $\cO(\epsilon)$ at $\tcO\left[\left(h^{-1}\left(\frac{\epsilon}{2(e^{2.5v}-1)}\right)\right)^{D}\log(\epsilon^{-1}))\right]$ points. Then we can estimate each parameter up to an error of 
\begin{align}\label{equ:final_error2}
   \cO\left[\epsilon\left(h^{-1}\left(\frac{\epsilon}{2(e^{2.5v}-1)}\right)^{D}\log(\epsilon^{-1})\right)^{2}\right],
\end{align}
through, the same procedure as in the local case. 
Note that the error in Eq.~\eqref{equ:final_error2} only tends to $0$ as $\epsilon\to0$ if  $h^{-1}\left(\frac{\epsilon}{2(e^{2.5v}-1)}\right)^{D}\log(\epsilon^{-1})=o(\epsilon^{-1})$, holds, i.e. the function $h$ must decay fast enough. In Sec.~\ref{sec:lieb-robinson} %
of the supplemental material, we discuss examples of systems with algebraically decaying interactions for which this is satisfied. For instance, for potentials that decay like $r^{-\alpha}$ with $\alpha>5D-1$ we obtain that $h^{-1}(\epsilon)=\cO(\epsilon^{-\frac{1}{\alpha-3D}})$, holds. We summarize the resulting resources in Tab.~\ref{tab:summary-resources} in the supplemental material.

But the message of bounds like~\eqref{equ:final_error2} is that it is still possible to obtain bounds on the error independent of the system's size beyond short-range systems. However, this comes at the expense of requiring higher precision and sampling from more points. 

Another important observation is that the assumption that we know the structure of the interactions exactly is not required. Indeed, our method is robust to Hamiltonian perturbations of the model as long as the resulting evolution still satisfies a LR bound. For instance, suppose that there actually is a non-negligible interaction between qubits $i$ and $j$ that is not accounted by our model. As long as the resulting time evolution still satisfies a LR bound, our results still hold. As the linear equation to isolate any parameter is independent of that parameter, we can still apply our techniques in this setting.

\subsection{Parallelizing the measurements}
So far we have only discussed how to obtain the estimate of one parameter of the state from experimental data in an efficient manner. However, it is possible to parallelize the measurement and ensure that we can obtain experimental data to estimate all parameters simultaneously. 

To that end, we resort to a classical shadow process tomography method. Although some papers in the literature already discussed classical shadows for process tomography~\cite{processtomo,processtomo2}, their proofs unfortunately contain a shortcoming. These previous approaches where based on applying the classical shadows protocol to the Choi state of the underlying evolution and importing existing classical shadow tomography results for states~\cite{Huang2020}. Unfortunately, this proof method produces a prefactor in the sample complexity that is exponential in the system size and was missed in previous works.
Fortunately, we give an alternative proof in Sec.~\ref{sec:shadows_parallel} for process shadow tomography that does not require the Choi state and bypasses this issue. In addition, our result also improves the sample complexity of previous results.

More precisely, we show that given a quantum channel $\Phi$, Pauli strings $P_1^a,\ldots,P^a_{K_1}$ that differ from the identity on at most $\omega_a$ sites and Pauli strings $P_1^b,\ldots P_{K_2}^b$ that differ from the identity on at most $\omega_b$ sites, it is possible to obtain estimates $\hat{e}_{m,l}$ of $2^{-n}\tr{P_m^a\Phi(P_l^b)}$ satisfying
\begin{align*}
    |2^{-n}\tr{P_m^a\Phi(P_l^b)}-\hat{e}_{m,l}|\leq\epsilon
\end{align*}
for all $m,l$ with probabillity at least $1-\delta$ from
\begin{align}\label{equ:number_samples_procc_shadows}
  \cO(3^{\omega_a+\omega_b}\epsilon^{-2}\log(K_1K_2\delta^{-1}))  
\end{align}
samples. More precisely, the protocol of shadow process tomography requires preparing Eq.~\eqref{equ:number_samples_procc_shadows} many different random initial product Pauli eigenstates and measuring them in random Pauli bases. This makes it feasible to implement it in the near-term. We discuss it in more detail in Sec.~\ref{sec:shadows_parallel} of the supplemental material, as this protocol may be of interest beyond the problem at hand.

The shadow process tomography protocol is ideally suited for our Hamiltonian learning protocol. Indeed, note that to learn $k$-body interactions, we only required the preparation of initial states $\rho_l$ that differ from the maximally mixed state on $k$ qubits and measure Pauli strings $P_m$ supported on at most $k$ qubits. Furthermore, for a system of $n$ qubits in total, there are at most $16^k\binom{n}{k}\leq 16^kn^k$ such states or Pauli strings. We conclude that we can estimate all required expectation values for a given time step using
\begin{align*} 
    \cO(9^k\epsilon^{-2}k\log(n\delta^{-1}))  
\end{align*}
samples.
As our protocol requires estimating expectation values at a total of $\textrm{polylog}(\epsilon^{-1})$ time steps, we can gather the data required to recover all the $\cO(n)$ parameters of the evolution from $\cO(\epsilon^{-2}\textrm{polylog}(n,\epsilon^{-1}))$ samples through the shadow process tomography protocol whenever $k=\cO(1)$.

\section{Acknowledgements}
D.S.F. was supported by VILLUM FONDEN via the QMATH Centre of Excellence under Grant No. 10059 and the European Research Council (Grant agreement No. 818761). A.H.W. thanks the VILLUM FONDEN for its support with a Villum Young Investigator Grant (Grant No. 25452). J.B. and L.A.M. acknowledge funding from the NWO Gravitation Program Quantum Software Consortium. V.V.D.\ work is a part of the research programme NWO QuTech Physics Funding (QTECH, programme 172) with project number 16QTECH02, which is (partly) financed by the Dutch Research Council (NWO); the work was partially supported by the Kavli Institute of Nanoscience Delft. 

\bibliographystyle{unsrt}
\bibliography{Bibliography}
\onecolumngrid

\appendix

\section{Selecting States and Observables to Isolate Parameters}\label{sec:howtoextractparameters}
Any Hamiltonian can be written as
\begin{eqnarray}\label{10_50_0}
H&=&\sum\limits_{m_1}\sum\limits_{\alpha_1} a_{\alpha_1}^{(m_1)}\sigma_{\alpha_1}^{(m_1)}+\sum\limits_{m_1,m_2}\sum\limits_{\alpha_1\alpha_2} a_{\alpha_1\alpha_2}^{(m_1,m_2)}\sigma_{\alpha_1}^{(m_1)}\sigma_{\alpha_2}^{(n)}+\dots\\\nonumber
&\equiv&\sum\limits_{m_1}\sum\limits_{\alpha_1} a_{\alpha_1}^{(m_1)}\mathcal{H}_{\alpha_1}^{(m_1)}+\sum\limits_{m_1,m_2}\sum\limits_{\alpha_1\alpha_2} a_{\alpha_1\alpha_2}^{(m_1,m_2)}\mathcal{H}_{\alpha_1\alpha_2}^{(m_1,m_2)}+\dots,
\end{eqnarray}
where Roman indices identify the subspace on
which the operator acts, and Greek indices identify the
Pauli operator, e.g. $\alpha=x$. No assumption about the dimension or structure of the hermitian Hamiltonian is needed for this expansion to be valid. 
For a Markovian noise environment, the
evolution of a quantum system $\rho_0$ is described by a
Master equation of the form
\begin{eqnarray}\label{10_50_1}
\frac{d\rho_t}{dt}\Big|_{t=0}=-i[H,\rho_0]+\sum\limits_{m=1}^{n}\sum\limits_{\mu,\nu=1}^{3}D^{(m)}_{\mu,\nu}
(\sigma_{\mu}^{(m)}\rho_0{\sigma_{\nu}^{(m)}}^{\dagger}-\frac{1}{2}\{{\sigma_{\nu}^{(m)}}^{\dagger}\sigma_{\mu}^{(m)},\rho_0\}),
\end{eqnarray}
where $H$ is the Hamiltonian describing the evolution of the system,  $L^{(m)}_{\mu,\nu}$ are the elements of the  Lindblad matrix, expressed in an operator basis consisting
of the different combinations of single-qubit Pauli matrices
$\{\sigma\}$.
Multiplying it from the right hand side on an observable $O$ and taking the trace, we can write
\begin{eqnarray}\label{10_50_2_1116}
\frac{d}{dt}\tr{\rho_t O}|_{t=0}&=&-i\sum\limits_{m_1}\sum\limits_{\alpha_1} a_{\alpha_1}^{(m_1)}\tr{[H_{\alpha_1}^{(m_1)},\rho_0]O}-i\sum\limits_{m_1m_2}\sum\limits_{\alpha_1\alpha_2} a_{\alpha_1\alpha_2}^{(m_1m_2)}\tr{[H_{\alpha_1\alpha_2}^{(m_1m_2)},\rho_0]O}-\dots\nonumber\\
&\dots &-i\sum\limits_{m_1\dots m_k}\sum\limits_{\alpha_1\dots \alpha_k} a_{\alpha_1\dots \alpha_k}^{(m_1\dots m_k)}\tr{[H_{\alpha_1\dots \alpha_k}^{(m_1\dots m_k)},\rho_0]O}\dots\\\nonumber
&+&\sum\limits_{m}\sum\limits_{\mu,
\nu}D^{(m)}_{\mu\nu}\tr{
(\sigma_{\mu}^{(m)}\rho_0{\sigma_{\nu}^{(m)}}^{\dagger}-\frac{1}{2}\{{\sigma_{\nu}^{(m)}}^{\dagger}\sigma_{\mu}^{(m)},\rho_0\})O}.
\end{eqnarray}
Let us introduce the notation
\begin{eqnarray}\label{16_59}
{B}^{(m_1,\dots,m_k)}_{\alpha_1,\dots,\alpha_k}(\rho_0,O)&\equiv&-i\tr{[H_{\alpha_1\dots \alpha_k}^{(m_1\dots m_k)},\rho_0]O}.
\end{eqnarray}
To isolate $\{a_{\alpha_1}^{(m_1)}, a_{\alpha_1\alpha_2}^{(m_1m_2)},\dots\}$ we observe the $N$ qubit state with the following density matrix
\begin{eqnarray}\label{0_10}
\rho_{\tau_i,\tau_j}^{(i,j)}=\rho_{\tau_i}^{(i)}\otimes \rho_{\tau_j}^{(j)}\otimes \rho_{N-2},\quad \tau=\{1,2,3\}
\end{eqnarray}
where the $i$ and $j$ are the Pauli qubits, namely $\rho_{\tau_{i,j}}^{(i,j)}=(I+\sigma^{(i,j)}_{\tau_{i,j}})/2$, and 
$\rho_{N-2}$ is the density matrix of all other qubits, which we set to be maximally mixed. 
Using $\sigma_{\alpha} \sigma_{\beta}=\delta_{\alpha\beta}I+i\varepsilon_{\alpha\beta\gamma}\sigma_{\gamma}$, we can find the following relations
\begin{eqnarray}\label{0_15_10_1}
{B}^{(i)}_{\alpha_1}(\rho_{\tau_i,\tau_j}^{(i,j)},O)%
&=&\frac{1}{2}\Big(\varepsilon_{\alpha_1\tau_i \gamma}\left(\tr{(\sigma^{(i)}_{\gamma}\otimes I)O}+\tr{(\sigma^{(i)}_{\gamma}\otimes \sigma_{\tau_j}^{(j)})O}\right)\Big),
\end{eqnarray}
\begin{eqnarray}\label{0_15_10_1115}
{B}^{(j)}_{\alpha_1}(\rho_{\tau_i,\tau_j}^{(i,j)},O)%
&=&\frac{1}{2}\Big(\varepsilon_{\alpha_1\tau_j \zeta}\left(\tr{( I\otimes\sigma^{(j)}_{\zeta})O}+\tr{(\sigma^{(i)}_{\tau_i}\otimes \sigma_{\zeta}^{(j)})O}\right)\Big),
\end{eqnarray}
\begin{eqnarray}\label{0_15_10}
{B}^{(ij)}_{\alpha_1\alpha_2}(\rho_{\tau_i,\tau_j}^{(i,j)},O)%
&=&\frac{1}{2}\Big(\varepsilon_{\alpha_1\tau_i \gamma}\left(\delta_{\alpha_2\tau_j}\tr{(\sigma^{(i)}_{\gamma}\otimes I)O}+
\tr{(\sigma_{\gamma}^{(i)}\otimes \sigma_{\alpha_2}^{(j)})O}\right)\\\nonumber
&+&\varepsilon_{\alpha_2\tau_j\eta}\left(\delta_{\alpha_1\tau_i}\tr{(I\otimes \sigma^{(j)}_{\eta})O}
+\tr{(\sigma_{\alpha_1}^{(i)}\otimes\sigma_{\eta}^{(j)})O}\right)\Big),
\end{eqnarray}
where $O$ is acting on $(i,j)$ qubits.
Selecting $O=\sigma_{\xi_i}^{(i)}\otimes \sigma_{\xi_j}^{(j)}$, we can rewrite the latter matrix element as follows
\begin{eqnarray}\label{0_15_1}
{B}^{(i)}_{\alpha_1}(\rho_{\tau_i,\tau_j}^{(i,j)},\sigma_{\xi_i}^{(i)}\otimes \sigma_{\xi_j}^{(j)})&=&2\varepsilon_{\alpha_1\tau_i \gamma}\delta_{\xi_i\gamma}\delta_{\xi_j\tau_j}
=\begin{cases}
2\varepsilon_{\alpha_1\tau_i \gamma}\delta_{\eta\tau_j}&\text{if $\xi_i=\gamma$, $\xi_j=\eta$,}\\
2\varepsilon_{\alpha_1\tau_i \gamma}  &\text{if $\xi_i=\gamma$, $\xi_j=\eta$, $\eta=\tau_j$,}\\
0&\text{else}\\
 \end{cases},\\\nonumber
 {B}^{(j)}_{\alpha_1}(\rho_{\tau_i,\tau_j}^{(i,j)},\sigma_{\xi_i}^{(i)}\otimes \sigma_{\xi_j}^{(j)})&=&2\varepsilon_{\alpha_1\tau_j \zeta}\delta_{\xi_j\zeta}\delta_{\xi_i\tau_i}
=\begin{cases}
2\varepsilon_{\alpha_1\tau_j \zeta}\delta_{\kappa\tau_i}&\text{if $\xi_j=\zeta$, $\xi_i=\kappa$,}\\
2\varepsilon_{\alpha_1\tau_j \zeta} &\text{if $\xi_j=\zeta$, $\xi_i=\kappa$, $\kappa=\tau_i$,}\\
0&\text{else}\\
 \end{cases},\\\nonumber
{B}^{(ij)}_{\alpha_1\alpha_2}(\rho_{\tau_i,\tau_j}^{(i,j)},\sigma_{\xi_i}^{(i)}\otimes \sigma_{\xi_j}^{(j)})&=&2\left(\varepsilon_{\alpha_1\tau_i\gamma}\delta_{\xi_i\gamma}\delta_{\xi_j\alpha_2}+\varepsilon_{\alpha_2\tau_j\eta}\delta_{\xi_i\alpha_1}\delta_{\xi_j\eta}\right)\\\nonumber&=&\begin{cases}
2\left(\varepsilon_{\alpha_1\tau_i\gamma}\delta_{\eta\alpha_2}+\varepsilon_{\alpha_2\tau_j\eta}\delta_{\gamma\alpha_1}\right) &\text{if $\xi_i=\gamma$, $\xi_j=\eta$,}\\
2\varepsilon_{\alpha_1\tau_i\gamma}\delta_{\tau_j\alpha_2} &\text{if $\xi_i=\gamma$, $\xi_j=\eta$, $\eta=\tau_j$,}\\
0&\text{else}\\
 \end{cases},
\end{eqnarray}
where $\eta,\gamma\in\{x,y,z\}$.
Selecting other observable $O=\sigma_{\xi_i}^{(i)}$, we can rewrite \eqref{0_15_10_1}-\eqref{0_15_10} differently as
\begin{eqnarray}\label{0_15_2}
{B}^{(i)}_{\alpha_1}(\rho_{\tau_i,\tau_j}^{(i,j)},\sigma_{\xi_i}^{(i)})&=&\varepsilon_{\alpha_1\tau_i \gamma}\delta_{\xi_i\gamma}=\begin{cases}
\varepsilon_{\alpha_1\tau_i \gamma}  &\text{if $\xi_i=\gamma$,}\\
0&\text{else}\\
 \end{cases},\quad 
 {B}^{(j)}_{\alpha_1}(\rho_{\tau_i,\tau_j}^{(i,j)},\sigma_{\xi_i}^{(i)})=0,\\\nonumber
{B}^{(ij)}_{\alpha_1\alpha_2}(\rho_{\tau_i,\tau_j}^{(i,j)},\sigma_{\xi_i}^{(i)})&=&2\varepsilon_{\alpha_1\tau_i \gamma}\delta_{\alpha_2\tau_j}\delta_{\xi_i\gamma}=\begin{cases}
2\varepsilon_{\alpha_1\tau_i \gamma} \delta_{\alpha_2\tau_j}&\text{if $\xi_i=\gamma$,}\\
0&\text{else}\\
 \end{cases}.
\end{eqnarray}
Next, selecting $O=\sigma_{\xi_j}^{(j)}$, we can rewrite \eqref{0_15_10_1}-\eqref{0_15_10} as
\begin{eqnarray}\label{0_15_3}
{B}^{(i)}_{\alpha_1}(\rho_{\tau_i,\tau_j}^{(i,j)},O)&=&0,\quad 
{B}^{(j)}_{\alpha_1}(\rho_{\tau_i,\tau_j}^{(i,j)},\sigma_{\xi_i}^{(i)})=\varepsilon_{\alpha_1\tau_j \zeta}\delta_{\xi_j\zeta}=\begin{cases}
\varepsilon_{\alpha_1\tau_j \zeta}  &\text{if $\xi_j=\zeta$,}\\
0&\text{else}\\
 \end{cases},\\\nonumber
{B}^{(ij)}_{\alpha_1\alpha_2}(\rho_{\tau_i,\tau_j}^{(i,j)},\sigma_{\xi_j}^{(j)})&=&2\varepsilon_{\alpha_2\tau_j\eta}\delta_{\alpha_1\tau_i}\delta_{\xi_j\eta}=
\begin{cases}
2\varepsilon_{\alpha_2\tau_j\eta}\delta_{\alpha_1\tau_i}  &\text{if  $\xi_j=\eta$,}\\
0&\text{else}\\
 \end{cases}.
\end{eqnarray}
For $k\geq 2$ we can write the general matrix element:
\begin{eqnarray}\label{0_15_10_2}
{B}^{(m_1,\dots,m_k)}_{\alpha_1,\dots,\alpha_k}(\rho_{\tau_i,\tau_j}^{(i,j)},O)%
&=&\frac{1}{2}\Bigg(\varepsilon_{\alpha_i\tau_i \gamma}\Big(\tr{(\sigma_{\alpha_1}^{(1)}\otimes\dots\otimes \sigma^{(i)}_{\gamma}\otimes\dots\otimes \sigma^{(j)}_{\alpha_j}\otimes\dots\otimes\sigma_{\alpha_k}^{(k)})O}\\\nonumber&+&
\delta_{\alpha_j\tau_j}\tr{(\sigma_{\alpha_1}^{(1)}\otimes\dots\otimes \sigma^{(i)}_{\gamma}\otimes\dots\otimes I^{(j)}\otimes\dots\otimes\sigma_{\alpha_k}^{(k)})O}\Big)\\\nonumber
&+&\varepsilon_{\alpha_j\tau_j \eta}\Big(\tr{(\sigma_{\alpha_1}^{(1)}\otimes\dots\otimes \sigma^{(i)}_{\alpha_i}\otimes\dots\otimes \sigma^{(j)}_{\eta}\otimes\dots\otimes\sigma_{\alpha_k}^{(k)})O}\\\nonumber&+&\delta_{\alpha_i\tau_i}\tr{(\sigma_{\alpha_1}^{(1)}\otimes\dots\otimes I^{(i)}\otimes\dots\otimes \sigma^{(j)}_{\eta}\otimes\dots\otimes\sigma_{\alpha_k}^{(k)})O}\Big)\Bigg),
\end{eqnarray}
where $O$ is acting on $(1,\dots, k)$ qubits.
Let $O=\sigma_{\xi_1}^{(1)}\otimes\dots\otimes \sigma_{\xi_k}^{(k)}$, holds. Then we can rewrite \eqref{0_15_10_2} as follows
\begin{eqnarray}\label{0_15_10_22}
&&{B}^{(1,\dots,k)}_{\alpha_1,\dots,\alpha_k}(\rho_{\tau_i,\tau_j}^{(i,j)},\sigma_{\xi_1}^{(1)}\otimes\dots\otimes \sigma_{\xi_k}^{(k)})\equiv-iTr\left([\sigma_{\alpha_1}^{(1)}\otimes\dots\otimes \sigma_{\alpha_k}^{(k)},\rho_{\tau_i,\tau_j}^{(i,j)}](\sigma_{\xi_1}^{(1)}\otimes\dots\otimes \sigma_{\xi_k}^{(k)})\right)\nonumber\\
&=&2^{k-1}(\varepsilon_{\alpha_i\tau_i \gamma}\delta_{\alpha_1\xi_1}\dots\delta_{\gamma\xi_i}\dots\delta_{\alpha_j\xi_j}\dots\delta_{\alpha_k\xi_k}+
\varepsilon_{\alpha_j\tau_j \eta}\delta_{\alpha_1\xi_1}\dots\delta_{\alpha_i\xi_i}\dots \delta_{\eta\xi_j}\dots\delta_{\alpha_k\xi_k})\\\nonumber
&=&\begin{cases}
2^{k-1}(\varepsilon_{\alpha_i\tau_i \gamma}\delta_{\eta\alpha_j}+\varepsilon_{\alpha_j\tau_j\eta}\delta_{\gamma\alpha_i}) &\text{if  $\xi_i=\gamma$, $\xi_j=\eta$, $\xi_{1,\dots,k}=\alpha_{1,\dots,k}$,}\\
0&\text{else}\\
 \end{cases}.
\end{eqnarray}
From this result, for $k=3$ we get
\begin{eqnarray}\label{8_30}
{B}^{(i,j,l)}_{\alpha_i\alpha_j\alpha_l}(\rho_{\tau_i,\tau_j}^{(i,j)},\sigma_{\xi_i}^{(i)}\otimes\sigma_{\xi_j}^{(j)}\otimes \sigma_{\xi_l}^{(l)})=2(\varepsilon_{\alpha_i\tau_i\gamma}\delta_{\gamma,\xi_i}\delta_{\alpha_j\xi_j}+\varepsilon_{\alpha_j\tau_j\eta}\delta_{\alpha_i\xi_i}\delta_{\eta\xi_j})\delta_{\alpha_l\xi_l}.
\end{eqnarray}
\par To isolate $\{D_{\mu\nu}^{(m)}\}$ we recall the assumption that   $\rho_{N-2}$ is the density matrix of the maximally mixed state.
Then the Lindblad part of the equation \eqref{10_50_2_1116} for an observable $O=\sigma_{\xi_i}^{(i)}\otimes\sigma_{\xi_j}^{(j)}$ is 
\begin{eqnarray}\label{0_10_10_1321}
\mathcal{D}^{k}_{\mu\nu}(\rho_{\tau_i,\tau_j}^{(i,j)},\sigma_{\xi_i}^{(i)}\otimes\sigma_{\xi_j}^{(j)})&\equiv&D^{(k)}_{\mu\nu}\tr{
\left(\left(\sigma_{\mu}^{(k)}\rho_{\tau_i,\tau_j}^{(i,j)}{\sigma_{\nu}^{(k)}}-\frac{1}{2}\{{\sigma_{\nu}^{(k)}}\sigma_{\mu}^{(k)},\rho_{\tau_i,\tau_j}^{(i,j)}\}\right)\right)(\sigma_{\xi_i}^{(i)}\otimes\sigma_{\xi_j}^{(j)})}\\&=&\nonumber
\begin{cases}D^{(i)}_{\mu\nu}
   (2i\varepsilon_{\mu\nu\gamma}\delta_{\gamma\xi_i}+2\delta_{\nu\tau_{i}}\delta_{\mu\xi_i}-\frac{3}{2}\delta_{\mu\nu}\delta_{\xi_i\tau_{i}})\delta_{\tau_j\xi_j}&\text{if $k=i$}\\
D^{(j)}_{\mu\nu}  (2i\varepsilon_{\mu\nu\gamma}\delta_{\gamma\xi_j}+2\delta_{\nu\tau_{j}}\delta_{\mu\xi_j}-\frac{3}{2}\delta_{\mu\nu}\delta_{\xi_j\tau_{j}})\delta_{\tau_i\xi_i} &\text{if $k=j$}\\
  0&\text{else}\\
 \end{cases}.
\end{eqnarray}
Let us substitute the conditions  \eqref{0_15_1} in \eqref{0_10_10_1321}. We get the following results:
\begin{eqnarray}\label{0_10_10_98}
\mathcal{D}^{i}_{\mu\nu}(\rho_{\tau_i,\tau_j}^{(i,j)},\sigma_{\xi_i}^{(i)}\otimes\sigma_{\xi_j}^{(j)})=
D^{(i)}_{\mu\nu}
   (2i\varepsilon_{\mu\nu\gamma}+2\delta_{\nu\tau_{i}}\delta_{\mu\gamma}-\frac{3}{2}\delta_{\mu\nu}\delta_{\gamma\tau_{i}})\delta_{\tau_j\eta}\quad \text{if $\xi_i=\gamma$, $\xi_j=\eta$,}
\end{eqnarray}
and 
\begin{eqnarray}\label{0_10_10_99}
\mathcal{D}^{j}_{\mu\nu}(\rho_{\tau_i,\tau_j}^{(i,j)},\sigma_{\xi_i}^{(i)}\otimes\sigma_{\xi_j}^{(j)})=
D^{(j)}_{\mu\nu}  (2i\varepsilon_{\mu\nu\gamma}\delta_{\gamma\eta}+2\delta_{\nu\tau_{j}}\delta_{\mu\eta}-\frac{3}{2}\delta_{\mu\nu}\delta_{\eta\tau_{j}})\delta_{\tau_i\gamma}\quad \text{if $\xi_i=\gamma$, $\xi_j=\eta$}.
\end{eqnarray}
However, for an observable $O=\sigma_{\xi_i}^{(i)}$ the Lindblad part of the equation \eqref{10_50_2_1116} is the following 
\begin{eqnarray}\label{0_10_10_1002}
\mathcal{D}^{k}_{\mu\nu}(\rho_{\tau_i,\tau_j}^{(i,j)},\sigma_{\xi_i}^{(i)})&\equiv&D^{(k)}_{\mu\nu}\tr{
\left(\left(\sigma_{\mu}^{(k)}\rho_{\tau_i,\tau_j}^{(i,j)}{\sigma_{\nu}^{(k)}}-\frac{1}{2}\{{\sigma_{\nu}^{(k)}}\sigma_{\mu}^{(k)},\rho_{\tau_i,\tau_j}^{(i,j)}\}\right)\right)(\sigma_{\xi_i}^{(i)})}\\\nonumber&=&
\begin{cases}D^{(i)}_{\mu\nu}
(2i\varepsilon_{\mu\nu\gamma}\delta_{\gamma\xi_i}+\delta_{\nu\tau_{i}}\delta_{\mu\xi_i}-\frac{3}{2}\delta_{\mu\nu}\delta_{\xi_i\tau_{i}}+\delta_{\mu\xi_i}\delta_{\tau_{i}\nu}) &\text{if $k=i$}\\
0 &\text{if $k=j$}\\
  0&\text{else}\\
 \end{cases}.
\end{eqnarray}
Substituting the conditions  \eqref{0_15_2} in \eqref{0_10_10_1002}, we get
\begin{eqnarray}\label{0_10_10_97}
\mathcal{D}^{i}_{\mu\nu}(\rho_{\tau_i,\tau_j}^{(i,j)},\sigma_{\xi_i}^{(i)})&=&
D^{(i)}_{\mu\nu}(2i\varepsilon_{\mu\nu\gamma}+2\delta_{\nu\tau_{i}}\delta_{\mu\gamma}-\frac{3}{2}\delta_{\mu\nu}\delta_{\gamma\tau_{i}}),\quad \text{if  $\xi_i=\gamma$}.
 \end{eqnarray}
Next, for an observable $O=\sigma_{\xi_i}^{(i)}\otimes\sigma_{\xi_j}^{(j)}\otimes\sigma_{\xi_l}^{(l)}$ we can write
\begin{eqnarray}\label{0_10_10_102}
\mathcal{D}^{k}_{\mu\nu}(\rho_{\tau_i,\tau_j}^{(i,j)},\sigma_{\xi_i}^{(i)}\otimes \sigma_{\xi_j}^{(j)}\otimes \sigma_{\xi_l}^{(l)})&\equiv&D^{(k)}_{\mu\nu}\tr{
\left(\left(\sigma_{\mu}^{(k)}\rho_{\tau_i,\tau_j}^{(i,j)}{\sigma_{\nu}^{(k)}}-\frac{1}{2}\{{\sigma_{\nu}^{(k)}}\sigma_{\mu}^{(k)},\rho_{\tau_i,\tau_j}^{(i,j)}\}\right)\right)(\sigma_{\xi_i}^{(i)}\otimes\sigma_{\xi_j}^{(j)}\otimes\sigma_{\xi_l}^{(l)})}\nonumber\\&=&
\begin{cases}
 2iD^{(l)}_{\mu\nu} \varepsilon_{\mu\nu\gamma}\delta_{\gamma\xi_l}\delta_{\tau_i\xi_i}\delta_{\tau_j\xi_j}&\text{ if $k=l$}\\
 0 &\text{else}.
 \end{cases}.
\end{eqnarray}
Then, according to the results of the previous subsection, we get
\begin{eqnarray}\label{0_10_10_104}
\mathcal{D}^{l}_{\mu\nu}(\rho_{\tau_i,\tau_j}^{(i,j)},\sigma_{\xi_i}^{(i)}\otimes \sigma_{\xi_j}^{(j)}\otimes \sigma_{\xi_l}^{(l)})=
2i D^{(l)}_{\mu\nu} \varepsilon_{\mu\nu\gamma}\delta_{\gamma\alpha_l}\delta_{\tau_i\gamma}\delta_{\tau_j\eta}\quad 
 \text{if $\xi_i=\gamma$, $\xi_j=\eta$, $\xi_l=\alpha_l$}.
\end{eqnarray}
Finally, for an observable $O=\sigma_{\xi_1}^{(1)}\otimes\dots\sigma_{\xi_i}^{(i)}\dots\otimes\sigma_{\xi_j}^{(j)}\dots\otimes  \sigma_{\xi_k}^{(k)}$, $k>3$, the Lindblad part of the equation \eqref{10_50_2_1116} is the following 
\begin{eqnarray}\label{0_10_10_103}
\!\!\!\!\!\!\!\!\!\!\!\!D^{(m)}_{\mu\nu}\tr{
\left(\left(\sigma_{\mu}^{(m)}\rho_{\tau_i,\tau_j}^{(i,j)}{\sigma_{\nu}^{(m)}}-\frac{1}{2}\{{\sigma_{\nu}^{(m)}}\sigma_{\mu}^{(m)},\rho_{\tau_i,\tau_j}^{(i,j)}\}\right)\right)(\sigma_{\xi_1}^{(1)}\otimes\dots\otimes\sigma_{\xi_i}^{(i)}\dots\otimes\sigma_{\xi_j}^{(j)}\dots\otimes  \sigma_{\xi_k}^{(k)})}=0.
\end{eqnarray}
\subsection{Final Results} 
After we selected the different observable operators and defined the  density matrix $\rho_{0}=\rho_{\tau_i}^{(i)}\otimes \rho_{\tau_j}^{(j)}\otimes \frac{I^{2n-2}}{2^{2n-2}}$, where the $i$ and $j$ qubits are in the Pauli states,
we are ready to isolate the desired coefficients.
For an observable $O=\sigma_{\gamma}^{(i)}\otimes \sigma_{\eta}^{(j)}$ we can write   the equation \eqref{10_50_2_1116} as follows
\begin{eqnarray}\label{10_50_2_1115}
&&\frac{d}{dt}\tr{\rho_t (\sigma_{\gamma}^{(i)}\otimes \sigma_{\eta}^{(j)})}|_{t=0}=-2\sum\limits_{\alpha_1} a_{\alpha_1}^{(i)}
\varepsilon_{\alpha_1\tau_i \gamma}\delta_{\eta\tau_j}-2\sum\limits_{\alpha_1}a_{\alpha_1}^{(j)}\varepsilon_{\alpha_1\tau_j \eta}\delta_{\gamma\tau_i}
\\\nonumber&-&2\sum\limits_{\alpha_1,\alpha_2} a_{\alpha_1\alpha_2}^{(ij)}
\left(\varepsilon_{\alpha_1\tau_i\gamma}\delta_{\eta\alpha_2}+\varepsilon_{\alpha_2\tau_j\eta}\delta_{\gamma\alpha_1}\right)+
\sum\limits_{\mu,\nu}\Big(D^{(i)}_{\mu\nu}
   (2i\varepsilon_{\mu\nu\gamma}+2\delta_{\nu\tau_{i}}\delta_{\mu\gamma}-\frac{3}{2}\delta_{\mu\nu}\delta_{\gamma\tau_{i}})\delta_{\tau_j\eta}\\\nonumber&+&D^{(j)}_{\mu\nu}  (2i\varepsilon_{\mu\nu\gamma}\delta_{\gamma\eta}+2\delta_{\nu\tau_{j}}\delta_{\mu\eta}-\frac{3}{2}\delta_{\mu\nu}\delta_{\eta\tau_{j}})\delta_{\tau_i\gamma}\Big).
\end{eqnarray}
Selecting $\tau_j\neq\eta$, $\tau_i\neq \gamma$, we can isolate the coefficients of the type $a_{\alpha_1\alpha_2}^{(ij)}$ in \eqref{10_50_2_1115}, namely 
\begin{equation}\label{10_50_2_3}
\frac{d}{dt}\tr{\rho_t (\sigma_{\gamma}^{(i)}\otimes \sigma_{\eta}^{(j)})}|_{t=0}=
-2\sum\limits_{\alpha_1,\alpha_2} a_{\alpha_1\alpha_2}^{(ij)}\left(\varepsilon_{\alpha_1\tau_i\gamma}\delta_{\eta\alpha_2}
+\varepsilon_{\alpha_2\tau_j\eta}\delta_{\gamma\alpha_1}\right).
\end{equation}
Let us call $\rho_{t \tau_i \tau_j}$ the density matrix evaluated by the Hamiltonian evolution from $\rho_0$.
From \eqref{10_50_2_3} we can find $a_{\alpha_1\alpha_2}^{(ij)}$. To this end, we select $\gamma=y$, $\eta=y$ and four pairs  $\tau_i,\tau_j\in\{z,x; z, z; x,x; x,z\}$  to get the system of equations
\begin{eqnarray}
\label{9204}
\frac{d}{dt}\tr{\rho_{tzx} (\sigma_{ y}^{(i)}\otimes \sigma_{ y}^{(j)})}|_{t=0}&=&
2(a_{x y}^{(ij)}-a_{yz}^{(ij)}),
\frac{d}{dt}\tr{\rho_{tzz} (\sigma_{ y}^{(i)}\otimes \sigma_{ y}^{(j)})}|_{t=0}=
2(a_{x y}^{(ij)}
+a_{yx}^{(ij)}),\\\nonumber
\frac{d}{dt}\tr{\rho_{txx} (\sigma_{ y}^{(i)}\otimes \sigma_{ y}^{(j)})}|_{t=0}&=&
-2(a_{yz}^{(ij)}+a_{zy}^{(ij)}),
\frac{d}{dt}\tr{\rho_{txz} (\sigma_{ y}^{(i)}\otimes \sigma_{ y}^{(j)})}|_{t=0}=
-2(a_{yx}^{(ij)}+a_{zy}^{(ij)}).
\end{eqnarray}
Since $a_{yz}^{(ij)}=a_{zy}^{(ij)}$ and $a_{x y}^{(ij)}=a_{yx}^{(ij)}$, we can write
\begin{equation}
\label{9204_1}
a_{x y}^{(ij)}=
\frac{1}{4}\frac{d}{dt}\tr{\rho_{tzz} (\sigma_{ y}^{(i)}\otimes \sigma_{ y}^{(j)})}|_{t=0},
\quad 
a_{yz}^{(ij)}=-\frac{1}{4}\frac{d}{dt}\tr{\rho_{txx} (\sigma_{ y}^{(i)}\otimes \sigma_{ y}^{(j)})}|_{t=0}.
\end{equation}
Selecting $\gamma=x$, $\eta=y$ and four pairs  $\tau_i,\tau_j\in\{y,z; y,x; z,z; z,x\}$  to get the system of equations
\begin{eqnarray}
\label{9201}
\frac{d}{dt}\tr{\rho_{tyz} (\sigma_{ x}^{(i)}\otimes \sigma_{ y}^{(j)})}|_{t=0}&=&
2(a_{xx}^{(ij)}+a_{zy}^{(ij)}
),
\frac{d}{dt}\tr{\rho_{tyx} (\sigma_{ x}^{(i)}\otimes \sigma_{ y}^{(j)})}|_{t=0}=
2(a_{zy}^{(ij)}-a_{xz}^{(ij)}
),\\\nonumber
\frac{d}{dt}\tr{\rho_{tzz} (\sigma_{ x}^{(i)}\otimes \sigma_{ y}^{(j)})}|_{t=0}&=&
2(a_{xx}^{(ij)}-a_{yy}^{(ij)}
),
\frac{d}{dt}\tr{\rho_{tzx} (\sigma_{ x}^{(i)}\otimes \sigma_{ y}^{(j)})}|_{t=0}=
-2(a_{yy}^{(ij)}+a_{xz}^{(ij)}
).
\end{eqnarray}
Hence
\begin{eqnarray}
\label{9206}
a_{xx}^{(ij)}&=&\frac{1}{2}\frac{d}{dt}\tr{\rho_{tyz} (\sigma_{ x}^{(i)}\otimes \sigma_{ y}^{(j)})}|_{t=0}+\frac{1}{4}\frac{d}{dt}\tr{\rho_{txx} (\sigma_{ y}^{(i)}\otimes \sigma_{ y}^{(j)})}|_{t=0},\\\nonumber
a_{xz}^{(ij)}&=&-\frac{1}{2}\frac{d}{dt}\tr{\rho_{tyx} (\sigma_{ x}^{(i)}\otimes \sigma_{ y}^{(j)})}|_{t=0}-\frac{1}{4}\frac{d}{dt}\tr{\rho_{txx} (\sigma_{ y}^{(i)}\otimes \sigma_{ y}^{(j)})}|_{t=0},\\\nonumber
a_{yy}^{(ij)}&=&-\frac{1}{2}\frac{d}{dt}\tr{\rho_{tzz} (\sigma_{ x}^{(i)}\otimes \sigma_{ y}^{(j)})}|_{t=0}+\frac{1}{2}\frac{d}{dt}\tr{\rho_{tyz} (\sigma_{ x}^{(i)}\otimes \sigma_{ y}^{(j)})}|_{t=0}\\\nonumber
&+&\frac{1}{4}\frac{d}{dt}\tr{\rho_{txx} (\sigma_{ y}^{(i)}\otimes \sigma_{ y}^{(j)})}|_{t=0}.
\end{eqnarray}
Selecting $\gamma=x$, $\eta=z$ and  four pairs  $\tau_i,\tau_j\in\{y,y; z,x; y,x; z,y\}$, we get
\begin{eqnarray}
\label{9202}
\frac{d}{dt}\tr{\rho_{tyy} (\sigma_{ x}^{(i)}\otimes \sigma_{ z}^{(j)})}|_{t=0}&=&
-2(a_{xx}^{(ij)}-a_{zz}^{(ij)}),
\frac{d}{dt}\tr{\rho_{tzx}(\sigma_{ x}^{(i)}\otimes \sigma_{ z}^{(j)})}|_{t=0}=
2(a_{x y}^{(ij)}-a_{yz}^{(ij)}),\\\nonumber
\frac{d}{dt}\tr{\rho_{tyx}(\sigma_{ x}^{(i)}\otimes \sigma_{ z}^{(j)})}|_{t=0}&=&
-2(a_{zz}^{(ij)}-a_{x y}^{(ij)}),
\frac{d}{dt}\tr{\rho_{tzy} (\sigma_{ x}^{(i)}\otimes \sigma_{ z}^{(j)})}|_{t=0}=
-2(a_{xx}^{(ij)}+a_{yz}^{(ij)}).
\end{eqnarray}
Hence the last coefficient is
\begin{equation}
\label{9208}
a_{zz}^{(ij)}=-\frac{1}{2}\frac{d}{dt}\tr{\rho_{tyx}(\sigma_{ x}^{(i)}\otimes \sigma_{ z}^{(j)})}|_{t=0}+\frac{1}{4}\frac{d}{dt}\tr{\rho_{tzz} (\sigma_{ y}^{(i)}\otimes \sigma_{ y}^{(j)})}|_{t=0}.
\end{equation}
 To find the other coefficients we select $\tau_j=\eta$, $\tau_i\neq \gamma$ and  rerewrite \eqref{10_50_2_1115} as 
\begin{eqnarray}\label{14_35}
\frac{d}{dt}\tr{\rho_t (\sigma_{\gamma}^{(i)}\otimes \sigma_{\eta}^{(j)})}|_{t=0}&=&-2\sum\limits_{\alpha_1} a_{\alpha_1}^{(i)}
\varepsilon_{\alpha_1\tau_i \gamma}- 2\sum\limits_{\alpha_1,\alpha_2} a_{\alpha_1\alpha_2}^{(ij)}
\varepsilon_{\alpha_1\tau_i\gamma}\delta_{\eta\alpha_2}\\\nonumber
&+&2\sum\limits_{\mu=x,y,z}D^{(i)}_{\mu\mu}\delta_{\mu\tau_{i}}\delta_{\mu\gamma}+
2\sum\limits_{\mu,\nu=x,y,z}^{\mu\neq \nu}
D^{(i)}_{\mu\nu}
   (i\varepsilon_{\mu\nu\gamma}+\delta_{\nu\tau_{i}}\delta_{\mu\gamma}).
\end{eqnarray}
Next, for $\tau_j\neq\eta$, $\tau_i=\gamma$, we can rewrite \eqref{10_50_2_1115} as 
\begin{eqnarray}\label{10_50_2_1150}
\frac{d}{dt}\tr{\rho_t (\sigma_{\gamma}^{(i)}\otimes \sigma_{\eta}^{(j)})}|_{t=0}&=&-2\sum\limits_{\alpha_1}a_{\alpha_1}^{(j)}\varepsilon_{\alpha_1\tau_j \eta}-2\sum\limits_{\alpha_1,\alpha_2} a_{\alpha_1\alpha_2}^{(ij)}
\varepsilon_{\alpha_2\tau_j\eta}\delta_{\gamma\alpha_1}\nonumber\\&+&
2\sum\limits_{\mu=x,y,z}D^{(i)}_{\mu\mu}\delta_{\mu\tau_{j}}\delta_{\mu\eta}+
2\sum\limits_{\mu,\nu=x,y,z}^{\mu\neq \nu}
D^{(i)}_{\mu\nu}  (i\varepsilon_{\mu\nu\gamma}\delta_{\gamma\eta}+\delta_{\nu\tau_{j}}\delta_{\mu\eta}).
\end{eqnarray}
Selecting an observable $O=\sigma_{\gamma}^{(i)}$, we can write 
\begin{eqnarray}\label{14_35_3}
\frac{d}{dt}\tr{\rho_t \sigma_{\gamma}^{(i)}}|_{t=0}&=&-\sum\limits_{\alpha_1} a_{\alpha_1}^{(i)}\varepsilon_{\alpha_1\tau_i\gamma}
-2\sum\limits_{\alpha_1,\alpha_2} a_{\alpha_1\alpha_2}^{(ij)}
\varepsilon_{\alpha_1\tau_i\gamma}\delta_{\alpha_2\tau_j}\\\nonumber
&+&2\sum\limits_{\mu=x,y,z}D^{(i)}_{\mu\mu}(\delta_{\mu\tau_{i}}\delta_{\mu\gamma}-\frac{3}{4}\delta_{\gamma\tau_{i}})+
2\sum\limits_{\mu,\nu=x,y,z}^{\mu\neq \nu}D^{(i)}_{\mu\nu}(i\varepsilon_{\mu\nu\gamma}+\delta_{\nu\tau_{i}}\delta_{\mu\gamma}).
\end{eqnarray}
For the other observable $O=\sigma_{\eta}^{(j)}$, the result is the following
\begin{eqnarray}\label{10_50_2_1442}
\frac{d}{dt}\tr{\rho_t \sigma_{\eta}^{(j)})}|_{t=0}&=&-\sum\limits_{\alpha_1}a_{\alpha_1}^{(j)}\varepsilon_{\alpha_1\tau_j \eta}
-2\sum\limits_{\alpha_1,\alpha_2} a_{\alpha_1\alpha_2}^{(ij)}\varepsilon_{\alpha_2\tau_j\eta}\delta_{\alpha_1\tau_i}\nonumber\\&+&
2\sum\limits_{\mu=x,y,z}D^{(j)}_{\mu\mu} (\delta_{\mu\tau_{j}}\delta_{\mu\eta}-\frac{3}{4}\delta_{\eta\tau_{j}})+2
\sum\limits_{\mu,\nu=x,y,z}^{\mu\neq \nu}D^{(j)}_{\mu\nu}  (i\varepsilon_{\mu\nu\gamma}\delta_{\gamma\eta}+\delta_{\nu\tau_{j}}\delta_{\mu\eta}).
\end{eqnarray}
Substituting \eqref{14_35_3} in \eqref{14_35}, we can write
\begin{eqnarray}\label{14_35_2118}
&&\frac{d}{dt}\tr{\rho_t (\sigma_{\gamma}^{(i)}\otimes \sigma_{\eta}^{(j)})}|_{t=0}-\frac{d}{dt}\tr{\rho_t \sigma_{\gamma}^{(i)}}|_{t=0}=-\sum\limits_{\alpha_1} a_{\alpha_1}^{(i)}
\varepsilon_{\alpha_1\tau_i \gamma}\\\nonumber
&+& 2\sum\limits_{\alpha_1,\alpha_2} a_{\alpha_1\alpha_2}^{(ij)}
\varepsilon_{\alpha_1\tau_i\gamma}(\delta_{\tau_j\alpha_2}-\delta_{\eta\alpha_2})
+\frac{3}{2}\sum\limits_{\mu=x,y,z}D^{(i)}_{\mu\mu}\delta_{\gamma\tau_{i}}.
\end{eqnarray}
Since in \eqref{14_35} the conditions $\tau_j=\eta$, $\tau_i\neq \gamma$, hold, we can rewrite the latter equation as
\begin{equation}\label{14_35_1024}
\frac{d}{dt}\tr{\rho_t (\sigma_{\gamma}^{(i)}\otimes \sigma_{\eta}^{(j)})}|_{t=0}-\frac{d}{dt}\tr{\rho_t \sigma_{\gamma}^{(i)}}|_{t=0}=-\sum\limits_{\alpha_1} a_{\alpha_1}^{(i)}
\varepsilon_{\alpha_1\tau_i \gamma}.
\end{equation}
Solving the latter  equation, we find  $a_{\alpha_1}^{(i)}$. To this end, we select
 $\gamma=y$, $\tau_i=x$, $\tau_j=\eta\in\{x,y,z\}$ and get
\begin{equation}\label{1747_1}
a_{x}^{(i)}=-\frac{d}{dt}\tr{\rho_{tx\eta} (\sigma_{y}^{(i)}\otimes \sigma_{\eta}^{(j)})}|_{t=0}+\frac{d}{dt}\tr{\rho_{tx\eta} \sigma_{y}^{(i)}}|_{t=0}.
\end{equation}
Next, for  $\gamma=x$, $\tau_i=z$, $\tau_j=\eta\in\{x,y,z\}$ we get the solution of \eqref{14_35_1024}, namely
\begin{equation}\label{1747_2}
a_{y}^{(i)}=-\frac{d}{dt}\tr{\rho_{tz\eta} (\sigma_{x}^{(i)}\otimes \sigma_{\eta}^{(j)})}|_{t=0}+\frac{d}{dt}\tr{\rho_{tz\eta} \sigma_{x}^{(i)}}|_{t=0}.
\end{equation}
Finally, for  $\gamma=x$, $\tau_i=y$, $\tau_j=\eta\in\{x,y,z\}$ the solution is
\begin{equation}\label{1747_3} 
a_{z}^{(i)}=\frac{d}{dt}\tr{\rho_{ty\eta} (\sigma_{x}^{(i)}\otimes \sigma_{\eta}^{(j)})}|_{t=0}-\frac{d}{dt}\tr{\rho_{ty\eta}\sigma_{x}^{(i)}}|_{t=0}.
\end{equation}
\par Substituting \eqref{10_50_2_1442} in \eqref{10_50_2_1150}, we can write
\begin{eqnarray}\label{10_50_2_1150_1}
&&\frac{d}{dt}\tr{\rho_t (\sigma_{\gamma}^{(i)}\otimes \sigma_{\eta}^{(j)})}|_{t=0}-\frac{d}{dt}\tr{\rho_t \sigma_{\eta}^{(j)})}|_{t=0}=
-\sum\limits_{\alpha_1}a_{\alpha_1}^{(j)}\varepsilon_{\alpha_1\tau_j \eta}\nonumber\\&+&2\sum\limits_{\alpha_1,\alpha_2} a_{\alpha_1\alpha_2}^{(ij)}
\varepsilon_{\alpha_2\tau_j\eta}(\delta_{\tau_i\alpha_1}-\delta_{\gamma\alpha_1})+
\frac{3}{2}\sum\limits_{\mu=x,y,z}D^{(i)}_{\mu\mu}\delta_{\eta\tau_{j}}.
\end{eqnarray}
Since $\tau_j\neq\eta$, $\tau_i=\gamma$, hold, we can rewrite it as follows
\begin{eqnarray}\label{10_50_2_1150_2}
\frac{d}{dt}\tr{\rho_t (\sigma_{\gamma}^{(i)}\otimes \sigma_{\eta}^{(j)})}|_{t=0}-\frac{d}{dt}\tr{\rho_t \sigma_{\eta}^{(j)})}|_{t=0}=
-\sum\limits_{\alpha_1}a_{\alpha_1}^{(j)}\varepsilon_{\alpha_1\tau_j \eta}.
\end{eqnarray}
Solving the latter  equation, we find  $a_{\alpha_1}^{(j)}$.
Selecting $\eta=y$, $\tau_j=z$ and $\tau_i=\gamma\in\{x,y,z\}$, we get
\begin{equation}\label{1729}
a_{x}^{(j)}=\frac{d}{dt}\tr{\rho_{t\gamma z}(\sigma_{\gamma}^{(i)}\otimes \sigma_{y}^{(j)})}|_{t=0}-\frac{d}{dt}\tr{\rho_{t\gamma z} \sigma_{y}^{(j)})}|_{t=0}.
\end{equation}
Selecting $\eta=x$, $\tau_j=z$ and $\tau_i=\gamma\in\{x,y,z\}$, we get
\begin{equation}\label{1730}
a_{y}^{(j)}=-\frac{d}{dt}\tr{\rho_{t\gamma z} (\sigma_{\gamma}^{(i)}\otimes \sigma_{x}^{(j)})}|_{t=0}+\frac{d}{dt}\tr{\rho_{t\gamma z} \sigma_{x}^{(j)})}|_{t=0}.
\end{equation}
Finally, selecting $\eta=y$, $\tau_j=x$ and $\tau_i=\gamma\in\{x,y,z\}$, we get the last coefficient of this type
\begin{equation}\label{1731}
a_{z}^{(j)}=-\frac{d}{dt}\tr{\rho_{t\gamma x} (\sigma_{\gamma}^{(i)}\otimes \sigma_{y}^{(j)})}|_{t=0}+\frac{d}{dt}\tr{\rho_{t\gamma x} \sigma_{y}^{(j)})}|_{t=0}.
\end{equation}
All the coefficients with the corresponding observables and initial states are given in Table~\ref{tab_1}.
\begin{table}[h!]
\centering
\begin{tabular}{|c| c| c|}
 \hline
$a_{\alpha}^{(i)}$ & \{$O$, $\rho_{\tau_i,\tau_j}^{(i,j)}$\} & Equation  \\ [0.5ex] 
 \hline
$a_{x}^{(i)}$& $\{\sigma_{y}^{(i)}\otimes \sigma_{\eta}^{(j)},\rho_{x,\eta}^{(i,j)}\}$; $\{\sigma_{y}^{(i)},\rho_{x,\eta}^{(i,j)})\}$,  $\forall\eta\in\{x,y,z\}$ & \eqref{1747_1} \\
$a_{y}^{(i)}$ &  $ \{\sigma_{x}^{(i)}\otimes \sigma_{\eta}^{(j)},\rho_{z,\eta}^{(i,j)}\}$; $ \{\sigma_{x}^{(i)},\rho_{z,\eta}^{(i,j)}\}$, $\forall\eta\in\{x,y,z\}$& \eqref{1747_2}\\
$a_{z}^{(i)}$ &  $ \{\sigma_{x}^{(i)}\otimes \sigma_{\eta}^{(j)},\rho_{y,\eta}^{(i,j)}\}$; $ \{\sigma_{x}^{(i)},\rho_{y,\eta}^{(i,j)}\}$, $\forall\eta\in\{x,y,z\}$& \eqref{1747_3}\\
$a_{x}^{(j)}$& $ \{\sigma_{\gamma}^{(i)}\otimes \sigma_{y}^{(j)},\rho_{\gamma,z}^{(i,j)}\}$; $ \{\sigma_{y}^{(i)},\rho_{\gamma,z}^{(i,j)}\}$, $\forall\gamma\in\{x,y,z\}$& \eqref{1729}\\
$a_{y}^{(j)}$ &  $ \{\sigma_{\gamma}^{(i)}\otimes \sigma_{x}^{(j)},\rho_{\gamma,z}^{(i,j)}\}$; $ \{\sigma_{x}^{(i)},\rho_{\gamma,z}^{(i,j)}\}$, $\forall\gamma\in\{x,y,z\}$& \eqref{1730}\\ 
$a_{z}^{(j)}$ &  $ \{\sigma_{\gamma}^{(i)}\otimes \sigma_{y}^{(j)},\rho_{\gamma,x}^{(i,j)}\}$; $ \{\sigma_{y}^{(i)},\rho_{\gamma,x}^{(i,j)}\}$, $\forall\gamma\in\{x,y,z\}$& \eqref{1731}\\
$a_{xx}^{(ij)}$&  $ \{\sigma_{x}^{(i)}\otimes \sigma_{y}^{(j)},\rho_{y,z}^{(i,j)}\}$; $ \{\sigma_{y}^{(i)}\otimes \sigma_{y}^{(j)},\rho_{x,x}^{(i,j)}\}$ &\eqref{9206}\\
$a_{yy}^{(ij)}$ &  $ \{\sigma_{x}^{(i)}\otimes \sigma_{y}^{(j)},\rho_{z,z}^{(i,j)}\}$; $ \{\sigma_{x}^{(i)}\otimes \sigma_{y}^{(j)},\rho_{y,z}^{(i,j)}\}$; $ \{\sigma_{y}^{(i)}\otimes \sigma_{y}^{(j)},\rho_{x,x}^{(i,j)}\}$  &\eqref{9206}\\
$a_{zz}^{(ij)}$ &  $ \{\sigma_{x}^{(i)}\otimes \sigma_{z}^{(j)},\rho_{y,x}^{(i,j)}\}$; $ \{\sigma_{y}^{(i)}\otimes \sigma_{y}^{(j)},\rho_{z,z}^{(i,j)}\}$ &\eqref{9208}\\
$a_{xy}^{(ij)}$&  $ \{\sigma_{y}^{(i)}\otimes \sigma_{y}^{(j)},\rho_{z,z}^{(i,j)}\}$ &\eqref{9204_1}\\
$a_{yz}^{(ij)}$ & $ \{\sigma_{y}^{(i)}\otimes \sigma_{y}^{(j)},\rho_{x,x}^{(i,j)}\}$ &\eqref{9204_1}\\
$a_{xz}^{(ij)}$ & $ \{\sigma_{x}^{(i)}\otimes \sigma_{y}^{(j)},\rho_{y,x}^{(i,j)}\}$; $ \{\sigma_{y}^{(i)}\otimes \sigma_{y}^{(j)},\rho_{x,x}^{(i,j)}\}$ &\eqref{9206}\\
\hline
\end{tabular}
\caption{The first column represents the type of the estimated Hamiltonian parameters $a_{\alpha_i}^{(i)}$, $a_{\alpha_i,\alpha_j}^{(ij)}$, ${\alpha_i,\alpha_j}\in\{x,y,z\}$. In the third column the number of  equation for every parameter is provided, depending from the pairs of the observable $O$ and the initial state $\rho_{\tau_i,\tau_j}^{(i,j)}=\rho_{\tau_i}^{(i)}\otimes \rho_{\tau_j}^{(j)}\otimes \rho_{N-2}$, ${\tau_i,\tau_j}=\{x,y,z\}$, given in the second column.
For example, to estimate $a_{x}^{(i)}$ we use \eqref{1747_1} with two pairs of observables and states: $\{O, \rho_{\tau_i,\tau_j}^{(i,j)}\}=\{(\sigma_{y}^{(i)}\otimes \sigma_{\eta}^{(j)},\rho_{x,\eta}^{(i,j)});(\sigma_{y}^{(i)},\rho_{x,\eta}^{(i,j)})\}$.}
\label{tab_1}
\end{table}
\par  For an observable $O=\sigma_{\gamma}^{(i)}\otimes\sigma_{\eta}^{(j)}\otimes\sigma_{\alpha_l}^{(l)}$ we can write the equation
\begin{eqnarray}\label{10_50_2_4}
\frac{d}{dt}\tr{\rho_t (\sigma_{\gamma}^{(i)}\otimes \sigma_{\eta}^{(j)}\otimes\sigma_{\alpha_l}^{(l)})}|_{t=0}&=&
-4\sum\limits_{\alpha_i\alpha_j\alpha_l} a_{\alpha_i\alpha_j\alpha_l}^{(ijl)}
(\varepsilon_{\alpha_i\tau_i \gamma}\delta_{\eta\alpha_j}+\varepsilon_{\alpha_j\tau_j\eta}\delta_{\gamma\alpha_i})
\\\nonumber&+&2i
\sum\limits_{\mu,\nu}D^{(l)}_{\mu\nu} \varepsilon_{\mu\nu\gamma}\delta_{\gamma\alpha_l}\delta_{\tau_i\gamma}\delta_{\tau_j\eta}.
\end{eqnarray}
Selecting $\tau_j=\eta$, $\tau_i\neq \gamma$, we can rewrite \eqref{10_50_2_4} as
\begin{equation}\label{10_50_2_5}
\frac{d}{dt}\tr{\rho_t (\sigma_{\gamma}^{(i)}\otimes \sigma_{\eta}^{(j)}\otimes\sigma_{\alpha_l}^{(l)})}|_{t=0}=
-4\sum\limits_{\alpha_i\alpha_j\alpha_l} a_{\alpha_i\alpha_j\alpha_l}^{(ijl)}\varepsilon_{\alpha_i\tau_i \gamma}\delta_{\eta\alpha_j}.
\end{equation}
From this equation we can find $a_{\alpha_i\alpha_j\alpha_l}^{(ijl)}$. For an observable $O=\sigma_{\alpha_1}^{(1)}\otimes\dots\sigma_{\gamma}^{(i)}\dots\otimes\sigma_{\eta}^{(j)}\dots\otimes  \sigma_{\alpha_k}^{(k)}$, $k>3$ we can write
\begin{eqnarray}\label{10_50_2_6}
\!\!\!\!\!\!\!\!\!\!\!\!\frac{d}{dt}\tr{\rho_t (\sigma_{\alpha_1}^{(1)}\otimes\dots\sigma_{\gamma}^{(i)}\dots\otimes\sigma_{\eta}^{(j)}\dots\otimes  \sigma_{\alpha_k}^{(k)})}|_{t=0}=
-2^{k-1}\!\!\!\!\sum\limits_{\alpha_1,\dots,\alpha_k} a_{\alpha_1\dots\alpha_k}^{(1\dots k)}(\varepsilon_{\alpha_i\tau_i \gamma}\delta_{\eta\alpha_j}+\varepsilon_{\alpha_j\tau_j\eta}\delta_{\gamma\alpha_i}).
\end{eqnarray}
Selecting $\tau_j=\eta$, $\tau_i\neq \gamma$, we can rewrite \eqref{10_50_2_6} as
\begin{equation}\label{10_50_2_7}
\frac{d}{dt}\tr{\rho_t (\sigma_{\alpha_1}^{(1)}\otimes\dots\sigma_{\gamma}^{(i)}\dots\otimes\sigma_{\eta}^{(j)}\dots\otimes  \sigma_{\alpha_k}^{(k)})}|_{t=0}=
-2^{k-1}\!\!\!\!\sum\limits_{\alpha_1,\dots,\alpha_k} a_{\alpha_1\dots\alpha_k}^{(1\dots k)}\varepsilon_{\alpha_i\tau_i \gamma}\delta_{\eta\alpha_j}.
\end{equation}
Solving this equation, we  find $a_{\alpha_1\dots\alpha_k}^{(1\dots k)}$.
\par From \eqref{0_10_10_1321} we can find three Lindbladian coefficients.
Selecting $\tau_i=y$, $\tau_j=\eta$, $\gamma=x$, we can find
\begin{equation}
\label{15_12}
D^{(i)}_{xy}=\frac{1}{2}\frac{d}{dt}\tr{\rho_{ty\eta} (\sigma_{x}^{(i)}\otimes \sigma_{\eta}^{(j)})}|_{t=0}-a_{z}^{(i)}
- a_{z\eta}^{(ij)}. 
\end{equation}
Next, for $\tau_i=z$, $\tau_j=\eta$, $\gamma=x$, we deduce
\begin{equation}\label{15_13}
D^{(i)}_{xz}=\frac{1}{2}\frac{d}{dt}\tr{\rho_{tz\eta} (\sigma_{ x}^{(i)}\otimes \sigma_{\eta}^{(j)})}|_{t=0}+ a_{y}^{(i)}+a_{y\eta}^{(ij)}
.
\end{equation}
Finally, for $\tau_i=z$, $\tau_j=\eta$, $\gamma=y$, the coefficient is
\begin{equation}\label{15_14}
D^{(i)}_{yz}=\frac{1}{2}\frac{d}{dt}\tr{\rho_{tz\eta}  (\sigma_{ y}^{(i)}\otimes \sigma_{\eta}^{(j)})}|_{t=0}-a_{x}^{(i)}
-a_{x\eta}^{(ij)}.
\end{equation}
From \eqref{10_50_2_1115} we can find three more coefficients. Selecting
$\tau_j=y$, $\tau_i=z$, $\gamma=z$, $\eta=x$, we get
\begin{eqnarray}\label{1236}
D^{(j)}_{xy}=\frac{1}{2}\frac{d}{dt}\tr{\rho_{tzy} (\sigma_{ z}^{(i)}\otimes \sigma_{x}^{(j)})}|_{t=0}-a_{z}^{(j)}-a_{zz}^{(ij)} .
\end{eqnarray}
For $\tau_j=y$, $\tau_i=y$, $\gamma=y$, $\eta=z$, we can deduce
\begin{eqnarray}\label{1237}
D^{(j)}_{zy}=\frac{1}{2}\frac{d}{dt}\tr{\rho_{tyy} (\sigma_{ y}^{(i)}\otimes \sigma_{z}^{(j)})}|_{t=0}+a_{x}^{(j)}+ a_{yx}^{(ij)}.
\end{eqnarray}
Selecting $\tau_j=z$, $\tau_i=y$, $\gamma=y$, $\eta=x$, we get
\begin{eqnarray}\label{1238}
D^{(j)}_{xz}=\frac{1}{2}\frac{d}{dt}\tr{\rho_{tyz} (\sigma_{ y}^{(i)}\otimes \sigma_{x}^{(j)})}|_{t=0}+a_{y}^{(j)}+a_{yy}^{(ij)}.
\end{eqnarray}
From \eqref{14_35_3} we find the following coefficients: 
\begin{eqnarray}\label{1059}
&&D^{(i)}_{xx}=\frac{1}{10}\left(-3\frac{d}{dt}\tr{\rho_{tx\tau_j} \sigma_{x}^{(i)}}|_{t=0}-3\frac{d}{dt}\tr{\rho_{ty\tau_j} \sigma_{y}^{(i)}}|_{t=0}+2\frac{d}{dt}\tr{\rho_{tz\tau_j} \sigma_{z}^{(i)}}|_{t=0}\right),\\\nonumber
&&D^{(i)}_{yy}=\frac{1}{10}\left(-3\frac{d}{dt}\tr{\rho_{tx\tau_j} \sigma_{x}^{(i)}}|_{t=0}+2\frac{d}{dt}\tr{\rho_{ty\tau_j} \sigma_{y}^{(i)}}|_{t=0}-3\frac{d}{dt}\tr{\rho_{tz\tau_j} \sigma_{z}^{(i)}}|_{t=0}\right),\\\nonumber
&&D^{(i)}_{zz}=\frac{1}{10}\left(2\frac{d}{dt}\tr{\rho_{tx\tau_j} \sigma_{x}^{(i)}}|_{t=0}-3\frac{d}{dt}\tr{\rho_{ty\tau_j} \sigma_{y}^{(i)}}|_{t=0}-3\frac{d}{dt}\tr{\rho_{tz\tau_j} \sigma_{z}^{(i)}}|_{t=0}\right).
\end{eqnarray}
From \eqref{10_50_2_1442} the following coefficients can be found
\begin{eqnarray}\label{1060}
&&D^{(j)}_{xx}=\frac{1}{10}\left(-3\frac{d}{dt}\tr{\rho_{t\gamma_1 x} \sigma_{x}^{(i)}}|_{t=0}-3\frac{d}{dt}\tr{\rho_{t\gamma_2y} \sigma_{y}^{(i)}}|_{t=0}+2\frac{d}{dt}\tr{\rho_{t\gamma_3z} \sigma_{z}^{(i)}}|_{t=0}\right),\\\nonumber
&&D^{(j)}_{yy}=\frac{1}{10}\left(-3\frac{d}{dt}\tr{\rho_{t\gamma_1 x} \sigma_{x}^{(i)}}|_{t=0}+2\frac{d}{dt}\tr{\rho_{t\gamma_2y} \sigma_{y}^{(i)}}|_{t=0}-3\frac{d}{dt}\tr{\rho_{t\gamma_3z} \sigma_{z}^{(i)}}|_{t=0}\right),\\\nonumber
&&D^{(j)}_{zz}=\frac{1}{10}\left(2\frac{d}{dt}\tr{\rho_{t\gamma_1 x}\sigma_{x}^{(i)}}|_{t=0}-3\frac{d}{dt}\tr{\rho_{t\gamma_2y} \sigma_{y}^{(i)}}|_{t=0}-3\frac{d}{dt}\tr{\rho_{t\gamma_3z} \sigma_{z}^{(i)}}|_{t=0}\right).
\end{eqnarray}
\begin{table}[h!]
\centering
\begin{tabular}{|c| c| c|}
 \hline
$D^{(i)}_{\mu\nu}$ & \{$O$, $\rho_{\tau_i,\tau_j}^{(i,j)}$\} & Equation  \\ [0.5ex] 
 \hline
$D^{(i)}_{xx}$, $D^{(i)}_{yy}$, $D^{(i)}_{zz}$ &  $\{\sigma_{x}^{(i)},\rho_{x,\tau_{j1}}^{(i,j)}\}$; $\{\sigma_{y}^{(i)},\rho_{y,\tau_{j2}}^{(i,j)}\}$; $\{\sigma_{z}^{(i)},\rho_{z,\tau_{j3}}^{(i,j)}\}$, $\forall\tau_{j1,2,3}\in\{x,y,z\}$  & \eqref{1059} \\
$D^{(j)}_{xx}$, $D^{(j)}_{yy}$, $D^{(j)}_{zz}$  & $\{\sigma_{x}^{(i)},\rho_{\gamma_1,x}^{(i,j)}\}$; $\{\sigma_{y}^{(i)},\rho_{\gamma_2,y}^{(i,j)}\}$; $\{\sigma_{z}^{(i)},\rho_{\gamma_3,z}^{(i,j)}\}$, $\forall\gamma_1\in\{y,z\}$, $\forall\gamma_2\in\{x,z\}$, $\gamma_3\in\{x,y\}$    & \eqref{1060} \\
$D^{(i)}_{xy}$ & $\{\sigma_{x}^{(i)}\otimes\sigma_{\eta}^{(j)},\rho_{y,\eta}^{(i,j)}\}$, $\forall\eta\in\{x,y,z\}$ & \eqref{15_12} \\
$D^{(i)}_{xz}$ & $\{\sigma_{x}^{(i)}\otimes\sigma_{\eta}^{(j)},\rho_{z,\eta}^{(i,j)}\}$, $\forall\eta\in\{x,y,z\}$ & \eqref{15_13} \\
$D^{(i)}_{yz}$ &  $\{\sigma_{y}^{(i)}\otimes\sigma_{\eta}^{(j)},\rho_{z,\eta}^{(i,j)}\}$, $\forall\eta\in\{x,y,z\}$ & \eqref{15_14} \\
$D^{(j)}_{xy}$ & $\{\sigma_{z}^{(i)}\otimes\sigma_{x}^{(j)},\rho_{z,y}^{(i,j)}\}$ & \eqref{1236} \\
$D^{(j)}_{xz}$ & $\{\sigma_{y}^{(i)}\otimes\sigma_{x}^{(j)},\rho_{y,z}^{(i,j)}\}$ & \eqref{1238} \\
$D^{(j)}_{yz}$ & $\{\sigma_{y}^{(i)}\otimes\sigma_{z}^{(j)},\rho_{y,y}^{(i,j)}\}$ & \eqref{1237} \\
\hline
\end{tabular}
\caption{
The first column represents the type of the estimated Lindbladian parameters $D^{(i)}_{\mu\nu}$, $\mu,\nu\in\{x,y,z\}$. In the third column the number of  equation for every parameter is provided, depending from the pairs of the observable $O$ and the initial state $\rho_{\tau_i,\tau_j}^{(i,j)}=\rho_{\tau_i}^{(i)}\otimes \rho_{\tau_j}^{(j)}\otimes \rho_{N-2}$, ${\tau_i,\tau_j}=\{x,y,z\}$, given in the second column.
For example, to estimate $D^{(j)}_{yz}$ we use \eqref{1237} with one pair of observable and state: $\{O, \rho_{\tau_i,\tau_j}^{(i,j)}\}=\{\sigma_{y}^{(i)}\otimes\sigma_{z}^{(j)},\rho_{y,y}^{(i,j)}\}$.}
\label{tab_3}
\end{table}
All the Lindbladian coefficients with the corresponding observables and initial states are given in Table~\ref{tab_3}.
\section{Numerical simulations}

For simulations of the Hamiltonian learning protocol, we employed direct numerical solution of the time-dependent Schrodinger equation for the whole system of $N=16$ qubits, whose time-dependent wavefunction 
$|\psi(t)\rangle$ is represented as an array of $2^N$ complex numbers, normalized to 1. The evolution includes both unitary component, governed by the system's Hamiltonian, and three non-unitary components: the first one, stemming from the quasi-static random frequency shifts, leading to essentially non-Markovian dephasing of the qubits with the characteristic time $T_2^*$, the second component, described as a set of Lindblad superoperators corresponding to the phase damping channel, leading to Markovian transverse decoherence of the qubits on the timescale $T_2$, and the third component, also leading to Markovian evolution of the qubit, and described as a set of Lindblad superoperators corresponding to the amplitude damping channel, which leads to longitudinal relaxation of the qubit on the timescale $T_1$. For precise meaning of the terms ``Markovian'' and ``non-Markovian'', see explanations below in subsection \ref{sec:nonun}.

\subsection{Unitary evolution}
\label{sec:un}

The simulation of the unitary (Hamiltonian) evolution was performed using the 2nd order Suzuki-Trotter decomposition of the evolution operator. The total Hamiltonian of the system in question can be written as 
\begin{equation}
\label{eq:NSHam}
H = \sum_{j,k=1}^N \sum_{\alpha=x,y} J_{jk} \sigma^\alpha_j \sigma^\alpha_k + \frac{1}{2} \sum_{j=1}^N \Omega_j \sigma_j^z,
\end{equation}
where the operators $\sigma^\alpha_j$ for $\alpha=x,y,z$ denote the Pauli matrices $\sigma^x_j$, $\sigma^y_j$, and $\sigma^z_j$, respectively, corresponding to the $j$-th qubit, and for the problem considered in this paper the couplings $J_{jk}$ are restricted to the nearest neighbor qubits on a 2-D lattice. 
Note that the actual frequency ${\Omega}_j$ of the $j$-th qubit in the Hamiltonian (\ref{eq:NSHam}) is different from its nominal frequency ${\tilde\omega}_m$ mentioned in the main text; the reasons for this difference are explained in subsection \ref{sec:nonun} below.

The Hamiltonian (\ref{eq:NSHam}) is represented as a sum
\begin{eqnarray}
H &=& H_X + H_Y + H_Z\\
H_X &=& \sum_{j,k=1}^N J_{jk} \sigma^x_j \sigma^x_k,\\
H_Y &=& \sum_{j,k=1}^N J_{jk} \sigma^y_j \sigma^y_k,\\
H_Z &=& \frac{1}{2} \sum_{j=1}^N {\Omega}_j \sigma_j^z,
\end{eqnarray}
and the corresponding Suzuki-Trotter decomposition of the evolution operator $U(\Delta t)$ for the (small) timestep of duration $\Delta t$ has the form 
\begin{equation}
\label{eq:st}
U(\Delta t) \equiv \exp{(-iH \Delta t)} \approx {\rm e}^{-iH_Z \Delta t/2}\,
{\rm e}^{-iH_Y \Delta t/2}\, {\rm e}^{-iH_X \Delta t}\, 
{\rm e}^{-iH_Y \Delta t/2}\, {\rm e}^{-iH_Z \Delta t/2},
\end{equation}
ensuring the overall time discretization error of the order $(\Delta t)^2$. The evolution operator over many time steps is a product of elementary operators $U(\Delta t)$.

Each term in the sum representing the Hamiltonian $H_X$ (and, similarly, $H_Y$ and $H_Z$) commutes with all other terms, therefore
\begin{equation}
\exp{(-iH_X \Delta t)} = \prod_{j,k=1}^N \exp{(-i J_{jk} \Delta t\, \sigma^x_j \sigma^x_k)} = \prod_{j,k=1}^N \left[\cos{(J_{jk} \Delta t)} - i \sigma^x_j \sigma^x_k\, \sin{(J_{jk} \Delta t)}\right].
\end{equation}
Each term in this direct product acts on the wavefunction $|\psi(t)\rangle$ in a straightforward manner: the entries of the array that represents the wavefunction turn into linear combinations of themselves. Similar direct-product representation holds for $H_Y$ and $H_Z$ as well, such that the action of the total evolution operator $U(\Delta t)$ is easy to compute, without the need to calculate or store $2^N\times 2^N$ matrices.

In order to represent the situation where the non-initialized part of the system is in the completely mixed state, but avoid using the density matrix explicitly (which would imply dealing with $2^N\times 2^N$ matrix instead of the single array of the size $2^N$), we represent the completely mixed state as a wavefunction with random entries \cite{VD1,VD2}. Specifically, we sampled the real and the imaginary parts of each entry of the corresponding wavefunction independently from Gaussian distribution with zero mean and unit variance, and then normalized the resulting wavefunction to one. In this way, for instance, the situation where the first and the second qubit are both initialized in the state $|0\rangle$, while the rest of the system is in completely mixed state, i.e.\ when the system's density matrix is
\begin{equation}
\rho = |0\rangle\langle 0| \otimes |0\rangle\langle 0| \otimes 
\frac{1}{2^{N-2}}\, {\mathbf 1}_{N-2},
\end{equation}
where ${\mathbf 1}_{N-2}$ is an identity matrix of the size 
$2^{N-2}\times 2^{N-2}$,
is represented using the total wavefunction in the form
\begin{equation}
|\psi\rangle = |0\rangle \otimes |0\rangle \otimes |\psi^{(r)}_{N-2}\rangle,
\end{equation}
where the random state $|\psi^{(r)}_{N-2}\rangle$ of the remaining $N-2$ qubits is generated as described above. Such an approximation provides high accuracy, of the order of $\exp{(-N/2)}$, due to the measure concentration phenomenon \cite{Ledoux}. 

Further improvement in accuracy was achieved by averaging the values of the relevant observables over $M=189$ independent realizations of the random wavefunction (as well as other random quantities, see below), which  reduced the error by an additional factor of the order $\sim 1/\sqrt{M}\approx 0.07$. The accuracy was also independently 
controlled by estimating the variance in the calculated values of the observables, and ensuring that this variance remains much smaller than the statistical error caused by the shot noise produced by sampling the relevant observables for each qubit.

\subsection{Non-unitary evolution}
\label{sec:nonun}

The first non-unitary component of the system's evolution, dephasing of the $j$-th qubit on the timescale $T_{2,j}^*$, caused by its random static frequency shift, is modeled by directly reproducing the underlying physical picture. Namely, we assumed that the actual frequency $\Omega_j$ of the $j$-th qubit, see Eq.~\ref{eq:NSHam}, is a sum of  two contributions: the nominal value  ${\tilde\omega}_j$, and a random shift $\beta_j$ that remains constant during the system's evolution. The values of $\beta_j$ were independently sampled from Gaussian distributions with zero mean and variance $b_j^2$, which can be different for different qubits. 
The parameter $b_j$ determines the dephasing time $T_{2,j}^*$ of the $j$-th qubit: if this qubit were uncoupled from the rest of the system, then, after averaging over $\beta_j$, its transverse ($x$- and $y$-) components would undergo Gaussian decay with time dependence $\exp{(-b_j^2 t^2/2)}$, i.e.\ 
$b_j=\sqrt{2}/T_{2,j}^*$. 

As mentioned above in subsection \ref{sec:un}, the evolution of the system was repeated $M=189$ times; each time we used different realizations of the set of the random frequency shifts $\beta_j$ (as well as other random quantities, such as e.g.\ different realizations of the random wavefunction, also see below). 
Within this approach, for each particular realization of the parameters $\beta_j$, the evolution of the system is unitary, and can be simulated using the system's wavefunction as described in subsection \ref{sec:un} above, while 
the non-unitary decay occurs due to averaging of the relevant observables over different realizations of the random parameters $\beta_j$ (along with other random quantities).

Note that the dephasing caused by averaging over the static random frequency shifts, with its characteristic Gaussian-like decay of the density matrix elements, cannot be described via Lindblad operators. It is an example of 
non-Markovian evolution, in the sense that it cannot be described by a set of first-order differential equations (with respect to time $t$), which would include only current values of the (averaged) elements of the system's density matrix $\rho(t)$; in other words, the future values of the (averaged) density matrix elements, at times $t+s$ ($s>0$), are not completely determined by their current values at the moment of time $t$. At the same time, the static noise processes $\beta_j(t)$ representing the random frequency shifts are, of course, Markovian random processes, sastisfying the Chapman-Kolmogorov equation.

The two other components of the non-unitary evolution, addressed below, are Markovian, and can be described using the Lindblad operators. However, in order to avoid dealing with the density matrix, these components were also modeled by employing the random processes and calculating the averages of the relevant observables.

The second non-unitary component of the evolution corresponds to the Markovian dephasing, and can be described via the set of Lindblad superoperators corresponding to the phase damping channel. For an isolated qubit, this would lead to exponential decay of the transverse components of the $j$-th qubit, having the form $\exp{(-t/T_{2,j})}$. This kind of dephasing, being Markovian, can be described by a set of first-order differential equations, generalizing the well-known Bloch-Redfield equations \cite{slichter}, which include only the current values of the elements of the system's density matrix $\rho(t)$, such that the future values of the density matrix elements, at times $t+s$ ($s>0$), are completely determined by their current values at the moment of time $t$.

This decay was modeled by taking the $z$-rotation of the $j$-th qubit produced by $\Omega_j$ (i.e., produced by the action of the operator 
$\exp{(-iH_Z\Delta t})$ in Eq.~\ref{eq:st}), and adding to it another time-dependent rotation around the $z$-axis by the angle $\gamma_j(t)$. For each time step of duration $\Delta t$, the values $\gamma_j(t)$ were sampled randomly, indepedently of each other and of their previous values, from Gaussian distribution with zero mean and variance $4 g_j^2 \tau_j \Delta t$. This choice for the quantity $\gamma(t)$ can be visualized as a rotation induced by a time-dependent frequency shift $\delta\omega_j(t)$, which is represented by an Ornstein-Uhlenbeck noise process with the correlation function 
$\langle\delta\omega_j(t)\, \delta\omega_j(t+s)\rangle=g_j^2 \exp{(-|s|/\tau_j)}$, in the limit where the correlation time $\tau_j$ is much smaller than $\Delta t$, while the magnitude $g_j$ is large (formally, $\tau_j\to 0$ and $g_j\to\infty$), but the combination $g_j^2 \tau_j = 1/T_{2,j}$ remains finite. For an isolated qubit, the average evolution under the influence of such Ornstein-Uhlenbeck noise $\delta\omega_j(t)$ is known \cite{kubo} to produce  exponential decay of the qubit's transverse ($x$- and $y$-) components with the decay time $T_2$. 
Again, for each particular realization of the time-dependent random process $\gamma_j(t)$, the evolution of the system is unitary, and can be simulated using the system's wavefunction as described in subsection \ref{sec:un} (provided, of course, that $\Delta t\ll T_{2,j}$, to ensure accuracy of the Suzuki-Trotter decomposition), while the non-unitary decay occurs due to averaging over different realizations of the noise.

The third non-unitary component, describing exponential relaxation of the $j$-th qubit towards the state $|0\rangle$ on a timescale $T_{1,j}$, was simulated in a similar manner, by representing the non-unitary evolution via averaging over many realizations of a random unitary evolution, employing the approach described in Ref.~\cite{DCM}, with some modifications improving the accuracy. Namely, at each time step, we calculated the probability $p_j$ for the $j$-th qubit to make a transition (``quantum jump'') from the state $|1\rangle$ to the state $|0\rangle$; the corresponding value is 
$p_j=w_{j1}\, \mu_j^2$, where $\mu_j=\sqrt{1-\exp{(-\Delta t/T_{1,j})}}$, and $w_{j1}$ is the total probability of the system to be in the subspace corresponding to the $j$-th qubit in the state $|1\rangle$. This transition was implemented with the probability $p_j$ at each time step: the part of the system's wavefunction corresponding to the $j$-th qubit in the state $|0\rangle$ was replaced by its complement, i.e.\ by the part corresponding to the $j$-th qubit in the state $|1\rangle$, multiplied by the factor $\mu_j$, and the part of the wavefunction corresponding to the $j$-th qubit in the state $|1\rangle$ was set to zero. Alternatively, with the probability $1-p_j$ at each time step, the part of the wavefunction corresponding to the $j$-th qubit in the state $|1\rangle$ was multiplied by the factor 
$\exp{[-(1/2)\,\Delta t/T_{1,j}]}$, while its complement was left unchanged.
These transformations were applied to the wavefunction in succession, for all qubits (for all $j=1,\dots N$), and the resulting modified wavefunction was normalized back to 1.
Since all these transformations commute with the action of the operator 
$\exp{(-iH_Z\Delta t)}$ in Eq.~\ref{eq:st}, they were applied at the end of each unitary time-step evolution, after application of the operator $U(\Delta t)$ given by Eq.~\ref{eq:st}, in parallel with the action of the operators 
$\exp{(-iH_Z\Delta t)}$ or $\exp{(-iH_Z\Delta t/2)}$.

Note that this implementation corresponds to the application to the wavefunction of the Krauss operators $E_0$ or $E_1$ (see Ref.~\onlinecite{NC}), describing the amplitude damping quantum channel,  with the corresponding probabilities, where $E_1$ corresponds to the event of the ``quantum jump'', and $E_0$ corresponds to the absence of it.

\section{Numerical simulation of superconducting qubit platform}\label{sec:example_Hamiltonian}
From the discussion in the main text we are simulating  a 2D grid of qubits that interact only with the nearest neighbours.
The coupling between two neighbouring qubits through a coupler can be described by an Hamiltonian \eqref{eq:hamil1}.
In our notations it can be rewritten as
\begin{equation} \label{eq:hamil0}
H=\sum_{k=i,j}a_{z}^{(k)}H_{z}^{(k)}+   a_{xx}^{(i,j)}H_{xx}^{(i,j)}+a_{yy}^{(i,j)}H_{yy}^{(i,j)},
\end{equation}
where we introduce the notations
\begin{eqnarray} \label{eq:hamil1_1}
&&a_{z}^{(i)}=\frac{1}{2}\tilde{\omega}^{(i)},\quad a_{z}^{(j)}=\frac{1}{2}\tilde{\omega}^{(j)},\quad a_{xx}^{(i,j)}=a_{yy}^{(i,j)}=\frac{1}{2}\left[\frac{g_ig_j}{\Delta}+g_{ij}\right],\\\nonumber
&&H_{z}^{(i)}=\sigma_z^{(i)},\quad H_{z}^{(j)}=\sigma_z^{(j)},\quad H_{xx}^{(i,j)}=\sigma_x^{(i)}\sigma_x^{(j)},\quad H_{yy}^{(i,j)}=\sigma_y^{(i)}\sigma_y^{(j)}.
\end{eqnarray}
Thus, we define the $16$ qubits $2D$ grid, where we generate $\omega_1$, $\omega_2$, $\omega_c$, $g_1$, $g_2$, $g_{12}$ from Gaussian distribution with mean and variance $\mathcal{N}(0 , 1)$. The parameters $a_{z}^{(i)}$, $a_{z}^{(j)}$, $a_{xx}^{(ij)}$,  $a_{yy}^{(ij)}$ we estimate in our simulation and the rates of decay of $T_1$, $T_2$ and $T_2^{\star}$ are given in the Table~\ref{tab_2_1}.
\begin{table}[h!]
\centering
\begin{tabular}{|c| c|c|c|c|c|c|c| }
 \hline
$(ij)$ & $a_{xx}^{(ij)},a_{yy}^{(ij)}$ [kHz] & & $(i)$& $a_{z}^{(i)} $ [kHz] &$T_1$ [$\mu$s] &$ T_2$ [$\mu$s]&$ T_2^{\star}$ [$\mu$s]\\ [0.5ex] 
\hline
1, 2 &	1.28112& &1& 1.73807 &58.5227&65.9752&151.515\\	\hline
2, 3 &	-0.716875& &2& -0.816877&60.0269 &65.1704&166.667\\	\hline	
3, 4 &	-0.956949& &3& -1.0602&59.2424&64.6375& 163.934\\	\hline	
4, 5&	-0.819328& &4& -0.913223&61.0255&65.7397&149.254\\	\hline	
1, 6&	-1.1682& &5& -1.23118&59.0545 &66.0886&147.059\\	\hline	
2, 7&  -0.213057& &6& -0.654699&60.0915&66.1118&151.515\\	\hline	
3, 8&  -0.563789& &7& -0.514756&59.8856 &65.1432& 153.846\\	\hline	
4, 9&	1.74022& &8&  2.0817&61.0389& 64.8252&158.73\\	\hline	
5, 10&	1.68348& &9& -0.568581&60.5375&66.2155& 158.73
\\	\hline	
6, 7& -1.51535& &10& -0.710498&61.5036&65.389&149.254\\	\hline	
7, 8&	-0.729672& &11& 1.86153&59.8949&65.825& 147.059\\	\hline	
8, 9&	1.6622& &12& -2.03725&60.3777&65.1203& 153.846\\	\hline	
9, 10&	-0.314438& &13& -1.31695&57.5781& 65.6052& 156.25\\	\hline	
11, 12& -0.475787& &14& -0.902159&59.1881& 65.8892&158.73\\	\hline	
6, 11&  1.3663& &15& -0.202118&60.0283& 66.2967&144.928
\\	\hline	
7, 12&	-2.03531& &16&  0.136975&58.9397&66.0541&149.254
\\	\hline	
12, 13&	-1.22632& && & &&\\	\hline	
13, 8& -0.717182& && & &&\\	\hline	
13, 14&	-0.546421& && & &&\\	\hline	
14, 9&	1.90836& && & &&\\	\hline
14, 15 &	-0.781306& && & &&\\ \hline
11, 16 & -0.358714& && &&&\\ \hline
\end{tabular}
\caption{The first column represents the numbers of  qubits whose  coefficients $a_{\alpha_i,\alpha_j}^{(ij)}$, $\alpha_i,\alpha_j\in \{x,y,z\}$ are not zero, given in the second column. The fourth column contains $a_{\alpha_i}^{(i)}$, $\alpha_i\in \{x,y,z\}$, $i=1,\dots,16$. The rates of decay of $T_1$, $T_2$ and $T_2^{\star}$, corresponding to $i$'s qubit are given in the fifth and six's columns, respectively.}
\label{tab_2_1}
\end{table}
The observables and initial states, isolating the desired coefficients $a_{z}^{(i)}$, $a_{z}^{(j)}$, $a_{xx}^{(i,j)}$ and  $a_{yy}^{(i,j)}$, are given in Table~\ref{tab_2}. One can see, that we need three starting states, namely  $\rho_{01}=\rho_{x}^{(i)}\otimes \rho_{x}^{(j)}\otimes \frac{I^{2n-2}}{2^{2n-2}}$, $\rho_{02}=\rho_{y}^{(i)}\otimes \rho_{z}^{(j)}\otimes \frac{I^{2n-2}}{2^{2n-2}}$ and $\rho_{03}=\rho_{z}^{(i)}\otimes \rho_{z}^{(j)}\otimes \frac{I^{2n-2}}{2^{2n-2}}$  to isolate all four unknown coefficients. We measure the expectation values of  observables in 
different times. Next, the time traces of these expectation
values are fitted, using the polynomial interpolation method, and  the derivatives estimation is preceded. Finally, using \eqref{1747_3}, \eqref{1731}, \eqref{9206} and \eqref{9206}, the estimates of the coefficients $a_{z}^{(i)}$, $a_{z}^{(j)}$, $a_{xx}^{(i,j)}$, $a_{yy}^{(i,j)}$ for the pair of $(i,j)$ are obtained. We repeat this process for all pairs of interacting qubits to obtain all coefficients of the Hamiltonian of the $2D$ grid. 
\par In the presence of the noise, the observables and initial states required to isolate the Lindbladian coefficients  are given in Table~\ref{tab_4}. One can see, that we need three extra starting states in the presence of the Lindbladian noise, namely  $\rho_{04}=\rho_{z}^{(i)}\otimes \rho_{y}^{(j)}\otimes \frac{I^{2n-2}}{2^{2n-2}}$, $\rho_{05}=\rho_{y}^{(i)}\otimes \rho_{y}^{(j)}\otimes \frac{I^{2n-2}}{2^{2n-2}}$ and $\rho_{06}=\rho_{z}^{(i)}\otimes \rho_{x}^{(j)}\otimes \frac{I^{2n-2}}{2^{2n-2}}$  to find $L^{(i)}_{\mu\nu}$, $\mu,\nu\in \{x,y,z\}$.
\begin{table}[h!]
\centering
\begin{tabular}{|c| c| c|}
 \hline
$a_{\alpha_i}^{(i)}$ & \{$O$, $\rho_{\tau_i,\tau_j}^{(i,j)}$\}  & Equation  \\ [0.5ex] 
 \hline
$a_{z}^{(i)}$ & $\{\sigma_{x}^{(i)}\otimes \sigma_{z}^{(j)},\rho_{y,z}^{(i,j)}\}$; $\{\sigma_{x}^{(i)},\rho_{y,z}^{(i,j)}\}$& \eqref{1747_3}\\
$a_{z}^{(j)}$ & $\{\sigma_{x}^{(i)}\otimes \sigma_{y}^{(j)},\rho_{x,x}^{(i,j)}\}$; $\{\sigma_{y}^{(i)},\rho_{x,x}^{(i,j)}\}$& \eqref{1731}\\
$a_{xx}^{(ij)}$& $\{\sigma_{x}^{(i)}\otimes \sigma_{y}^{(j)},\rho_{y,z}^{(i,j)}\}$; $\{\sigma_{y}^{(i)}\otimes \sigma_{y}^{(j)},\rho_{x,x}^{(i,j)}\}$ &\eqref{9206}\\
$a_{yy}^{(ij)}$ &  $\{\sigma_{x}^{(i)}\otimes \sigma_{y}^{(j)},\rho_{z,z}^{(i,j)}\}$; $\{\sigma_{x}^{(i)}\otimes \sigma_{y}^{(j)},\rho_{y,z}^{(i,j)}\}$; $\{\sigma_{y}^{(i)}\otimes \sigma_{y}^{(j)},\rho_{x,x}^{(i,j)}\}$  &\eqref{9206}\\
\hline
\end{tabular}
\caption{In this table the minimal selection of the pairs \{$O$, $\rho_{\tau_i,\tau_j}^{(i,j)}$\} is presented for our specific example. The first column represents the type of the estimated Hamiltonian \eqref{eq:hamil0} parameters $a_{\alpha_i}^{(i)}$, $a_{\alpha_i,\alpha_j}^{(ij)}$, ${\alpha_i,\alpha_j}\in\{x,y,z\}$. In the third column the number of  equation for every parameter is provided, depending from the pairs of the observable $O$ and the initial state $\rho_{\tau_i,\tau_j}^{(i,j)}=\rho_{\tau_i}^{(i)}\otimes \rho_{\tau_j}^{(j)}\otimes \rho_{N-2}$, ${\tau_i,\tau_j}=\{x,y,z\}$, given in the second column. To estimate all four parameters we need only three initial states:
$\rho_{x,x}^{(i,j)}$, $\rho_{z,z}^{(i,j)}$ and $\rho_{y,z}^{(i,j)}$.}
\label{tab_2}
\end{table}
\begin{table}[h!]
\centering
\begin{tabular}{|c| c| c|}
 \hline
$D^{(i)}_{\mu\nu}$ &\{$O$, $\rho_{\tau_i,\tau_j}^{(i,j)}$\}  & Equation  \\ [0.5ex] 
 \hline
$D^{(i)}_{xx}$, $D^{(i)}_{yy}$, $D^{(i)}_{zz}$ &   $\{\sigma_{x}^{(i)},\rho_{x,x}^{(i,j)}\}$; $\{\sigma_{y}^{(i)},\rho_{y,z}^{(i,j)}\}$; $\{\sigma_{z}^{(i)},\rho_{z,z}^{(i,j)}\}$  & \eqref{1059} \\
$D^{(j)}_{xx}$, $D^{(j)}_{yy}$, $D^{(j)}_{zz}$  & $\{\sigma_{x}^{(i)},\rho_{z,x}^{(i,j)}\}$; $\{\sigma_{y}^{(i)},\rho_{z,y}^{(i,j)}\}$; $\{\sigma_{z}^{(i)},\rho_{y,z}^{(i,j)}\}$  & \eqref{1060} \\
$D^{(i)}_{xy}$ & $\{\sigma_{x}^{(i)}\otimes\sigma_{z}^{(j)},\rho_{y,z}^{(i,j)}\}$ & \eqref{15_12} \\
$D^{(i)}_{xz}$ & $\{\sigma_{x}^{(i)}\otimes\sigma_{\eta}^{(j)},\rho_{z,z}^{(i,j)}\}$ & \eqref{15_13} \\
$D^{(i)}_{yz}$ &  $\{\sigma_{y}^{(i)}\otimes\sigma_{\eta}^{(j)},\rho_{z,z}^{(i,j)}\}$ & \eqref{15_14} \\
$D^{(j)}_{xy}$ & $\{\sigma_{z}^{(i)}\otimes\sigma_{x}^{(j)},\rho_{z,y}^{(i,j)}\}$ & \eqref{1236} \\
$D^{(j)}_{xz}$ & $\{\sigma_{y}^{(i)}\otimes\sigma_{x}^{(j)},\rho_{y,z}^{(i,j)}\}$ & \eqref{1238} \\
$D^{(j)}_{yz}$ & $\{\sigma_{y}^{(i)}\otimes\sigma_{z}^{(j)},\rho_{y,y}^{(i,j)}\}$ & \eqref{1237} \\
\hline
\end{tabular}
\caption{In this table the minimal selection of the pairs \{$O$, $\rho_{\tau_i,\tau_j}^{(i,j)}$\} is presented for our specific example. The first column represents the type of the estimated Lindbladian parameters $D^{(i)}_{\mu\nu}$, $\mu,\nu\in\{x,y,z\}$. In the third column the number of  equation for every parameter is provided, depending from the pairs of the observable $O$ and the initial state $\rho_{\tau_i,\tau_j}^{(i,j)}=\rho_{\tau_i}^{(i)}\otimes \rho_{\tau_j}^{(j)}\otimes \rho_{N-2}$, ${\tau_i,\tau_j}=\{x,y,z\}$, given in the second column.  To estimate all  parameters we need only three extra states to the ones given by the previous table, namely:
$\rho_{z,x}^{(i,j)}$, $\rho_{z,y}^{(i,j)}$ and $\rho_{y,y}^{(i,j)}$.
}
\label{tab_4}
\end{table}

\section{Approximating local time evolutions by polynomials}\label{sec:quality_approx_poly}

One of the main points behind our method is the fact that the time evolution of local observables at constant times is well-approximated by polynomials. The purpose of this section is to make this assertion precise. 

Before we do that, let us set some notation. Given a system of $n$ qubits on a $D$, we let $\cL_{\Gamma}:\M_{2^n}\to\M_{2^n}$ be a Lindbladian which models the time evolution of the system in the Heisenberg picture. Note that in the supplementary information we consider a slightly more general class of evolutions than in the main text. There, we restricted to evolutions whose Hamiltonians were short range with two-body interactions and the noise acted on at most one qubit at a time. Here, in contrast, we will also consider $k$-local evolutions with long range.

We will assume that this Lindbladian can be written as:
\begin{align}\label{equ:generator_local}
    \cL_{\Gamma}=\sum\limits_{A\subset\Gamma}\cL_{A},
\end{align}
where $\cL_{A}$ is a Lindbladian only acting on the qudits in $A$. Given some graph $G=(V,E)$ on $n$ vertices, we will say that $\cL$ is $k-$local if $\cL_A\not=0$ only if $A$ is a subset of vertices of $G$ containing at most $k$ vertices. Furthermore, we will say that $\cL_{\Gamma}$ is locally bounded if there is a constant $g>0$ such that for all $B\subset\Gamma$ we have that:
\begin{align}\label{equ:lindblad_local_bounded}
    \|\sum\limits_{A\subset \Gamma:A\cap B\not=\varnothing}\cL_A\|\leq g |B|.
\end{align}
This condition is satisfied if e.g. $\cL$ is a local Lindbladian on a $D-$dimensional lattice. In that case, we have $g=\cO(D)$. However, this condition is also fulfilled for generators with algebraically decaying tails, as long as these tails decays fast enough.
Moreover, for ease of notation we will let for a region $B\subset\Gamma$
\begin{align}
    \cL_B=\sum\limits_{A\subset B }\cL_A
\end{align}
be the generator restricted to a subregion $B$.

Furthermore, given the graph $G$, some region $X\subset V$ and $r>0$, we will denote by $\Lambda_r(X)$ the set of vertices that are a distance at most $r$ from $X$:
\begin{align}
    \Lambda_r(X)=\{v\in V:\exists x\in X\textrm{ s.t. }d(x,v)\leq r\}.
\end{align}

We will also require some norms for superoperators. Given a superoperator $\Phi:\M_{2^n}\to\M_{2^n}$ we define for $p,q\geq1$
\begin{align*}
    \|\Phi\|_{p\to q}=\sup\limits_{X\in\M_{2^n}}\frac{\|\Phi(X)\|_{q}}{\|X\|_p},
\end{align*}
where $\|\cdot\|_p$ corresponds to the Schatten $p$-norm. Also note that $p=\infty$ corresponds to the operator norm.

Our goal will be to prove the following statement:

\begin{theorem}[Polynomials approximate the evolution of local expectation values, informal]\label{thm:polynomialswork2}
Let $\linGen_{\Gamma}$ be a local Lindbladian on a $D$-dimensional regular lattice. Moreover, let $t_{\max},\epsilon>0$ be given and $O_Y$ and observable supported on a constant number of qubits. Assume that $\linGen_{\Gamma}$ satisfies a Lieb-Robinson bound. Then there is a polynomial $p$ of degree 
\begin{align*}
    d=\tcO\left[\left(h^{-1}\left(\frac{\epsilon}{e^{vt\max}-1}\right)\right)^{D}t_{\max}\log(\epsilon^{-1})\right]
\end{align*}
such that for all $0\leq t\leq t_{\max}$:
\begin{align}\label{equ:quality_approx2}
    \left|\tr{e^{t\cL_{\Gamma}}(O_Y)\rho}-p(t)\right|\leq\epsilon.
\end{align}
and
\begin{align}\label{equ:first_derivative_polynomial2}
p'(0)=\tr{\cL_{\Gamma}(O_Y)\rho}.
\end{align}
\end{theorem}

We will start by showing a similar statement for the time evolution of truncated local evolutions. After that we will conclude by showing that truncated, local evolutions approximate the global evolution well.
It is simple to see that derivatives of locally bounded, truncated evolutions can only increases with the size of the region they are defined on:
\begin{lem}[Derivatives of truncated local evolutions]\label{lem:local_derivatives_bounded}
Let $\cL_{\Gamma}$ be a locally bounded Lindbladian with constant $g$. For an observable $O$ such that $\|O\|\leq1$, an initial state $\rho$ and a region $B\subset \Gamma$ define the evolution of the truncated evolution $f_B:\R^+\to\R$ as $f_B(t)=\tr{e^{t\cL_B}(O)\rho}$. Then for all $t\geq0$:
\begin{align}\label{equ:local_derivatives_bound}
    \left|f_B^{(k)}(t)\right|\leq \left(tg|B|\right)^k
\end{align}
In particular, for any $0<t<t_{\max}$ we have that:
\begin{align}\label{equ:estimate_derivative_local}
    \left|f_B(t)-\sum\limits_{k=0}^{d}\frac{f_B^{(k)}(0)}{k!}t^k\right|\leq \frac{\left(t_{\max}g|B|\right)^{d+1}}{(d+1)!}.
\end{align}
\end{lem}
\begin{proof}
The proof is elementary. Note that:
\begin{align*}
    f_B^{(k)}(t)=\tr{e^{t\cL_B}((t\cL_B)^k(O))\rho}.
\end{align*}
Now, by H\"older's inequality we have that:
\begin{align*}
    \left|f_B^{(k)}(t)\right|\leq \|e^{t\cL_B}((t\cL_B)^k(O))\|_{\infty}\|\rho\|_{1}\overset{(1)}{\leq}t^k\|e^{t\cL_B}\|_{\infty\to\infty}\|\cL_B\|^k_{\infty\to\infty} \overset{(2)}{\leq}t^kg^k|B|^k,
\end{align*}
where in $(1)$ we used the submultiplicativity of the operator norm, i.e. $\|\Phi_1\Phi_2\|_{\infty\to\infty}\leq \|\Phi_1\|_{\infty\to\infty}\|\Phi_2\|_{\infty\to\infty}$ for all linear maps $\Phi_1,\Phi_2$. In (2) we used the fact that for any quantum channel $\|e^{t\linGen_B}\|_{\infty\to\infty}=1$ and the fact that the Lindbladian is locally bounded with constant $g$.
The estimate in Eq.~\eqref{equ:estimate_derivative_local} then immediately follows from Taylor's remainder theorem.
\end{proof}
We then immediately have:
\begin{cor}\label{cor:expectation_values_low_degree_poly}
In the same setting as Lemma~\ref{lem:local_derivatives_bounded} it holds that for any given $\epsilon>0$ and $t_{\max}>0$ there is a polynomial $p$ of degree 
\begin{align}\label{equ:degree_polynomial}
   d=2et_{\max}g|B|\log(\epsilon^{-1})-1 
\end{align}
such that for all $0\leq t\leq t_{\max}$ we have that 
\begin{align*}
    \left|f_B(t)-p(t)\right|\leq \epsilon.
\end{align*}
\end{cor}
\begin{proof}
It follows from Sitrling's approximation that the error in Eq.~\eqref{equ:estimate_derivative_local} is bounded by
\begin{align}\label{equ:estimate_decay}
    \left|f(t)-\sum\limits_{k=0}^{d}\frac{f^{(k)}(0)}{k!}t^k\right|\leq \frac{1}{d\sqrt{2\pi}}\left(\frac{et_{\max}g|B|}{d+1} \right)^{d+1}.
\end{align}
It is then easy to see that picking $d=2et_{\max}g|B|\log(\epsilon^{-1})-1$ is sufficient to ensure that the error in~\eqref{equ:estimate_derivative_local} is at most $\epsilon$. 
Indeed, plugging in the value of $d$ into Eq.~\eqref{equ:estimate_decay} we get:
\begin{align}\label{equ:last_estimate_before_poly}
    &\left|f(t)-\sum\limits_{k=0}^{d}\frac{f^{(k)}(0)}{k!}t^k\right|\leq \frac{1}{d\sqrt{2\pi}}\left(\frac{1}{\log(\epsilon^{-1})} \right)^{2et_{\max}g|B|\log(\epsilon^{-1})}=\\&\frac{1}{d\sqrt{2\pi}}\operatorname{exp}\left[-\log(\log(\epsilon^{-1}))\log(\epsilon^{-1})2et_{\max}g|B|\right]=\frac{1}{d\sqrt{2\pi}}\epsilon^{\log(\epsilon^{-1})2et_{\max}g|B|}\leq \epsilon.
\end{align}

Thus, the truncated Taylor expansion yields the desired polynomial.
\end{proof}
We conclude from Lemma~\ref{lem:local_derivatives_bounded} and Cor.~\ref{cor:expectation_values_low_degree_poly} that local, truncated time evolutions are well-approximated by polynomials whose degree grows like the size of the region times the maximal time of evolution.

Also note that the estimate in Eq.~\eqref{equ:last_estimate_before_poly} is quite loose and shows that for $d$ as in Eq.~\eqref{equ:degree_polynomial} the error decays like a polynomial of high-degree in $\epsilon$. 
But that rough approximation will be sufficient for our purposes.

Cor.~\ref{cor:expectation_values_low_degree_poly} is an important step to prove our Thm.~\ref{thm:polynomialswork}, but still does not correspond to the exact statement we wish to prove. 
This is because Cor.~\ref{cor:expectation_values_low_degree_poly} is a statement about the \emph{local, truncated} evolution, whereas Thm.~\ref{thm:polynomialswork} is a statement about the \emph{global} evolution being well-approximated by a polynomial of small degree.
The strategy to go from the local to the global evolution, is to show that for the (local) observables required for our protocl, the local evolution approximates the global one well.

Our main tool to show this  approximatability of expectation values are Lieb-Robinson bounds~\cite{hastings_locality_2010,poulin_lieb-robinson_2010,barthel2012quasilocality,bach_lieb-robinson_2014}, which exactly give conditions under which  the local time evolution and the global one are close for small enough times and local observables. In order to provide a self-contained presentation, we include a brief introduction to Lieb-Robinson bounds in Sec.~\ref{sec:lieb-robinson} of this appendix.

In fact, there are various ways of quantifying this idea of local approximability and, thus, LR-bounds come in various forms. 
The version that we are going to work with here and which is discussed in detail in Sec.~\ref{sec:lieb-robinson}, considers an observable $O_Y$ initially supported on in a region $Y$. For a region $B\supset Y$ we set $R=\operatorname{dist}(\Gamma\backslash\{B\},Y)/k$, where $k$ is the locality of the generator, then it is shown in Lemma~\ref{lem:LR_locDyn} that indeed
\begin{align}\label{equ:LRexpo}
\|(e^{t\cL_\Gamma}-e^{t\cL_{\Lambda_r(Y)}})(O_Y)\|_{\infty}\leq (e^{vt}-1)h(R),
\end{align}
where $h:\R^+\to\R^+$ is a monotonically increasing function such that $\lim_{R\to+\infty}h(R)=0$ and $v$ is some constant that depends on the generator, usually called the LR-velocity. 
On the other hand, the decay of the function $h$ typically depends on how fast the interactions in the system decays spatially (i.e. if it is strictly local, exponentially decaying in the distance or even algebraically decaying) and the geometry of the underlying lattice. However, the important point for our purposes is that it does not depend on the system size. 
For the specific case of short-range Hamiltonians discussed in the main text, we have that $h(R)=e^{-\mu R}$ for some constant $\mu>0$. For algebraically decaying evolutions we usually have $h(x)=\mathcal{O}(x^{-k})$ for some $k\in\mathbb{R}^+$.

We refer again to Sec.~\ref{sec:lieb-robinson} for a discussion of various LR-bounds available in the literature. But from Eq.~\eqref{equ:LRexpo} we immediately conclude that the values of the expectation values of global and local evolutions are well-approximated by each other. More precisely:
\begin{prop}\label{prop:lieb-robinson_approx}
Let $O_Y$ be an observable supported on some region $Y$, $\epsilon>0$ and $t_{\max}$ be given. Assume Eq.~\eqref{equ:LRexpo} holds for the time evolution $\linGen_{\Gamma}$ and a function $h$.
Let $l>0$ be given by
\begin{align*}
    l=h^{-1}\left(\frac{\epsilon}{e^{t\max}-1}\right).
\end{align*}
Then we have for $\Lambda_l(Y)$ and all $0<t<t_{\max}$ and any initial state $\rho$ that
\begin{align}\label{equ:expectation_diff_small}
    |\tr{\left(e^{t\cL_\Gamma}-e^{t\cL_{\Lambda_l(Y)}}\right)(O_Y)\rho}|\leq \epsilon.
\end{align}
\end{prop}
\begin{proof}
The claim follows directly from Eq.~\eqref{equ:LRexpo} or Lemma~\ref{lem:LR_locDyn} and a H\"older inequality. Indeed, for the value of $l$ in Eq.~\eqref{equ:expectation_diff_small}, we obtain from Eq.~\eqref{equ:LRexpo} after some simplification that
\begin{align}
    &\|(e^{t\cL_\Gamma}-e^{t\cL_{\Lambda_r(Y)}})(O_Y)\|_{\infty}\leq\epsilon.
\end{align}
\end{proof}
Note that in the case of short-range systems we have that $h^{-1}(x)=\mu^{-1}\log(x)$.
From now on we will suppress the terms of order $\log(\log(\epsilon^{-1}))$ or higher from the equations and denote bounds where we do this with $\tcO$.
Thus, combining~\ref{lem:local_derivatives_bounded} with Prop.~\ref{prop:lieb-robinson_approx} we conclude that:
\begin{theorem}\label{thm:polynomialswork}
Let $\linGen_{\Gamma}$ be a locally bounded Lindbladian on a $D$-dimensional regular lattice with constant $g$. Moreover, let $t_{\max},\epsilon>0$ be given and $O_Y$ and observable such that $\|O_Y\|\leq 1$ and $O_Y$ is supported on a constant number of qubits. Assume that $\linGen_{\Gamma}$ satisfies Eq.~\eqref{equ:LRexpo}. Then there is a polynomial $p$ of degree 
\begin{align*}
    d=\tcO\left[\left(h^{-1}\left(\frac{\epsilon}{e^{vt\max}-1}\right)\right)^{D}t_{\max}\log(\epsilon^{-1})\right]
\end{align*}
such that for all $0\leq t\leq t_{\max}$:
\begin{align}\label{equ:quality_approx}
    \left|\tr{e^{t\cL_{\Gamma}}(O_Y)\rho}-p(t)\right|\leq\epsilon.
\end{align}
and
\begin{align}\label{equ:first_derivative_polynomial}
p'(0)=\tr{\cL_{\Gamma}(O_Y)\rho}.
\end{align}
\end{theorem}
\begin{proof}
It follows from Prop.~\ref{prop:lieb-robinson_approx} that a region of radius
$$l=\tcO\left[h^{-1}\left(\frac{\epsilon}{2(e^{t\max}-1)}\right)\right]$$ is enough to approximate the time evolution of $e^{t\cL_{\Gamma}}(O_Y)$ up to $\epsilon/2$. If the original region $Y$ has a constant number of qubits, then for a $D$-dimensional lattice we have $|\Lambda_l(Y)|=\cO(l^D)$.
It then follows from Cor.~\ref{cor:expectation_values_low_degree_poly} that for regions of this size, it is sufficient to pick a degree that is $$\cO\left[h^{-1}\left(\frac{\epsilon}{2(e^{t\max}-1)}\right)^{D}t_{\max}\log(\epsilon^{-1}))\right]$$ to approximate the expectation value of the local evolution up to an error $\epsilon/2$. This concludes the proof by a triangle inequality. Eq.~\eqref{equ:first_derivative_polynomial} is clear from properties of the truncated Taylor series.
\end{proof}
Thus, we see that as long as the time evolution of the system satisfies a Lieb-Robinson bound, we can approximate the expectation value of a local observable as a function of time by a polynomial whose degree is dictated by how fast the function $h$ decays and the maximal time of the evolution. In particular, for $\epsilon^{-1}=\cO(1)$ and time $t_{\max}=\cO(1)$, we conclude that the degree of the polynomial is independent of the system's size.

In Thm.~\ref{thm:polynomialswork} we established that we can approximate the function $f:t\mapsto\tr{e^{t\cL_{\Gamma}}(O_Y)\rho}$ well by a polynomial for constant times. However, to estimate the parameters of the Hamiltonian we are ultimately interested in the derivative of $f$ at time $0$. We will later show in Sec.~\ref{sec:robust_interpolation} that for the special case of polynomials of bounded degree, a good recovery of the polynomial also implies a good recovery of the derivative. In particular, as for the polynomial in Lemma~\ref{lem:local_derivatives_bounded} we have that $p'(0)=f'(0)$, it is sufficient to argue that any two polynomials that approximates the curve $f(t)$ up to sufficiently large precision in a sufficiently large number of points must have close derivatives at $0$ as well. This will be the subject of the next section and proved in Prop.~\ref{prop:precision_robust}.

\section{(Robust) polynomial interpolation and derivative estimation}\label{sec:robust_interpolation}
In this section of the appendix, we are going to review a result in the literature~\cite{kane_robust_2017} that shows how to perform polynomial interpolation in a robust way even in the presence of outliers. 
Furthermore, we will show that good polynomial interpolation also implies a good approximation of derivatives of the polynomial, which is our end goal. 
We will use and review the results and algorithms of~\cite{kane_robust_2017} for the robust polynomial interpolation and resort to Markov brothers' inequality~\cite{markoff_ber_1916} for estimating the error on the derivatives.

Let us start by briefly recalling the technical problems we wish to overcome. We assume we are able to approximate the expectation value of $f(t)=\tr{\rho e^{t\linGen}(O)}$ for some suitably-picked initial state $\rho$ and time-evolved observable $O$. As argued in Sec.~\ref{sec:quality_approx_poly}, the function $f$ is well-approximated by a low degree polynomial $p$ whenever the time evolution is generated by a local Hamiltonian. 
Moreover, as shown in Sec.~\ref{sec:howtoextractparameters}, by suitably choosing the observable and initial state, we can easily read off the value of the coupling of the Hamiltonian from the value of $f'(0)$. 
Thus, our goal is to find a polynomial $p$ that approximates $f$ from values of $f(t_i)$ for some $t_i\in[a,b]$ for then use $p$ to infer $f'(0)$.

It is well-known that if $p$ is a polynomial of degree $d$, then it is uniquely determined by its values at $d+1$ points. Thus, one could naively expect that having access to $f(t_1),\ldots,f(t_{\tilde{d}})$ points for $\tilde{d}\sim d$ times is sufficient to reconstruct $d$.

However, the present situation exhibits three challenges that need to be overcome to ensure that we can reliably apply polynomial interpolation methods and recover $p$ from points $f(t_i)$:
\begin{enumerate}
\item we can only estimate $f(t)$, and not $p(t)$. And the value of $f$ only approximates that of $p$ up to some error $\epsilon_a$, as discussed in Thm.~\ref{thm:polynomialswork}.
\item we do not have access to the value of $f(t)$ directly, but can only approximate it to a precision $\epsilon_s$ by sampling from the output of the device at time $t$ $\cO(\epsilon_s^{-2})$ times.
\item we are interested in the value of $p'(0)$ and not in the polynomial $p$ itself. Thus, we need to ensure a small error in estimating the derivative.
\end{enumerate}

To deal with the first two problems the polynomial interpolation technique we use has to be robust to the noise stemming from both the approximation error from the polynomial approximation and the statistical noise. 
To deal with the third issue, we will show that we need to pick the final and initial time of the interpolation in a judicious manner. 

To obtain some intuition about how to pick the times, let us consider the case of estimating the derivative of a quadratic polynomial at $0$. I.e., if we have two linear functions $p,\hat{p}$ that are $\epsilon$ close in some interval $[a,b]$, how well can we infer the derivative of $p$ at $0$ from that of $\hat{p}$? First, note that if the interval $[a,b]$ is very small, then two functions can differ by $\epsilon$ and their derivatives can still differ by $\epsilon/(b-a)$ even for linear functions. 
This indicates we probably do not want to pick $b-a$ too small.
However, as we know that increasing $b$ also implies that, in our setting, we need to increase the degree of the polynomial for the fitting, which in turns increases the number of points we need to estimate, this hints at the fact that picking $b-a$ of constant order will be optimal. 

On the other hand, it is also clear that the closer $a$ is to $0$, the more information about the value of $p'(0)$ we can infer from the interpolation. Thus, this discussion suggests that picking $a$ as close to $0$ as possible and $b$ of constant order should give the best results. We will prove this intuition later in this section, but first will discuss robust interpolation.

Directly interpolating through the noisy data can be an unstable procedure if we do not pick the interpolating points wisely and perform a suitable regression. 
Recent results have shown how to perform polynomial interpolation in an essentially optimal fashion in a robust way even with a fraction of the points being outliers~\cite{kane_robust_2017}. Let us now review the results of~\cite{kane_robust_2017}.

We will now assume we wish to estimate a polynomial $p:[-1,1]\to\R$ of degree $d$, as this corresponds to the setting of~\cite{kane_robust_2017}. 
Note that for Hamiltonian learning, we will be interested in the case where the domain is of the form $[a,b]$ for $a,b\geq0$. 
However, we can simply shift and rescale the domain to $[-1,1]$. When we summarize our results later, we will dicuss the effect of this rescaling explicitly.
We will assume we are given access to $m$ random samples $(x_i,y_i)$ of points such that a fraction of at least $\alpha>1/2$ of them satisfies for some $\sigma>0$ that:
\begin{align}\label{equ:samples_noisy}
    p(x_i)=y_i+w_i,\quad |w_i|\leq\sigma.
\end{align}
There are results available for various different ways of sampling the points $x_i$. However, the best available sample complexity is given by sampling from the Chebyshev measure, which has density
\begin{align*}
    \frac{1}{\pi\sqrt{1-x^2}}
\end{align*}
on the interval $[-1,1]$.
We then have:
\begin{theorem}[Robust polynomial interpolation]\label{thm:poly_interpolation_robust}
Let $p:[-1,1]\to\R$ be a polynomial of degree $d$ and assume we are given $m$ samples $(x_i,y_i)$ such that a fraction $\alpha>1/2$ of them satisfies Eq.~\eqref{equ:samples_noisy} for some $\sigma>0$. Moreover, suppose that the $x_i$ were sampled independently and at random from the Chebyshev measure. Then for any $\delta>0$
\begin{align}\label{equ:number_samples_poly}
    m=\cO\left(d\log \left(\tfrac{d}{\delta}\right)\right)
\end{align}
samples suffice to with probability of success at least $1-\delta$ recover a polynomial $\hat{p}$ that satisfies:
\begin{align}\label{equ:infinity_norm_polynomial}
    \max\limits_{x\in[-1,1]}\left|p(x)-\hat{p}(x)\right|\leq3\sigma.
\end{align}
Moreover, $\hat{p}$ can be computed in time polynomial in the number of samples $m$.
\end{theorem}
\begin{proof}
We refer to~\cite[Corollary 1.5]{kane_robust_2017} for a proof and note that we obtain the statement by setting the parameter $\epsilon=1/2$ in their statement.
\end{proof}
We note that the same result holds for random points picked from the uniform measure with $m=\cO(d^2)$.

The result above solves our problem of robust polynomial interpolation outlined in points 1 and 2. 
It shows that it if we can ensure that we can approximate sufficiently many points of the polynomial up to some $\sigma$, then we also recover the whole polynomial up to some error proportional to $\sigma$. 
Moreover, the number of sampled required only has a logarithmic overhead in $d$ when compared with the case where we know the points exactly. As we will see later, for our puproses it will be important to choose $d$ to be small. Thus, in a nutshell, we see that Thm.~\ref{thm:poly_interpolation_robust} ensures that we can reliably and robustly perform polynomial interpolation by only a small overhead when compared to when we know the points exactly.

We will later describe in more detail the algorithm given in~\cite{kane_robust_2017} whose output satisfies the promises of Thm.~\ref{thm:poly_interpolation_robust}. However, before that we will show how the condition in Eq.~\eqref{equ:infinity_norm_polynomial} ensures that we can also recover the derivative of the polynomial as long as the degree $d$ is small.

To do that, we will resort to Markov brothers' inequality, which we restate now for completeness.

\begin{lem}[Markov brothers' inequality]\label{lem:Markov_brother}
For $d,k\in\N$ define the constant $C_M(d,k)$ to be given by
\begin{align}
    C_M(d,k)=\frac{d^{2}\left(d^{2}-1^{2}\right)\left(d^{2}-2^{2}\right) \cdots\left(d^{2}-(k-1)^{2}\right)}{1 \cdot 3 \cdot 5 \cdots(2 k-1)}.
\end{align}
Then for any polynomial $p$ of degree $d$ we have that:
\begin{align}\label{equ:Markov_brothers}
    \max\limits_{x\in[-1,1]}\left|p^{(k)}(x)\right|\leq C_M(d,k)\max\limits_{x\in[-1,1]}\left|p(x)\right|.
\end{align}
\end{lem}
\begin{proof}
We refer to~\cite[Theorem 1.2]{rassias_mathematical_2016} for a proof and discussion of this result.
\end{proof}
Note that the value of $CM(d,k)$ increases exponentially with $d$ for $k$ constant.
We remark that having further promises on the structure of the polynomial, such as the location of its zeros, can greatly improve this estimate. We once again refer to~\cite[Chapter 1]{rassias_mathematical_2016} for a discussion on this. It would be interesting to see if recent results on the analyticity of the partition function~\cite{2108.04842} could be used in our context to also improve this estimate, as it grows exponentially with $d$. However, this general bound will suffice for our purposes.

It is easy to see that for polynomials defined on some interval $[a,b]$, the polynomial $\tilde{p}(x)=p(\tfrac{b-a}{2}x+\tfrac{a+b}{2})$ is defined on $[-1,1]$ and we can use this simple transformation to obtain a variation of Eq.~\eqref{equ:Markov_brothers} for polynomials defined on general intervals. Indeed, applying Eq.~\eqref{equ:Markov_brothers} to $\tilde{p}$, it follows from a straightforward application of the chain rule that
\begin{align}\label{equ:Markov_brothers_shifted}
    \max\limits_{x\in[a,b]}\left|p^{(k)}(x)\right|\leq\left|\frac{2}{(b-a)}\right|^k C_M(d,k)\max\limits_{x\in[a,b]}\left|p(x)\right|.
\end{align}
From this we conclude that:
\begin{lem}[Extrapolating the derivative at $0$]\label{prop:derivative_0}
Let $p:[0,b]\to\R$ be a polynomial of degree $d$ such that for some $\epsilon>0$ and $0<a<b$:
\begin{align}\label{equ:poly_small}
    \max\limits_{x\in[a,b]}|p(x)|\leq \epsilon.
\end{align}
Then
\begin{align}\label{equ:value_poly_small_rest_interval}
    \max\limits_{x\in[0,a]}|p(x)|\leq \epsilon\left(\sum\limits_{k=0}^{d}\left|\frac{2}{(b-a)}\right|^ka^{k} \frac{C_M(d,k)}{(k)!}\right)
\end{align}
and
\begin{align}\label{equ:dervative_at_0}
    |p'(0)|\leq \epsilon\left(\sum\limits_{k=1}^{d}\left|\frac{2}{b-a}\right|^ka^{k-1} \frac{C_M(d,k)}{(k-1)!}\right).
\end{align}
\end{lem}
\begin{proof}
It follows from Eq.~\eqref{equ:poly_small} and Markov brothers' inequality that 
\begin{align}\label{equ:derivatives_small}
   \left|p^{(k)}(a)\right|\leq\left|\frac{2}{b-a}\right|^k C_M(d,k)\epsilon.
\end{align}
By a Taylor expansion we know that for $x\in[0,a]$:
\begin{align}\label{equ:taylor_original}
    p(x)=\sum\limits_{k=0}^{d}p^{(k)}(a)\frac{(x-a)^{k}}{k!}.
\end{align}
The claim in Eq.~\eqref{equ:value_poly_small_rest_interval} then follows by combining this expansion with Eq.~\eqref{equ:derivatives_small} and a triangle inequality. Similarly we have
\begin{align}\label{equ:taylor_p}
    p'(0)=\sum\limits_{k=1}^{d}p^{(k)}(a)(-1)^{k}\frac{a^{k-1}}{(k-1)!},
\end{align}
for which a similar argument yields Eq.~\eqref{equ:dervative_at_0}.
\end{proof}
The proposition above essentially allows us to control to what extent the derivative of a polynomial at $0$ can deviate from $0$ given that the polynomial is small on another interval $[a,b]$. We can then apply it to the polynomial $p-\hat{p}$, as in Eq.~\eqref{equ:infinity_norm_polynomial} to control the error we make by estimating the derivative at $0$ by evaluating $\hat{p}(0)$.

By combining the arguments above we conclude that:
\begin{prop}[Precision and number of samples for robust interpolation]\label{prop:precision_robust}
Let $p$ be a polynomial of degree $d$. For some $0<a<b$ define 
\begin{align}\label{equ:definition_ead}
    E(a,b,d)=\left(\sum\limits_{k=1}^{d}\left|\frac{2}{b-a}\right|^ka^{k-1} \frac{C_M(d,k)}{(k-1)!}\right).
\end{align}
Then for $\delta>0$ sampling 
\begin{align}
    m=\cO(d\log(d\delta^{-1}))
\end{align}
i.i.d. points $(x_i,y_i)$ from the Chebyshev measure on $[a.b]$ satisfying
\begin{align*}
    p(x_i)=y_i+w_i,\quad |w_i|\leq\sigma 
\end{align*}
for at least a fraction $\alpha>\tfrac{1}{2}$ of the points is sufficient to obtain a polynomial $\hat{p}$ satisfying
\begin{align}\label{equ:final_estimate_derivative}
    \left|p'(0)-\left(\hat{p}\right)'(0)\right|\leq3\sigma E(a,b,d).
\end{align}
\end{prop}
\begin{proof}
From Thm.~\ref{thm:poly_interpolation_robust} we know that this number of samples suffices to obtain the polynomial $\hat{p}$ satisfying 
\begin{align*}
   \max\limits_{x\in[a,b]} \left|p(x)-\hat{p}(x)\right|\leq 3\sigma.
\end{align*}
Applying Lemma~\ref{prop:derivative_0} to $p-\hat{p}$ yields the claim.
\end{proof}
This then yields a simple condition on how small $\sigma$ has to be in the regime of interest to us:
\begin{cor}\label{cor:obtaining_derivative_error}
In the same setting as in Prop.~\ref{prop:precision_robust} for some $\epsilon>0$ let $a\leq d^{-2}$, $b=2+a$ and
\begin{align}\label{equ:error_degree}
   \sigma=\epsilon d^{-2}.  
\end{align}
Then 
\begin{align}
|p'(0)-\hat{p}'(0)|=\cO(\epsilon).
\end{align}
\end{cor}
\begin{proof}
It is easy to see that we have:
\begin{align*}
    C_M(d,k)\leq \frac{d^{2k}}{k!!},
\end{align*}
where $k!!=1\times 3\times 5\times\cdots\times (2k-1)$ is the double factorial.

Thus, we see from this and Eq.~\eqref{equ:dervative_at_0} that by our choice
\begin{align*}
    a=\frac{1}{d^2},\quad b=2+a
\end{align*}
we have from Eq.~\eqref{equ:final_estimate_derivative} that the estimated polynomial $\hat{p}(0)$ satisfies
\begin{align*}
    \left|\hat{p}'(0)-p'(0)\right|\leq 3\sigma \left(\sum\limits_{k=1}^{d}\left|\frac{2}{(a-b)}\right|^ka^{-1} \frac{(a d^2)^k}{k!!(k-1)!}\right)\leq 3\sigma \left(\sum\limits_{k=1}^{d}a^{-1} \frac{1}{k!!(k-1)!}\right),
\end{align*}
where we used the fact that $a\leq d^{-2}$.

Thus, as 
\begin{align*}
    \left(\sum\limits_{k=1}^{d}\frac{1}{k!!(k-1)!}\right)\leq e,
\end{align*}
we conclude that with this choice of parameters we have
\begin{align*}
    \left|\hat{p}'(0)-p'(0)\right|\leq 3e\sigma a^{-1},
\end{align*}
which gives the claim.
\end{proof}

We will discuss in Sec.~\ref{sec:performance_guarantee} how to specialize the discussion and results above to the scenario of Hamiltonian learning.

\section{Choice of parameters and performance guarantee of the protocol}\label{sec:performance_guarantee}
Let us now combine the results from Sections~\ref{sec:quality_approx_poly} and~\ref{sec:robust_interpolation} to see how to pick the various parameters of the algorithm to ensure a good recovery of the couplings of the Hamiltonian. 

More precisely, given a coupling parameter $a_{i,j}$ of $\cL$ we will be interested in estimating the sample complexity of obtaining an estimate $\hat{a}_{i,j}$ satisfying
\begin{align}\label{equ:looerror}
    \left|a_{i,j}-\hat{a}_{i,j}\right|\leq\epsilon
\end{align}
with high probability for some given error $\epsilon$. 
As extensively discussed by now, we can easily reduce estimating the couplings to estimating derivatives of time evolutions of local observables.

As expected, we will see that the main parameters we need to control are the maximal observation time $t_{\max}$ and the initial time of measurement $t_0$.
This is showcased in the following Theorem:
\begin{thm}[Choice of final and initial time]\label{thm:choice_initial_final}
Let $\cL_{\Gamma}$ be a locally bounded Lindbladian on a $D$-dimensional regular lattice with $g=\cO(1)$ growth constant. Let $O$ be an observable of constant support and $\rho$ and arbitrary quantum state.
Let $\epsilon>0$ be given. Assume that $\cL_{\Gamma}$ satisfies Eq.~\eqref{equ:LRexpo} for some function $h$. Then picking $t_0$ as
\begin{align}
    t_0=\cO\left[\left(h^{-1}\left(\frac{\epsilon}{2(e^{2.5v}-1)}\right)^{D}\log(\epsilon^{-1})\right)^{-2}\right]
\end{align}
and $t_{\max}=2+t_0$ and measuring the expectation value $f(t)=\tr{e^{t\cL_{\Gamma}}(O)\rho}$ for 
\begin{align*}
    m=\tcO\left[\left(h^{-1}\left(\frac{\epsilon}{2(e^{2.5v}-1)}\right)\right)^{D}\log(\epsilon^{-1}))\right]
\end{align*}
random times $t_i\in[t_0,t_{\max}]$ up to precision $\cO(\epsilon)$ is sufficient to obtain an estimate $(\hat{f})'(0)$ satisfying
\begin{align}\label{equ:final_error}
    \left|(\hat{f})'(0)-\tr{\cL_{\Gamma}(O)\rho}\right|\leq 3e\epsilon t_0^{-1}=\cO\left[\epsilon\left(h^{-1}\left(\frac{\epsilon}{2(e^{2.5v}-1)}\right)^{D}\log(\epsilon^{-1})\right)^{2}\right].
\end{align}
In particular, this estimate can be obtained from 
\begin{align}
    \tcO(\epsilon^{-2}\log(\delta^{-1}))
\end{align}
samples from the time evolved state $\rho$ with probability of success at least $1-\delta$.
\end{thm}
\begin{proof}
First, note that by Lemma.~\ref{cor:expectation_values_low_degree_poly}, if we pick $t_{\max}$ as described above, then a polynomial $p$ of degree 
\begin{align}
    d=\cO\left[\left(h^{-1}\left(\frac{\epsilon}{2(e^{2.5v}-1)}\right)\right)^{D}\log(\epsilon^{-1}))\right]
\end{align}
is sufficient to approximate the expectation value in the interval $[0,t_{\max}]$ up to an error $\epsilon/2$, as $t_{\max}\leq 2.5$. 
We will estimate the value of the polynomial at each point up to an error $\sigma>0$, which is to be determined later.

Thus, by inserting the bound on the degree $d$ in Eq.~\eqref{equ:error_degree}, we need to estimate each value of the polynomial up to a precision $\sigma=\cO(\epsilon)$ to obtain an overall error of
\begin{align}\label{equ:error_estimate_h}
\cO\left[\epsilon\left(h^{-1}\left(\frac{\epsilon}{2(e^{2.5v}-1)}\right)^{D}\log(\epsilon^{-1})\right)^{2}\right]
\end{align}

As we imposed that the precision with which the polynomial approximates the expectation values is $\epsilon/2$, we can estimate the value of the polynomial for a given time up to an error $\cO(\epsilon)$ from $\cO(\epsilon^{-2})$ samples.

As we have to sample
\begin{align*}
    \cO(d\log(d))=\tcO\left[\left(h^{-1}\left(\frac{\epsilon}{2(e^{2.5v}-1)}\right)^{D}\log(\epsilon^{-1})\right)\right]
\end{align*}
points to perform the stable interpolation, we obtain the advertised sample complexity.
\end{proof}

For the case of strictly local or exponentially decaying interactions we have that 
\begin{align}
    h^{-1}\left(\frac{\epsilon}{2(e^{2.5v}-1)}\right)=\textrm{poly log}(\epsilon^{-1}).
\end{align}
In that case the sample complexity is of order $\tcO(\epsilon^{-2})$. Thus, in this case we see that the inverse initial time $t_0^{-1}$ and the number of points we need to sample from is polylogarithmic in $\epsilon$. 
Furthermore, the sample complexity to obtain an error $\epsilon$ is also $\tcO(\epsilon^{-2})$ up to polylogarithmic corrections.

For the sake of completeness, let us now discuss the conditions under which our protocol works beyond the setting of exponentially decaying or short-range interactions.
From Eq.~\eqref{equ:final_error} the condition for our procedure to work becomes transparent: we need that
\begin{align}
    \left(h^{-1}\left(\frac{\epsilon}{2(e^{2.5v}-1)}\right)^{D}\log(\epsilon^{-1})\right)^{-2}=o(\epsilon^{-1}).
\end{align}
Indeed, in this case we have the property that it is possible to suitably re-scale the error $\epsilon$ to ensure that the total precision is at some desired precision $\tilde{\epsilon}$. For instance, let us assume that
\begin{align}
    h^{-1}\left(\frac{\epsilon}{2(e^{2.5v}-1)}\right)=\cO(\epsilon^{-r}),
\end{align}
for some $r>0$. As we discuss later, this is typically the case for algebraically decaying interactions. For such a LR-bound, we see that the resulting error in Eq.~\eqref{equ:final_error} is
\begin{align}
    \cO(\epsilon^{1-2Dr}\log(\epsilon^{-1})^2).
\end{align}
Ignoring the $\log(\epsilon^{-1})$ term, we see that by picking $\epsilon=\tilde{\epsilon}^{\frac{1}{1-2Dr}}$ we can ensure an error of order $\tilde{\epsilon}$ for the estimate. Thus, the growth of $h^{-1}$ has to be at most $r\leq \tfrac{1}{2D}$ and the sample complexity would also grow like $\tcO(\epsilon^{-2-2Dr-r})$, as we would need to sample $\tcO(d)=\tcO(\epsilon^{-r})$ points up to precision $\tilde{\epsilon}^{\frac{1}{1-2Dr}}$.

Thus, we see that our protocol has a sample complexity that is independent of the system size to estimate one parameter and the expected $\epsilon^{-2}$ scaling for short range evolutions evolutions, up to log factors. For algebraically decaying interactions, however, the sample complexity has a worse scaling that depends on the exact decay of the potential, but still independent of system size. 

We summarize the sample complexities, smallest initial time $t_0$ and number of different times steps we need for various different potentials in Table~\ref{tab:summary-resources}. 
\begin{table}
\begin{center}
\begin{tabular}{ |c|c|c|c| }
\hline
  & Sample Complexity & Number of points & Initial time \\ 
 \hline
 Finite range & $\epsilon^{-2}$ & $\textrm{polylog}(\epsilon^{-1})$ & $\textrm{polylog}(\epsilon)$\\  
 \hline
  Exponentially & $\epsilon^{-2}$ & $\textrm{polylog}(\epsilon^{-1})$ & $\textrm{polylog}(\epsilon)$\\  
 \hline
 Algebraically ($\alpha\geq5D-1$) & $\epsilon^{-2\frac{\alpha-3D}{\alpha-5D}}$ & $\epsilon^{-\frac{D}{\alpha-5D}}$ & $\epsilon^{\frac{2D}{\alpha-5D}}$\\  
 \hline

\end{tabular}
\caption{scaling of different resources required to obtain a recovery up to additive error $\epsilon$ of a parameters of the evolution. We have only included the leading order term and $\alpha$ denotes the decay of the potential in space, whereas $D$ the dimension of the lattice.}\label{tab:summary-resources}
\end{center}
\end{table}
\subsection{Algorithm for robust polynomial interpolation}
Now that we have established that the results of~\cite{kane_robust_2017} indeed allow us to estimate the derivative at $0$, let us now describe their polynomial interpolation algorithm in more detail for completeness. The algorithm consists of two parts, one $\ell_1$ regression and an iteration of $\ell_\infty$ regressions. Following~\cite{kane_robust_2017}, we will only consider the case in which we interpolate over $[-1,1]$. But it is straightforward to also interpolate over other intervals by a suitable affine transformation of the domain, as discussed before.
\paragraph{$\ell_1$ regression:} before we define the $\ell_1$ regression, we need to define the Chebyshev partitions:
\begin{defi}[Chebyshev partitions]
Let $m\in\N$ be given. The size $m$ Chebyshev partitions of $[-1,1]$ is the set of intervals $I_j=\left[\cos\frac{\pi j}{m},\cos\frac{\pi(j-1)}{m}\right]$ for $1\leq j\leq m$.
\end{defi}
We also define $\mathcal{P}_d$ to be the space of polynomials of degree at most $d$.

With these definitions at hand, we define the $\ell_1$ regression solution as follows:
\begin{defi}\label{equ:l1regression}
Given a set of $n$ points $(x_i,y_i)$ and $m\in\N$, we define the result of the degree $d$ $\ell_1$ regression with $m$ Chebyshev partitions $\hat{p}$ to be the polynomial
\begin{align*}
    \operatorname{argmin}_{\hat{p}\in\mathcal{P}_d} \sum\limits_{i=1}^n|I_j|\operatorname{mean}_{x_i\in I_j}\left|y_i-\hat{p}(x_i)\right|,
\end{align*}
where $\mathcal{P}_d$ is the set of polynomials of degree $d$.
\end{defi}
Note that the optimization problem above is a linear program and, thus, can be solved efficiently. Solving the $\ell_1$ regression problem with $n=\cO(d\log(d))$ samples from the Chebyshev measure is guaranteed to give us a good solution on average. More precisely, as shown in~\cite[Lemma 1.2]{kane_robust_2017}, the solution is guaranteed to satisfy
\begin{align*}
    \|p-\hat{p}\|_{\ell_1}=\cO(\sigma),
\end{align*}
where as usual $\sigma$ is the error in each estimate $y_i$ and 
\begin{align*}
    \|p-\hat{p}\|_{\ell_1}=\int\limits_{-1}^1\left|p(x)-\hat{p}(x)\right|dx.
\end{align*}
However, the results of the previous sections required us to obtain a good solution in the $\|\cdot\|_{\ell_\infty}$ distance, and in general
\begin{align}\label{equ:l1toloo}
\|p-\hat{p}\|_{\ell_1}=\cO(d^2\|p-\hat{p}\|_{\infty}).
\end{align} 
Although, as commented in the last section, we are interested in the regime of polynomial of relatively small degree, by adding a $\ell_\infty$ regression iteration on top of the $\ell_1$ regression, it is possible to get rid of this $d^2$ prefactor.

\paragraph{$\ell_\infty$ regression:} besides getting rid of the unwanted $d^2$ factor on the promise for the error of the $\ell_1$ regression, adding a $\ell_\infty$ regression step also has the favourable feature of making the whole procedure more robust to outliers in the data. 
\begin{defi}[$\ell_\infty$ regression]
Given a set of $n$ points $(x_i,y_i)$ and $m\in \N$ given. For the $m$ Chebyshev partitions $I_j$, choose $\tilde{x}_j\in I_j$ arbitrarily and let 
\begin{align*}
    \widetilde{y}_{j}=\operatorname{median}_{x_{i} \in I_{j}} y_{i}.
\end{align*}
We define the result of the degree $d$ $\ell_\infty$ regression with $m$ Chebyshev partitions $\hat{p}$ to be the polynomial $\hat{p}\in\mathcal{P}_d$
\begin{align}\label{equ:looregression}
     \operatorname{argmin}_{\hat{p}\in\mathcal{P}_d} \max _{j \in[m]}\left|\widehat{p}\left(\widetilde{x}_{j}\right)-\widetilde{y}_{j}\right|.
\end{align}
\end{defi}
Note that the problem in Eq.~\eqref{equ:looregression} also corresponds to a linear program and, thus, can be solved efficiently. The output of the $\ell_\infty$ regression algorithm is guaranteed to satisfy 
\begin{align}
    \|\hat{p}-p\|_{\infty}\leq 2.5\sigma+\frac{1}{2}\|p\|_{\infty}
\end{align}
as long as $m=\cO(d)$ and the we pick $n=\cO(d\log(d))$ samples from the Chebyshev measure, as shown in~\cite[Lemma 1.3]{kane_robust_2017}. Thus, the procedure gives us a promise of recovery in $\infty$-norm up to the unwanted $\|p\|_{\infty}$ term. This can be solved by iterating the $\ell_\infty$ regression step.

\paragraph{Iterating the $\ell_\infty$ step:} the last step to obtain the desired robust polynomial interpolation is to iteratively apply the $\ell_\infty$ iteration step to the residual. More precisely, we first perform the $\ell_1$-regreesion on our data, obtaining a polynomial $\hat{p}_0$. We can then define the new data points
\begin{align}\label{equ:residual_points}
    (x_i,\tilde{y}_i^0=y_i-\hat{p}_0(x_i))
\end{align}
and run the $\ell_\infty$ interpolation on this residual error, obtaining a polynomial $\hat{p}_1$. From Eq.~\eqref{equ:l1toloo} and our promise on the output of the $\ell_\infty$ interpolation, we know that the result of the interpolation will satisfy
\begin{align*}
    \|p-\hat{p}_1\|_{\infty}\leq 2,5\sigma+\frac{1}{2}\cO(d^2\sigma). 
\end{align*}
But then we can iterate this procedure by just running the $\ell_\infty$ regression on
\begin{align}\label{equ:residual_points2}
    (x_i,\tilde{y}_i^1=y_i-\hat{p}_1(x_i)).
\end{align}
Each time we run the interpolation on the residual, we exponentially reduce the error. By repeating the procedure $\cO(\log_2(d))$ times, we then arrive at a polynomial satisfying the promises of Thm.~\ref{thm:poly_interpolation_robust}.

Note, however, that the procedure used in Sec.~\ref{sec:numerics} to demonstrate the viability of our method differs slightly from the ones discussed here. 
The main difference is that we used equally spaced time steps that were not random.
However, in spite of this difference, we still obtained high quality solutions.

\section{Lieb-Robinson bounds}\label{sec:lieb-robinson}
This section gives a brief overview of Lieb-Robinson bounds. In particular, we give more explicit formulas for the functions $h$ in \eqref{equ:LRgeneral_main_text} in terms of the decay of the interactions and the dimension of the lattice. Lieb-Robinson bounds are by now a standard tool in quantum many-body systems and quantum information theory and we refer to~\cite{Lieb_1972,poulin_lieb-robinson_2010,bach_lieb-robinson_2014,hastings_locality_2010,kliesch2014lieb,sweke2019lieb,cubitt2015stability,kuwahara_strictly_2020} for a more general overview over the mathematical background and some latest bounds for algebraically decaying interactions.

At the heart of any Lieb-Robinson bound is the intuitive idea that if interactions in a system happen locally this should imply a bound on how fast information can be transmitted. The usual way to codify this property for Hamiltonian systems in the Heisenberg picture is to give a bound on the operator norm of the commutator between a time-evolved observable initially located in region $Y$ and a second observable located in a region $X$ in the distance $\operatorname{dist}(X,Y)$ between the regions $X,Y$, where $d$ could refer to the lattice or graph distance, i.e. a bound of the form
\begin{align}
  \norm{[A_X, O(t)_Y]}\leq C h(\operatorname{dist}(X,Y)) (e^{vt}-1),
\end{align}
where $C$ will typically depend on the size of the regions $|X|$ and $|Y|$ as well as on the operator norms of $A$ and $B$. However, in the context of Markovian dynamics and master equations, the bound is usually generalized by substituting the super-operator $[A_X, \cdot ]$ for an arbitrary bounded super-operator $\mathcal{K}_X:\M_{2^n}\to\M_{2^n}$ supported on $X$ leading to a Lieb-Robinson-bound of the form
\begin{align}\label{app:eq:LRI}
  \norm{\mathcal{K}_X(O(t)_Y)}\leq C h(\operatorname{dist}(X,Y)) (e^{vt}-1),
\end{align}
with $C$ depending on $\norm{\mathcal{K}_X}_{\infty\rightarrow\infty,cb}$. 
However, if $\mathcal{K}_X$ is of the form $\mathcal{K}_X=[A_X, \cdot]$, we have $\norm{\mathcal{K}_X}_{\infty\rightarrow\infty,cb}\leq 2 \norm{A_X}_{\infty}$, which allows us to recover the commutator  \cite{barthel2012quasilocality,cubitt2015stability}. In the following, we consider a regular lattice $\Lambda$ and assume that the dynamics is generated by a Lindbladian that decomposes according to
\begin{align}\label{app:eq:locL}
  \cL = \sum_{X\subset\Lambda} \mathcal{L}_X\;.
\end{align}
Following \cite{nachtergaele2011lieb,poulin_lieb-robinson_2010,cubitt2015stability}, we define the maximal interaction strength $J= \sup_{X\subset\Lambda} \norm{\mathcal{L}_X}_{1\rightarrow 1,cb}$ as well as the decay behaviour of the interactions $\mu(r) = \sup_{X\subset\Lambda: \operatorname{diam}(X)=r} \frac{\norm{\mathcal{L}_X}_{1\rightarrow 1,cb} }{J}$ in terms of the stabilized 1-to-1-norm $\norm{T}_{1\rightarrow 1,cb} = \sup_n \norm{T\otimes \id_n}_{r\rightarrow 1}$. We can then characterize $\cL$ as finite range if $\mu(r)=0$ for $r\geq R>0$, exponentially decaying if $\mu(r)\leq e^{-\mu r}$ and algebraically decaying if $\mu(r)\leq (1+r)^{-\alpha}$ for $\alpha>0$ and state the following Lieb-Robinson-bound for Lindbladians
\begin{thm}[dissipative LR-bound \cite{cubitt2015stability}]\label{app:eq:LRbthm}
  Let $\cL$ be a Lindbladian of the form \eqref{app:eq:locL}, $O_Y$ an observable supported on $Y\subset\Lambda$ and $\mathcal{K}_X:\M_{2^n}\to\M_{2^n}$ with $\mathcal{K}_X(\id_X) = 0$. Then
  \begin{align}
     \norm{\mathcal{K}_X(O(t)_Y)}\leq \norm{\mathcal{K}_X}_{\infty\rightarrow\infty,cb}\norm{O_Y} \min( |X|,|Y|)  h(\operatorname{dist}(X,Y)) (e^{vt}-1),
  \end{align}
  with $h(r) = e^{-\nu r}$ for $\cL$ exponentially decaying or finite range and $h(r)=(1+r)^\nu$ if $\cL$ is algebraically decaying with $\alpha>2D+1$ with $\nu<\alpha-(2D+1)$. 
\end{thm}

As stated above, in this work, we require a slightly different formulation of the LR-bound as given in \eqref{equ:LRexpo_main_text}, namely

  \begin{align}
        \|(e^{t\cL_\Lambda}-e^{t\cL_{\Lambda_r(Y)}})(O_Y)\|\leq c_1 h(\operatorname{diam}(\Lambda_{r(Y)})) (e^{vt}-1),
        \end{align}
which reflects directly that the dynamics of the system can already be described by a generator $\cL_{\Lambda_{r(Y)}}$ restricted to a region $\Lambda_{r(Y)}$ of diameter $r$ around the initial support $Y$ of the observable $O_Y$. To convert a bound of the form \eqref{app:eq:LRI}, we follow the reasoning given in \cite{barthel2012quasilocality, cubitt2015stability}. 

We can express the difference of the dynamics generated by the full Lindblad generator $\cL$ as compared to a restriction $\cL_{\Lambda_r}$ to the subset $\Lambda_r\subset\Lambda$ according to
\begin{align}
  (e^{t \cL} - e^{t\cL_{\Lambda_r}})O_Y = -\int_{0}^t \text{d}s\, \partial_s \left(e^{s\mathcal{L}_{\Lambda_r}}e^{(t-s) \mathcal{L}}\right) O_Y = \int_{0}^t \text{d}s\, e^{s\mathcal{L}_{\Lambda_r}}\left( \mathcal{L}-\mathcal{L}_{\Lambda_r} \right)e^{(t-s) \mathcal{L}} O_Y
\end{align}
Taking norms on both sides, we therefore obtain an upper bound of the form
\begin{align}
\norm{(e^{t \linGen} - e^{t\cL_{\Lambda_r}})O_Y}\leq  \sum_{X\not\subset \Lambda_r} \int_{0}^t \text{d}s\, \norm{\cL_X e^{(t-s) \cL O_Y}} \;.
\end{align}
We notice, that the term inside the integral is exactly of the form of the left-hand side of \eqref{app:eq:LRI} with $\mathcal{K}_X = \cL_X$. Hence, we can insert the standard LR-bound for dissipative dynamics from \eqref{app:eq:LRI} here and are left with a combinatorial problem in terms of the decay bounds. This can be done explicitly for several standard interaction decays, such as finite range, exponentially decaying or algebraically decaying interactions \cite{barthel2012quasilocality, cubitt2015stability}. In particular, based on the Lieb-Robinson bound in Thm.~\ref{app:eq:LRbthm}, we obtain
\begin{lem}[\cite{cubitt2015stability}]\label{lem:LR_locDyn}
Let $\mathcal{L}$ be a Lindbladian of the form \eqref{app:eq:locL}, $O_Y$ an observable supported on $Y\subset\Lambda$ and $r>0$ then for $\mathcal{L}_{\Lambda_r(y)} = \sum_{X\subset \Lambda_r(y)} \Lambda_X$ we have 
\begin{align}
\norm{\left(e^{t\cL} - e^{t\cL_{\Lambda_r(y)}}\right) O_Y}\leq \norm{O_Y} |Y| J \frac{e^{v t} - 1 - vt}{v} \, h(r)
\end{align}
with $h(r)$ exponentially decaying in $r$ for $\cL$ finite range or exponentially decaying and $h(r)$ decaying as $(1+r)^{-\beta}$ if $\cL$ is algebraically decaying with $\alpha>2D+1$ and $\beta = \alpha - 3D$ for $\alpha\geq 5D-1$ and $\beta = \frac{1}{2}(\alpha-D-1)$ if $\alpha\leq 5D-1$.  
\end{lem}

We remark that all these bound give us the required independence of the right-hand side from the overall system size. We expect that with the help of recent more stringent estimates on Hamiltonians with algebraic decay, it will most likely be possible to extend and strengthen these bounds for other algebraic decays.

\section{Parallelizing the Measurements: shadow process tomography}\label{sec:shadows_parallel}
\par In this section we introduce a method to parallelize the estimation of the Pauli overlaps required for our protocol. For a quantum channel $\Phi_t$ and two colletions of Pauli strings $P_{a}^1,\ldots,P_{a}^{K_1}$ and $P_b^1,\ldots,P_{b}^{K_2}$ that have combined weight $\omega_a+\omega_b$ at most $\omega$ it allows us to estimate all the overlaps of the form
\begin{align}
2^{-n}\tr{P_a^j\Phi_t(P_b^k)}
\end{align}
up to an error $\epsilon$ with probability at least $1-\delta$ from a number of samples that grows like $\cO(3^w\log(K_1K_2\delta^{-1})\epsilon^{-2})$. Moreover, it only requires us to prepare simple Pauli eigenstates and measure in Pauli eigenbases. Recently two works~\cite{processtomo,processtomo2} have considered how to generalize the shadows protocol to the setting of process tomography. Unfortunately, it was noticed that the proofs in~\cite{processtomo,processtomo2} are incorrect. From a high-level perspective, the issue with their argument was that it was based on the Choi-Jamilkowski isomorphism to map the problem into a state tomography problem. However, this unfortunately incurs in an exponential prefactor the authors missed.

Here we present a correction of the proof that also obtains a better exponent for the sample complexity in terms of the locality. Thus, we believe that this section may be of independent interest.

The protocol we will introduce now allows us to estimate all the data required for our Hamiltonian learning protocol in parallel. Every round $i$ of the protocol is performed as follows:
\begin{enumerate}
\item Draw two Pauli strings $B_i,S_i\in\{X,Y,Z\}^n$ uniformly at random and also a sign vector $E_i=\{-1,+1\}^n$ uniformly at random. 
\item Prepare the quantum state $\rho_i=\otimes_{j=1}^n\phi_j(S_i,E_i)$, where $\phi_j(S_i,E_i)$ is an eigenstate of the $j-$th Pauli on the string $S_i$ corresponding to the eigenvalue in the $j$-th entry of $E_i$. 
\item Evolve $\rho_i$ by $\Phi_t$.
\item Measure in the Pauli basis defined by $B_i$. Denote the measurement outcome by $M_i$.
\item Output $(B_i,S_i,E_i,M_i)$.
\end{enumerate}

Let us now introduce some notation to explain how to postprocess the samples obtained from the protocol above.

\par Let $P_a$ and $P_b$ be the given Pauli operators on  $n$ qubits.
We say that a basis $B_i$ overlaps with $P_a$ if all qubits on which $P_a$ acts non trivially are also measured in the same basis.
For example, $P_a=X\otimes Y\otimes I$ and $B_i=\{X,Y,Z\}$ overlap. However, $P_a=X\otimes Y\otimes I$ and $B_i=\{X,Z,Z\}$ do not overlap.
Moreover, if the basis and $P_a$ overlap, we say that the measurement outcome $M_i$ overlaps positively if we measure a positive eigenstate of $P_a$.
Otherwise, we say it is negative.
\par We will also define a similar notion of the state overlap given the Pauli $P_b$, the basis $S_i$ and sign $E_i$. We say that $P_b$ and 
$(S_i,E_i)$ overlap positively if $S_i$ coincides with $P_b$ on all qubits it acts non trivially and the state $\rho_i$ is a positive eigenstate of $P_b$.
We say we overlap negatively if it is a negative eigenstate. We will also define $\omega(P_a)$ to be the number of qubits on which $P_a$ acts non trivially. 

\par We can now finally introduce a random variable given the Paulis  $P_a$,  $P_b$ and the data from the experiment $B_i$, $S_i$, $E_i$ and $M_i$. 
Let us define a function in terms of the outcomes and inputs of one round of the protocol:
\begin{eqnarray}\label{Xab}
X_{a,b}(B_i,S_i,E_i,M_i)= \begin{cases}
 0&\text{ if $B_i$ does not overlap with $P_a$ or $S_i$ does not overlap with $P_b$,}\\
3^{\omega(P_a)+\omega(P_b)}/2 &\text{ if $(B_i,S_i)$  overlap with $(P_a,P_b)$, and both do so positively,}\\
 3^{\omega(P_a)+\omega(P_b)}/2 &\text{ if $(B_i,S_i)$  overlap with $(P_a,P_b)$, and both do so negatively,}\\
-3^{\omega(P_a)+\omega(P_b)}/2 &\text{ if $(B_i,S_i)$   overlap with $(P_a,P_b)$, one positively, the other negatively}.
 \end{cases}.
\end{eqnarray}

We will now show that:
\begin{eqnarray}\label{EXab}
\mathbb{E}(X_{a,b})=2^{-n}\Tr{(P_a\phi_t(P_b))}.
\end{eqnarray}
and 
\begin{align}\label{VarXab}
    \mathbb{E}(X_{a,b}^2)\leq 3^{w(P_a)+w(P_b)}
\end{align}

Before we prove that, let us discuss how these estimates on moments on $X_{a,b}$ suffice to obtain the claimed sample complexity. As in other works on classical shadows, it will be crucial to use the method of median of means estimator~\cite{Devroye2016} to estimate the expectation value of $X_{a,b}$. The method of medians of means works as follows. We take a sample of size $S$ and divide it into $K$ subsets of size $B$, i.e. $S=KB$. We then compute the empirical mean on each of the $K$ subsamples. Denote them by $\hat{\mu}_i$, with $1\leq i\leq K$. We then set our estimator of the mean to be:
\begin{align}
    \hat{\mu}_{\operatorname{MoM}}=\operatorname{median}(\hat{\mu}_1,\ldots,\hat{\mu}_K).
\end{align}
The main proeprty of this estimator is that we have that if the variance of $X_{a,b}$ is $\sigma^2$, then:
\begin{align}
\mathbb{P}(|\hat{\mu}_{\operatorname{MoM}}-\mathbb{E}(X_{a,b})|\geq\epsilon)\leq e^{-2K\left(\frac{1}{2}-\frac{\sigma^2}{B\epsilon^2}\right)}.
\end{align}
In particular, if we pick $B=4\sigma^2\epsilon^{-2}$ and $K=2\log(\delta^{-1})$, then:
\begin{align}
\mathbb{P}(|\hat{\mu}_{\operatorname{MoM}}-\mathbb{E}(X_{a,b})|\geq\epsilon)\leq \delta.
\end{align}
We see that the median of means method allows to have a logarithmic scaling of the error probability from a bound on the variance. On the other hand, just using the empirical mean directly combined with an estimate on the variance gives a polynomial dependance only.

Armed with these facts about the median of means estimator, the following result  immediately follows from \eqref{EXab} and \eqref{VarXab}:
\begin{cor}\label{equ:sample_complexity_shadows_process}
Let $\{P_a^l\}_{l=1}^{K_1}$ and $\{P_b^j\}_{j=1}^{K_2}$ be two collections of Pauli matrices on $n$ qubits such that
for any $l,j$ the following condition
\begin{eqnarray}\label{omega}\omega(P_a^l)+\omega(P_b^j)\leq \omega,\end{eqnarray}
holds.
Then 
\begin{eqnarray}\label{Oqw}
O(3^{\omega}\log{(K_1K_2\delta^{-1})}\epsilon^{-2})
\end{eqnarray}
runs of the protocol above suffice to obtain an estimate $e_{a,b}^{l,j}$ satisfying
\begin{eqnarray}\label{1204}
\Big|
|2^{-n}\tr{P_b^l\Phi_t(P_a^j)}-e_{a,b}^{l,j}\Big|\leq \epsilon
\end{eqnarray}
for all pairs $l,j$ with probability at least $1-\delta$.
\end{cor}
\begin{proof}
Let $X_{a,b}^{l,j}$ be the random variable we have defined above for a pair of Paulis $P_a^l$ and $P_b^j$. By the bound in \eqref{VarXab} we have that $\mathbb{E}((X_{a,b}^{l,j})^2)\leq 3^{\omega}$, holds. Thus, for the median of means estimator with $B=3^\omega\epsilon^{-2}$ samples per group and $K=2\log(\delta^{-1}K_1K_2)$ we have an estimate satisfying \eqref{1204} with probability of failure at most $\delta/K_1K_2$.
\par By a union bound, the median of means estimator of all $K_1K_2$ combinations of $P_a^l$ and $P_b^j$ satisfy \eqref{1204} with probability of failure at
most $\delta$.  The total number of samples required for this is
\begin{align}
S=KB=O(3^{\omega}\log{(K_1K_2\delta^{-1})}\epsilon^{-2}),
\end{align}
which yields the claim.
\end{proof}
\par To conclude we only need to show that \eqref{EXab} and \eqref{VarXab} hold. Let us start with the expectation value.
 First, note that
\begin{eqnarray}\label{1359}
\mathbb{E}(X_{a,b})=\mathbb{E}(X_{a,b}|B \,\,\mbox{and}\,\, S\,\, \mbox{overlap}) \mathbb{P}(B\,\, \text{and}\,\,  S\,\, \text{overlap}),
\end{eqnarray}
holds. Here we say that $B$ and $S$ overlap if $B$ overlaps with $P_a$ and $S$ overlaps with $P_b$. The latter formula holds, because if we do not have overlap then $X_{a,b}=0$ by the case $1$ of the definition \eqref{Xab}. Now observe that 
\begin{eqnarray}\label{1451}\mathbb{P}(B\,\, \text{and}\,\,  S\,\,  \text{overlap})=3^{-\omega(P_a)}3^{-\omega(P_b)}.\end{eqnarray}
This holds because we have a $1/3$ chance of "hitting the right Pauli" at each point of the support of either $P_a$ or $P_b$ and they are all independent.
Thus, all we need is to determine 
\begin{eqnarray}\label{1458}\mathbb{E}(X_{a,b}|B\,\,  \mbox{and}\,\,  S\,\,  \text{overlap}).\end{eqnarray}
Here we distinguish four cases:
\begin{enumerate}
    \item $S$ overlaps positively with $P_b$ and $B$ overlaps positively with $P_a$.
     \item $S$ overlaps negatively with $P_b$ and $B$ overlaps negatively with $P_a$,
     \item $S$ overlaps positively with $P_b$ and $B$ overlaps negatively with $P_a$,
     \item $S$ overlaps negatively with $P_b$ and $B$ overlaps positively with $P_a$.
\end{enumerate}
We can then break down the expectation of \eqref{1359} down into the cases:
\begin{eqnarray}\label{1506}
\mathbb{E}(X_{a,b}|B\,\, \mbox{and}\,\,  S\,\, \text{overlap}) \mathbb{P}(B\,\,  \text{and}\,\,  S\,\, \text{overlap})=
\sum\limits_{i=1}^{4}\mathbb{E}(X_{a,b}|B\,\,  \mbox{and}\,\, S\,\,  \text{overlap}, c=i)\mathbb{P}(c=i),
\end{eqnarray}
where $c$ is a random variable which keeps track of which case we have. By definition, if $c=1,2$, then 
\begin{eqnarray}\label{1510}
\mathbb{E}(X_{a,b}|B\,\,  \mbox{and}\,\,  S\,\, \text{overlap}, c=1,2) = 3^{\omega(P_a)+\omega(P_b)}/2.
\end{eqnarray}
If $c=3,4$:
\begin{eqnarray}\label{1511}
\mathbb{E}(X_{a,b}|B\,\,\mbox{and}\,\,  S\,\, \text{overlap}, c=3,4) = -3^{\omega(P_a)+\omega(P_b)}/2.
\end{eqnarray}
Note that the $3^{\omega(P_a)+\omega(P_b)}$ counteracts the $ \mathbb{P}(B\,\, \text{and}\,\,  S\,\, \text{overlap})$ term, up to an additional $1/2$ factor.
\par Thus, all that is left is to estimate 
\begin{eqnarray}\label{1514}\mathbb{P}(c=i),\quad i=1,2,3,4.\end{eqnarray}
Let us estimate $\mathbb{P}(c=1)$, the other cases will be analogous.
To this end, we will introduce $Q_a^{+}$ and $Q_a^{-}$, which are the projectors onto the positive and  negative eigenvalues, respectively, of the support of $P_a$. We will use analogous notation for $Q_b^+$ and $Q_b^-$.
Clearly, $P_a=(Q_a^{+}-Q_a^{-})\otimes I^{\otimes n-\omega(P_a)}$, holds.  Here and in what follows we will always assume for simplicity that the Pauli strings are supported on the first qubits.

Let us estimate the expected initial state for the case $c=1$. Conditioned on being a state  with positive overlap
with $P_b$, by constructiton we know that the state is uniformly distributed on the positive eigenspace of $P_b$.
This corresponds to the state
\begin{eqnarray}\label{1526}\sigma= \frac{Q^{+}_b}{2^{\omega(P_b)-1}}\otimes \left(\frac{\mathbb{I}}{2}\right)^{\otimes n-\omega(P_b)}.\end{eqnarray}
Given that this is the initial state and that we are measuring in the eigenbasis of $P_a$ (recall that we have overlap), the probability of measuring 
a positive outcome is $\Tr{(Q_a^{+}\otimes \mathbb{I}^{n-\omega(P_a)}\phi_t(\sigma))}$.
\par We conclude that
\begin{eqnarray}\label{1536} \mathbb{P}(c=1)=\frac{1}{2}\Tr{(Q_a^{+}\otimes \mathbb{I}^{n-\omega(P_a)}\phi_t(\sigma))}=
\frac{1}{2^{n-1}}\Tr{(Q_a^{+}\otimes \mathbb{I}^{n-\omega(P_a)}\phi_t(Q_b^{+}\otimes \mathbb{I}^{n-\omega(P_b)}))}.\end{eqnarray}
Similarly
\begin{eqnarray}\label{1539} \mathbb{P}(c=2)=
\frac{1}{2^{n-1}}\Tr{(Q_a^{-}\otimes 1^{n-\omega(P_a)}\phi_t(Q_b^{-}\otimes 1^{n-\omega(P_b)}))},\end{eqnarray}
\begin{eqnarray}\label{1537} \mathbb{P}(c=3)=
\frac{1}{2^{n-1}}\Tr{(Q_a^{+}\otimes 1^{n-\omega(P_a)}\phi_t(Q_b^{-}\otimes 1^{n-\omega(P_b)}))},\end{eqnarray}
\begin{eqnarray}\label{1537} \mathbb{P}(c=4)=
\frac{1}{2^{n-1}}\Tr{(Q_a^{-}\otimes 1^{n-\omega(P_a)}\phi_t(Q_b^{+}\otimes 1^{n-\omega(P_b)}))},\end{eqnarray}
As 
\begin{eqnarray}\label{1540} 2^{-n}\Tr{(P_a \phi_t(P_b) )}=2^{-n}(\Tr{(Q_a^{+}\phi_t(Q_b^{+}))}+\Tr{(Q_a^{-}\phi_t(Q_b^{-}))}-
\Tr{(Q_a^{+}\phi_t(Q_b^{-}))}-\Tr{(Q_a^{-}\phi_t(Q_b^{+}))}),
\end{eqnarray}
we conclude that
\begin{eqnarray}\label{1511}
\mathbb{E}(X_{a,b})=2^{-n}\Tr{(P_a\phi_t(P_b))}.
\end{eqnarray}

Computing the second moment of $X_{a,b}$ turns out to be quite simple. Note that the only nonzero value the random variable $X_{(a,b)}^2$ takes is $3^{2\omega(P_a)+2\omega(P_b)}/4$ with probability $3^{-\omega(P_a)-\omega(P_b)}$. Thus, we clearly have that:
\begin{align}
\mathbb{E}(X_{(a,b)}^2)=3^{\omega(P_a)+\omega(P_b)}/4,
\end{align}
which yields \eqref{VarXab}. This concludes all the computations required for Cor.~\eqref{equ:sample_complexity_shadows_process}.

Note that the protocol only requires the preparation of Pauli states and Pauli measurements. Thus, it should be feasible to implement it on near term devices. Additionally, the required postprocessing is efficient, as evaluating the value of $X_{a,b}$ can be efficiently given a sample.

Furthermore, note that if we further have the information that we do not to wish to recover certain bases (i.e. we do not wish to recover Pauli strings with $Y$ terms), it is possible to adapt the protocol and do not prepare initial states or measure in that basis. This will reduce the sample complexity accordingly.

\end{document}